\documentclass{article}
\usepackage[a4paper]{geometry}
\usepackage{amssymb,amsmath,amsthm}
\usepackage{graphics,graphicx,color}
\usepackage{microtype,verbatim,dsfont}
\usepackage{xspace,enumerate}
\usepackage{subcaption}
\usepackage[pdfpagelabels,colorlinks,citecolor=blue,linkcolor=blue,urlcolor=blue]{hyperref}
\usepackage[capitalise]{cleveref}

\theoremstyle{plain}
\newtheorem{theorem}{Theorem}
\newtheorem{definition}{Definition}
\newtheorem{lemma}[theorem]{Lemma}
\newtheorem{corollary}[theorem]{Corollary}

\newtheorem{property}[theorem]{Property}

\Crefname{observation}{Observation}{Observations}
\Crefname{algorithm}{Algorithm}{Algorithms}
\Crefname{algocf}{Algorithm}{Algorithms}
\Crefname{section}{Section}{Sections}
\Crefname{lemma}{Lemma}{Lemmata}
\Crefname{claim}{Claim}{Claims}
\Crefname{property}{Property}{Properties}
\Crefname{enumi}{Condition}{Conditions}
\Crefname{figure}{Fig.}{Figs.}

\newcommand{\skeleton}{planar skeleton\xspace}
\newcommand{\map}[1]{${#1}$-map\xspace}
\newcommand{\hfmap}[1]{hole-free \map{#1}}

\newcommand{\framed}[1]{\mbox{${#1}$-framed}\xspace}

\newcommand{\pframed}[1]{partial \framed{#1}}

\title{Book Embeddings of Nonplanar Graphs\\with Small Faces in Few Pages}

\author{Michael A. Bekos$^1$, Giordano Da~Lozzo$^2$, Svenja Griesbach$^3$,\\Martin Gronemann$^3$, Fabrizio Montecchiani$^4$, Chrysanthi Raftopoulou$^5$
\\
\medskip
\\
\small$^1$Wilhelm-Schickhard-Institut f\"ur Informatik, Universit\"at T\"ubingen, Germany\\
\small\texttt{bekos@informatik.uni-tuebingen.de}
\\
\small$^2$Department of Engineering, Roma Tre University, Italy\\
\small\texttt{giordano.dalozzo@uniroma3.it}
\\
\small$^3$Department of Mathematics and Computer Science, University of Cologne, Germany\\
\small\texttt{sgriesba@smail.uni-koeln.de, gronemann@informatik.uni-koeln.de}
\\
\small$^4$Department of Engineering, Universit\'a degli Studi di Perugia, Italy\\
\small\texttt{fabrizio.montecchiani@unipg.it}
\\
\small$^5$School of Applied Mathematical \& Physical Sciences, NTUA, Greece\\
\small\texttt{crisraft@mail.ntua.gr}
}
\date{}

\begin{document}

\maketitle

\begin{abstract}
An embedding of a graph in a book, called \emph{book embedding}, consists of a linear ordering of its vertices along the spine of the book and an assignment of its edges to the pages of the book, so that no two edges on the same page cross.
The \emph{book thickness} of a graph is the minimum number of pages over all its book embeddings. For planar graphs, a fundamental result is due to Yannakakis, who proposed an algorithm to compute embeddings of planar graphs in books with four pages. Our main contribution is a technique that generalizes this result to a much wider family of nonplanar graphs, which is characterized by a biconnected skeleton of crossing-free edges whose faces have bounded degree. Notably, this family includes all $1$-planar, all optimal $2$-planar, and all \map{k} (with bounded $k$) graphs as subgraphs. We prove that this family of graphs has bounded book thickness, and as a corollary, we obtain the first constant upper bound for the book thickness of optimal $2$-planar and \map{k} graphs.
\end{abstract}

\section{Introduction}
\label{sec:introduction}

Book embeddings of graphs form a well-known topic in topological graph theory that has been a fruitful subject of intense research over the years, with seminal results dating back to the 70s~\cite{Oll73}. 
In a \emph{book embedding} of a graph $G$, the vertices of $G$ are restricted to a line, called the~\emph{spine} of the book, and the edges of $G$ are assigned to different half-planes delimited by the spine, called \emph{pages} of the book. From a combinatorial point of view, computing a book embedding of a graph corresponds to finding a \emph{linear ordering} of its vertices and a partition of its edges, such that no two edges in the same part cross; see \cref{fig:intro}. The \emph{book thickness} (also known as  \emph{stack~number} or \emph{page~number}) of a graph is the minimum number of pages required by any of its book embeddings, while the \emph{book thickness} of a family of graphs $\mathcal G$ is the maximum book thickness of any graph $G$ that belongs to $\mathcal G$.

\begin{figure}[t]
    \centering
    \begin{subfigure}{.22\textwidth}
	\centering
	\includegraphics[width=\textwidth,page=1]{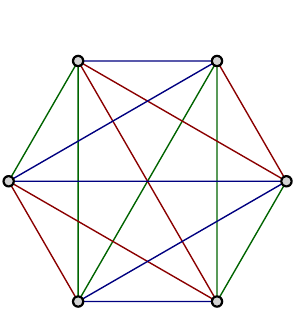}
	\end{subfigure}
	\hfil
	\begin{subfigure}{.22\textwidth}
	\centering
	\includegraphics[width=\textwidth,page=2]{figures/intro}
	\end{subfigure}
    \caption{Graph $K_6$ and a book embedding of it with the minimum of three pages.}
    \label{fig:intro}
\end{figure}

Book embeddings were originally motivated by the design of VLSI circuits~\cite{CLR87,DBLP:journals/tc/Rosenberg83}, but they also find applications, among others, in sorting permutations~\cite{DBLP:conf/stoc/Pratt73,DBLP:journals/jacm/Tarjan72}, compact graph encodings~\cite{DBLP:conf/focs/Jacobson89,DBLP:journals/siamcomp/MunroR01}, graph drawing~\cite{DBLP:journals/jgaa/BiedlSWW99,DBLP:conf/compgeom/BinucciLGDMP19,DBLP:journals/jal/Wood02}, and computational origami~\cite{DBLP:conf/gd/AkitayaDHL17}; for a more complete list, we point the reader to~\cite{DBLP:journals/jct/DujmovicMW17}. Unfortunately, determining the book thickness of a graph turns out to be an NP-complete problem even for maximal planar graphs~\cite{Wig82}. This negative result has motivated a large body of research devoted to the study of upper bounds on the book thickness of meaningful graph families. 

In this direction, there is a very rich literature concerning planar graphs. The most notable result is due to Yannakakis, who back in 1986 exploited a peeling-into-levels technique (a flavor of it is given in \cref{sec:planaralgorithm}) to prove that the book thickness of any planar graph is at most~$4$~\cite{DBLP:conf/stoc/Yannakakis86,DBLP:journals/jcss/Yannakakis89}, improving uppon a series of previous results~\cite{DBLP:conf/stoc/BussS84,DBLP:conf/focs/Heath84,Istrail1988a}. Even though it is not yet known whether the book thickness of planar graphs is $3$ or $4$, there exist several improved bounds for particular subfamilies of planar graphs. 

Bernhart and Kainen~\cite{DBLP:journals/jct/BernhartK79} showed that the book thickness of a graph $G$ is $1$ if and only if $G$ is outerplanar, while its book thickness is at most $2$ if and only if $G$ is \emph{subhamiltonian}, that is, $G$ is a subgraph of a Hamiltonian planar graph. In particular, several subfamilies of planar graphs are known to be subhamiltonian, e.g., 4-connected planar graphs~\cite{NC08}, planar graphs without separating triangles~\cite{DBLP:journals/appml/KainenO07}, Halin graphs~\cite{DBLP:journals/mp/CornuejolsNP83}, series-parallel graphs~\cite{DBLP:conf/cocoon/RengarajanM95}, bipartite planar graphs~\cite{DBLP:journals/dcg/FraysseixMP95}, planar graphs of maximum degree~4~\cite{DBLP:journals/algorithmica/BekosGR16}, triconnected planar graphs of maximum degree~5~\cite{DBLP:conf/esa/0001K19}, and maximal planar graphs of maximum degree~6~\cite{Ewald1973}. In this plethora of results, we should also mention that planar $3$-trees have book thickness $3$~\cite{DBLP:conf/focs/Heath84} and that general (i.e., not necessarily triconnected) planar graphs of maximum degree~$5$ have book thickness~at~most~$3$~\cite{DBLP:journals/corr/GuanY2018}.

In contrast to the planar case, there exist far fewer results for non-planar graphs. Bernhart and Kainen first observed that the book thickness of a graph can be linear in the number of its vertices; for instance, the book thickness of the complete graph $K_n$ is $\lceil n/2 \rceil$~\cite{DBLP:journals/jct/BernhartK79}. Improved bounds are usually obtained by meta-theorems exploiting standard parameters of the graph. In particular, Malitz proved that if a graph has $m$ edges, then its book thickness is $O(\sqrt{m})$~\cite{DBLP:journals/jal/Malitz94}, while if its genus is $g$, then its book thickness is $O(\sqrt{g})$~\cite{DBLP:journals/jal/Malitz94a}. Also, Dujmovic and Wood~\cite{DBLP:journals/dcg/DujmovicW07} showed that if a graph has treewidth $w$, then its book thickness is at most $w+1$, improving an earlier linear bound by Ganley and Heath~\cite{DBLP:journals/dam/GanleyH01}. It is also known that all graphs belonging to a minor-closed family have bounded book thickness~\cite{Bla03}, while the other direction is not necessarily true. As a matter of fact, the family of 1-planar graphs  is not closed under taking minors~\cite{DBLP:books/daglib/0030491}, but it has bounded book thickness~\cite{DBLP:journals/corr/AlamBK15,DBLP:journals/algorithmica/BekosBKR17}. We recall that a graph is \emph{$h$-planar} (with $h \geq 0$), if it can be drawn in the plane such that each edge is crossed at most $h$ times; the reader is referred, e.g., to~\cite{DBLP:journals/csur/DidimoLM19,DBLP:journals/csr/KobourovLM17} for recent surveys. 

Notably, the approaches presented in~\cite{DBLP:journals/corr/AlamBK15,DBLP:journals/algorithmica/BekosBKR17} form the first non-trivial extensions of the above mentioned peeling-into-levels technique by Yannakakis~\cite{DBLP:conf/stoc/Yannakakis86,DBLP:journals/jcss/Yannakakis89} to graphs that are not planar. Both approaches exploit an important property of 3-connected 1-planar graphs, namely, they can be augmented and drawn so that all pairs of crossing edges are ``caged'' in the interior of degree-4 faces of a \emph{\skeleton}, which is defined as the graph consisting of all vertices and of all crossing-free edges of the drawing~\cite{R65}. A similar property also holds for the optimal $2$-planar graphs. Namely, each graph in this family admits a drawing whose \skeleton is simple, biconnected, and has only degree~$5$ faces, each containing five crossing edges~\cite{DBLP:conf/compgeom/Bekos0R17}. The book thickness of these graphs, however, has not been studied yet; the best-known upper bound of $O(\log{n})$ is derived from the corresponding one for general $h$-planar graphs~\cite{DBLP:journals/jgaa/DujmovicF18}.

\smallskip\noindent\textbf{Our contribution.} We present a technique that further generalizes the result by Yannakakis to a much wider family of non-planar graphs, called \pframed{k} graphs, which is general enough to include all $1$-planar graphs and all optimal $2$-planar graphs. A graph is \emph{\framed{k}}, if it admits a drawing having a simple biconnected \skeleton, whose faces have degree at most $k \geq 3$, and whose crossing edges are in the interiors of these faces. A  \emph{\pframed{k}} graph is a subgraph of a \framed{k} graph. 
Clearly, the book thickness of \pframed{k} graphs is lower bounded by $\lceil k/2 \rceil$, as they may contain cliques of size $k$~\cite{DBLP:journals/jct/BernhartK79}. In this work, we present an upper bound on the book thickness of \pframed{k} graphs that depends linearly {\em only} on $k$ (but not on $n$). Our main result is~as~follows.

\begin{theorem}\label{th:main}
The book thickness of a \pframed{k} graph is at most $6\lceil \frac{k}{2} \rceil + 5$.
\end{theorem}

Note that the \pframed{3} graphs are exactly the (simple) planar graphs. Also, it is known that \mbox{$3$-connected} $1$-planar graphs are \pframed{4}~\cite{DBLP:conf/gd/AlamBK13}, while general \mbox{$1$-planar} graphs can be augmented to \framed{8}. 
In fact, every two crossing edges can be caged inside a cycle of length (at most) $8$ passing through the endpoints of such crossing edges; the faces of the resulting \skeleton that do not contain any crossing edge can be triangulated. Hence, \cref{th:main} implies constant upper bounds for the book thickness of these families of graphs. Since  optimal $2$-planar graphs are \framed{5}, the next corollary guarantees the first constant upper bound on the book thickness of this~family.

\begin{corollary}\label{cor:optimal-2-planar}
The book thickness of an optimal $2$-planar graph is at most $23$. 
\end{corollary}

\noindent More in general, each \pframed{k} graph  is $h$-planar for $h=(\frac{k-2}{2})^2$, and hence for this family of $h$-planar graphs we prove that the book thickness is $O(\sqrt{h})$, while the best-known upper bound for general $h$-planar graphs is $O(h\log{n})$~\cite{DBLP:journals/jgaa/DujmovicF18}.

\smallskip

Besides $h$-planar graphs, another well-known generalization of planarity are the \emph{map graphs}, introduced by Chen, Grigni, and Papadimitriou~\cite{DBLP:journals/jacm/ChenGP02}. Roughly speaking, a \emph{\map{k} graph} is one whose vertices are in correspondence with a set of regions in the sphere (possibly not covering its entire surface) and whose edges correspond to boundary intersections between pair of regions such that at most $k$ regions meet at the same point (see \cref{sec:map} for a formal definition). As map graphs find applications in graph drawing, circuit board design and topological inference problems~\cite{DBLP:conf/soda/ChenHK99}, they have been extensively studied in the literature, in particular in terms of characterization and recognition~\cite{DBLP:journals/dam/Brandenburg19,DBLP:journals/algorithmica/Brandenburg19,DBLP:journals/jacm/ChenGP02,DBLP:journals/algorithmica/ChenGP06,DBLP:journals/disopt/MnichRS18,DBLP:conf/focs/Thorup98}. For instance, it is known that planar graphs are the \map{2} graphs, and that the \map{4} graphs are exactly those $1$-planar graphs that have a $1$-planar drawing $\Gamma$ such that each pair of crossing edges is caged in a face of the \skeleton of $\Gamma$~\cite{DBLP:journals/algorithmica/Brandenburg19}.  We prove that every \map{k} graph is a \pframed{2k} graph (\cref{th:map}), which together with~\cref{th:main} yield the first nontrivial upper bound for the book thickness of $k$-map graphs. On the other hand, one can easily show that every \pframed{k} graph is a subgraph of a \map{k} graph.

\begin{corollary}\label{co:map}
The book thickness of a \map{k} graph is at most $6k + 5$.
\end{corollary}

\smallskip\noindent\textbf{Paper organization.} In \cref{se:preliminaries}, we give basic definitions and notation. \cref{sec:planaralgorithm} is devoted to the proof of \cref{th:main}: We start by recalling the peeling-into-level decomposition, and we proceed with an inductive proof based on the resulting leveling of the graph. The base case is described in \cref{ssec:two-levels} and corresponds to graphs consisting of two levels only, while the inductive case is described in \cref{ssec:multiple-levels} and deals with general (i.e., multi-level) graphs. The proof of \cref{th:map} is given in \cref{sec:map}. Finally, \cref{sec:open} contains conclusions and open problems that stem from our research.

\section{Preliminaries}\label{se:preliminaries}


\smallskip\noindent{\textbf{Drawings and planar embeddings.}} A graph is \emph{simple}, if it contains neither self-loops nor parallel edges. A \emph{drawing} of a graph $G$ is a mapping of the vertices of $G$ to distinct points of the plane, and of the edges of $G$ to Jordan arcs connecting their corresponding endpoints. A drawing is \emph{planar}, if no two edges intersect, except possibly at a common endpoint. A graph is \emph{planar}, if it admits a planar drawing. A planar drawing partitions the plane into topologically connected regions, called \emph{faces}. The infinite region is called the \emph{unbounded face}; any other face is a \emph{bounded face}. The \emph{degree} of a face is the number of occurrences of its edges encountered in a clockwise traversal of its boundary (counted with multiplicity). Note that if $G$ is biconnected, then each of its faces is bounded by a simple cycle. A \emph{planar embedding} of a planar graph is an equivalence class of topologically-equivalent (i.e., isotopic) planar drawings. A planar graph with a given planar embedding is a \emph{plane graph}.
	
\begin{figure}[t]
    \centering
	\includegraphics[width=.3\textwidth,page=1]{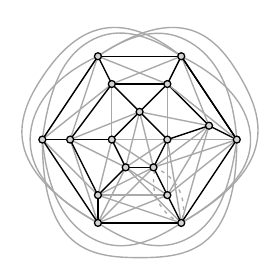}
    \caption{%
    A drawing of a \framed{6} graph, whose crossing-free (crossing) edges are black (gray).}
    \label{fig:kframe}
\end{figure}

\smallskip\noindent{\textbf{\framed{\mathbf{k}} graphs.}}
Let $\Gamma$ be a drawing of a graph $G$. The \emph{\skeleton} $\sigma(G)$ of $G$ in $\Gamma$ is the plane subgraph of $G$ induced by the crossing-free edges of $G$ in $\Gamma$ (where the embedding of $\sigma(G)$ is the one induced by $\Gamma$). The edges of $\sigma(G)$ are called \emph{crossing-free}, while the edges that belong to $G$ but not to $\sigma(G)$ are \emph{crossing} edges. A \emph{\framed{k} drawing} of a graph is one such that its crossing-free edges determine a \skeleton, which is simple, biconnected, spans all the vertices, and has faces of degree at most $k \geq 3$. A graph is \emph{\framed{k}}, if it admits a \framed{k} drawing; refer to \cref{fig:kframe}. A \emph{\pframed{k}} graph is a subgraph of a \framed{k} graph. 
Clearly, if a \framed{k} graph has book thickness at most $b$, then the book thickness of any of its subgraphs is at most $b$. Thus, in the remainder of the paper, we will only consider \framed{k} graphs. Further, w.l.o.g., we will also assume that each pair of vertices that belongs to a face $f$ of $\sigma(G)$ is connected either by a crossing-free edge (on the boundary of $f$) or by a crossing edge (drawn inside $f$). In other words, the vertices on the boundary of $f$ induce a clique of size at most $k$. Under this assumption, graph $G$ may  contain parallel crossing edges connecting the same pair of vertices, but drawn in the interior of different faces of $\sigma(G)$;  see, e.g., the dashed edges of~\cref{fig:kframe}.

\smallskip\noindent{\textbf{Book embeddings.}} A \emph{book embedding} of a graph $G$ consists of a linear ordering $\prec$ of the vertices of $G$ along a line, called the \emph{spine} of the book, and an assignment of the edges of $G$ to different half-planes delimited by the spine, called \emph{pages} of the book, such that no two edges of the same page \emph{cross}, that is, no two edges $(u,v)$ and $(w,z)$ of the same page with $u \prec v$ and $v \prec w$ are such that $u \prec w \prec v \prec z$. We further say that $(u,v)$ and $(w,z)$ of the same page with $u \prec v$ and $v \prec w$ \emph{nest}, if $u \prec w \prec z \prec v$. The \emph{book thickness} of $G$ is the minimum integer $k$, such that $G$ has a book embedding on $k$ pages.

\section{Proof of Theorem 1}
\label{sec:planaralgorithm}

Our approach adopts some ideas from the seminal work by Yannakakis on book embeddings of planar graphs.
In particular, we refer to the algorithm which embeds any (internally-triangulated) plane graph in a book with five pages~\cite{DBLP:journals/jcss/Yannakakis89}, not four.
The main challenges of our generalization are posed by the crossing edges and by the fact that we cannot augment the input graph so that its underlying \skeleton is internally-triangulated.
In the following, we explain the basic ideas of Yannakakis' algorithm and recall basic definitions and properties from~\cite{DBLP:journals/jcss/Yannakakis89}, which we generalize and exploit to introduce new ones.

Our technique is based on the so-called \emph{peeling-into-levels} decomposition. Let $G$ be an $n$-vertex \framed{k} graph with a \framed{k} drawing $\Gamma$. We classify the vertices of $G$ as follows: 
$(i)$~vertices on the unbounded face of $\sigma(G)$ are at level $0$, and
$(ii)$~vertices that are on the unbounded face of the subgraph of $\sigma(G)$ obtained by deleting all vertices of levels $\leq i-1$ are at level $i$ ($0 < i < n$); see, e.g., \cref{fig:level-partition}.
Denote by $\sigma_i(G)$ the subgraph of $\sigma(G)$ induced by the vertices of $L_i$. Observe that $\sigma_i(G)$ is outerplane, but not necessarily connected.
Next, we consider $\sigma_i(G)$ and delete any edge that is not incident to the unbounded face. The resulting spanning subgraph of $\sigma_i(G)$ is denoted by $C_i(G)$.
By definition, each connected component of $C_i(G)$ is a cactus. Also, the only edges that belong to $\sigma_i(G)$ but not to $C_i(G)$ are the chords of $\sigma_i(G)$.
Finally, we denote by $G_i$ the subgraph of $G$ induced by the vertices of $L_0 \cup \ldots \cup L_i$ containing neither chords of $\sigma_i(G)$ nor the crossing edges that are in the interior of the unbounded face of $\sigma(G)$.

\begin{figure}[t!]
    \centering
	\includegraphics[width=.3\textwidth,page=1]{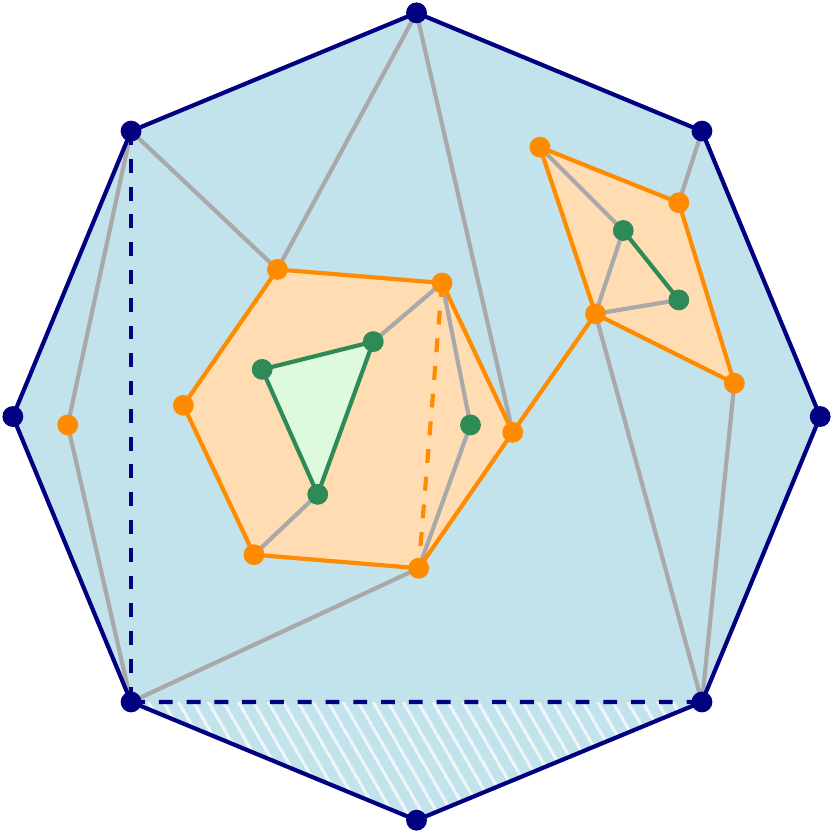}
    \caption{%
    The peeling-into-levels decomposition of an \framed{8} graph without its crossing edges.
    The vertices and level-edges of level $L_0$ ($L_1; L_2$, resp.) are blue (orange; green, resp.) and induce $\sigma_0(G)$ ($\sigma_1(G); \sigma_2(G)$, resp.).
    Chords are drawn dashed; binding edges are drawn gray.
    The blue (orange; green, resp.) faces are the intra-level faces of $\sigma_1(G)$ ($\sigma_2(G); \sigma_3(G)$, resp.).
    Graph $\sigma_0(G)$ ($\sigma_1(G); \sigma_2(G)$, resp.) without the dashed chords forms $C_0(G)$ ($C_1(G); C_2(G)$, resp.).
    The striped blue face is an intra-level face of $\sigma_1(G)$, whose boundary exists exclusively of $L_0$-level edges.    
    }
    \label{fig:level-partition}
\end{figure}

Consider an edge $e$ that belongs to $\sigma(G)$.
If the endpoints of $e$ are assigned to the same level, $e$ is a \emph{level} edge; otherwise, $e$ connects vertices of consecutive levels and is called a \emph{binding} edge; see \cref{fig:level-partition}.
By the definition of the level-partition, there is no edge $e \in E$, that connects two vertices of levels $i$ and $j$, such that $\vert i-j\vert > 1$. 
Another consequence of the level-partition is that any vertex of level $i+1$ lies in the interior of a cycle of level $i$. 
Next, we give a characterization for bounded faces of $\sigma(G)$. 
A bounded face of $\sigma(G)$ is an \emph{intra-level face} of $\sigma_i(G)$ if it is incident to at least one vertex of $L_{i-1}$ but to no vertex of $L_{i-2}$.
We denote by $\mathcal{F}_i$ the set of all the intra-level faces of $\sigma_i(G)$.
By definition, the unbounded face of $\sigma_i(G)$ is not an intra-level face.
Also, each intra-level face of $\sigma_i(G)$ has either at least one binding edge between $L_{i-1}$ and $L_i$ on its boundary, or it consists exclusively of edges of level $L_{i-1}$.

\smallskip\noindent\textbf{Overview.} We give an short overview of how our algorithm embeds a \framed{k} graph $G$ with a given \framed{k} drawing $\Gamma$ on $6\cdot \left\lceil \frac{k}{2} \right\rceil +5$ pages.
In a high level description, we will inductively compute a book embedding of $G_{i+1}$, assuming that we have already computed a book embedding of $G_i$. For this inductive strategy to work, the computed book embeddings satisfy particular invariants, which we  define subsequently. We first focus on the base case, in which $G$ consists of only two levels $L_0$ and $L_1$ under some additional assumptions (see \cref{ssec:two-levels}). Afterwards, we consider the inductive case, in which $G$ consists of more than two levels (see \cref{ssec:multiple-levels}). 

\subsection{Base case: Two-level instances}
\label{ssec:two-levels}

A \emph{two-level instance} is a \framed{k} graph $G$ consisting of two levels $L_0$ and $L_1$, such that there is no crossing edge in the unbounded face of $\sigma_0(G)$, and either $L_1 = \emptyset$ or $\sigma_1(G)=C_1(G)$, i.e., $\sigma_1(G)$ is chord-less; refer to \cref{fig:example} for an illustration of a two-level instance. Since $\sigma(G)$ is biconnected, $C_0(G)$ is a simple cycle. Let $u_0,u_1,\ldots,u_{s-1}$ with $s \ge 3$ be the vertices of $L_0$ in the order that they appear in a clockwise traversal of $C_0(G)$ starting from~$u_0$. An edge $(u_i,u_j)$ of $\sigma_0(G)$ is  \emph{short} if $i-j=\pm 1$; otherwise it is \emph{long}. By definition, $(u_0,u_{s-1})$ is long.
In the following, we will refer to the intra-level faces of $\sigma_1(G)$ simply as intra-level faces, and we will further denote $\mathcal{F}_1$ as $\mathcal{F}$. Consider now the graph $C_1(G)$. Each of its connected components is a cactus; thus, its biconnected components, called \emph{blocks}, are either single edges or simple cycles (that are chordless, as $\sigma_1(G)=C_1(G)$). A connected component of $C_1(G)$ may degenerate into a single vertex, and this vertex itself is a \emph{degenerate} block. A block that consists of more than one vertex is called \emph{non-degenerate}. 

We equip $\mathcal{F}$ with a linear ordering $\lambda(\mathcal{F})$ as follows. For $i=0,\dots,s-1$, the intra-level faces incident to vertex $u_i$ are appended to $\lambda(\mathcal{F})$ as they appear in counterclockwise order around $u_i$ starting from the one incident to $(u_{i-1},u_i)$ and ending at the one incident to $(u_i,u_{i+1})$ (indices taken modulo $s$), unless already present. For a pair of intra-level faces $f$ and $f'$, we write $f \prec_\lambda f'$ if $f$ precedes $f'$~in~$\lambda(\mathcal{F})$; similarly, we write $f \preceq_\lambda f'$ if $f = f'$ or $f \prec_\lambda f'$.

Let $C_1,\ldots,C_\gamma$ be the connected components of $C_1(G)$ and let $C \in \{C_1,\ldots,C_\gamma\}$. 
In general, several intra-level faces in $\mathcal{F}$ may contain vertices of $C$ on their boundary. Let $f_C$ be the first face in the ordering $\lambda(\mathcal{F})$ that contains a vertex of $C$. Consider now a counterclockwise traversal of the boundary of $f_C$ starting from the vertex of $L_0$ with the smallest subscript that belongs to $f_C$. We refer to the vertex, say $v_C$, of $C$ that is encountered first in this traversal as the \emph{first vertex} of $C$. Observe that, by definition, $v_C$ is incident to a binding edge that is on the boundary of $f_C$. We will further assume that $v_C$ forms a degenerate block $r_C$ of $C$. The \emph{leader} of a block $B$ of $C$, denoted by~$\ell(B)$, is the first vertex of $B$ that is encountered in any path of $C$ from $v_C$ to $B$; note that $\ell(B)$ is uniquely defined.

Consider a vertex $v$ of $C$. If $v$ belongs to only one block of $C$, then $v$ is \emph{assigned} to that block. Otherwise $v$ is assigned to the block $B$ of $C$ such that $v$ belongs to $B$ and the graph-theoretic distance in $C$ between $\ell(B)$ and $v_C$ is the smallest. It follows that $v_C$ is assigned to the degenerate block $r_C$, and that for any non-degenerate block $B$ the leader $\ell(B)$ is not assigned to $B$. We denote by $B(v)$ the block of $C$ that a vertex $v$ is assigned to. 
Let $B$ be a block of $C$. Assume first that $B$ is non-degenerate. We refer to the first face in the ordering $\lambda(\mathcal{F})$ containing an edge of $B$ as the face that \emph{discovers} $B$. Assume now that $B$ is degenerate, i.e., it consists of a single vertex $v$. We refer to the first face in the ordering $\lambda(\mathcal{F})$ that has $v$ on its boundary as the face that discovers $B$. In both cases, we denote by $d(B)$ the face in $\mathcal{F}$ that discovers block~$B$.

We extend the notion of discovery to the vertices of $G$.
To this end, let $v$ be a vertex of $G$ (which can be incident to several intra-level faces in $\mathcal{F}$).
We distinguish whether $v$ belongs to $L_0$ or $L_1$.
In the former case, face $f$ of $\mathcal{F}$ \emph{discovers} vertex~$v$ if $f$ is the first intra-level face in the ordering $\lambda(\mathcal{F})$ that contains $v$ on its boundary.
In the latter case, face $f$ in $\mathcal{F}$ \emph{discovers} vertex~$v$ if $f$ is the face that discovers the block vertex~$v$ is assigned to.
In both cases we denote by $d(v)$ the face in $\mathcal{F}$ that discovers vertex~$v$.
This yields $d(v)=d(B(v))$ for any $v\in L_1$.
The \emph{dominator} $\text{dom}(B)$ of block $B$ is the vertex of $L_0$ with the smallest subscript that is on the boundary of $d(B)$.
Several blocks of $C$ can be discovered by the same face, and by definition, these blocks have the same dominator.
Analogously, we define the dominator $\text{dom}(f)$ of an intra-level face $f$ as the vertex of $L_0$ with the smallest subscript that is on the boundary of $f$. This yields $\text{dom}(B)=\text{dom}(d(B))$.

\begin{figure}[t]
    \centering
	\includegraphics[width=.68\textwidth,page=1]{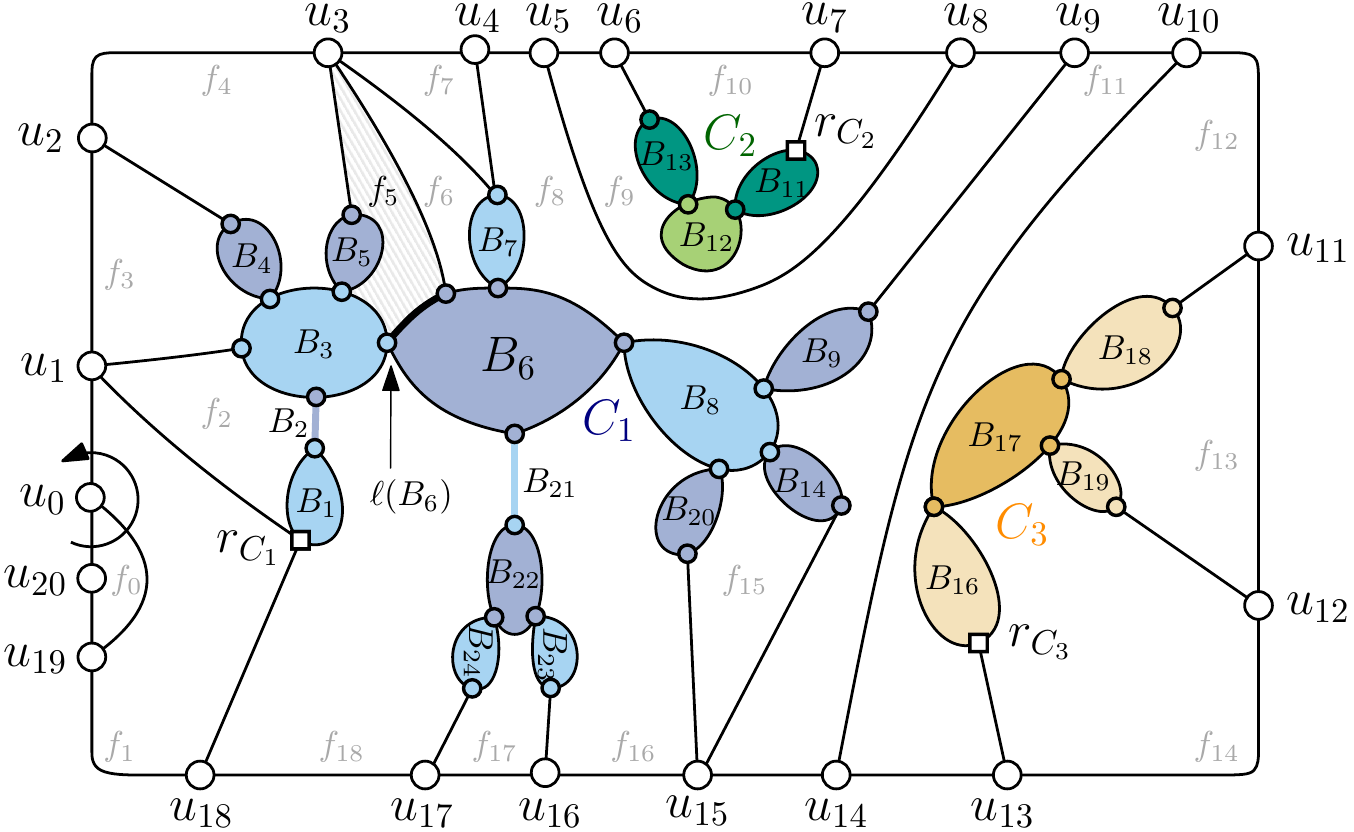}
    \caption{%
    Illustration of the graph $\sigma_1(G)$ of a two-level instance $G$:
	the vertices of $L_0$ are denoted by $u_0,\ldots,u_{20}$;
	the vertices of $L_1$ are the remaining ones;
	$C_1(G)$ consists of three connected components $C_1$, $C_2$ and $C_3$, whose first vertices are denoted by $v_{C_1}$, $v_{C_2}$ and $v_{C_3}$, resp.;
	the vertices assigned to each block have the same color as the block;
	$C_1$ contains two blocks $B_2$ and $B_{21}$ that are simple edges; 
	the two level edges $(u_5,u_6)$ and $(u_5,u_8)$ are short and long, resp.;  
	edge $(u_1,v_{C_1})$ is a binding edge;
	the intra-level faces of $\mathcal{F}$ are all numbered from $f_0$ to $f_{18}$ according to $\lambda(\mathcal{F})$; 
	the intra-level face that discovers $B_6$ is the face $f_5$ tilled gray;
	$f_1$, $f_9$ and $f_{12}$ discover the degenerate blocks.
    }
    \label{fig:example}
\end{figure}

\begin{property}\label{prp:vertex-on-boundary}
The face $d(B)$ that discovers block $B$ is the first face in $\lambda(\mathcal{F})$ that has a vertex assigned to block $B$ on its boundary. 
\end{property}
\begin{proof}
If $B$ is a degenerate block, the property follows by definition. Otherwise, $B$ contains at least one edge on its boundary. The face $d(B)$ is the first intra-level face in $\lambda(\mathcal{F})$ that contains an edge $(v,w)$ of~$B$ on its boundary. Since only the leader $\ell(B)$ of $B$ is not assigned to block $B$ and since $(v,w)$ is a boundary edge of $B$, at least one of $v$ and $w$ is assigned to $B$. The property follows from the fact that at most one of the endpoints of $(v,w)$ is not assigned to $B$.
\end{proof}

Consider now two blocks $B$ and $B'$ of $C_1(G)$. Note that $B$ and $B'$ do not necessarily belong to the same connected component of $C_1(G)$. We say that $B$ \emph{precedes} $B'$ if 
(i)~$d(B) \prec_\lambda d(B')$, or
(ii)~$d(B) = d(B')$ and in a counterclockwise traversal of $d(B)$ starting from $\text{dom}(d(B))$ block $B$ is encountered before block $B'$.
We denote this relationship between $B$ and $B'$ by $B\prec B'$.
Since $\lambda(\mathcal{F})$ is a well-defined ordering, it follows that the relationship ``precedes'' is also defining a total ordering of the blocks of $C_1(G)$. In the following, we introduce a useful property of $\lambda(\mathcal{F})$.

\begin{property}\label{prp:first-the-discoverer}
Let $v$ be a vertex of $G$ and let $f_v \in \mathcal{F}$ be an intra-level face that contains $v$ on its boundary. Then, $d(v) \preceq_\lambda f_v$ holds.
\end{property}
\begin{proof}
If $v$ belongs to $L_0$, then the property follows by definition. Otherwise,  $v$ belongs to $L_1$, and $d(v)$ is the intra-level face that discovers the block $B(v)$, that is, $d(v)=d(B(v))$. If $B(v)$ is degenerate, then $d(v)$ is the first intra-level face in $\lambda(\mathcal{F})$ that has $v$ on its boundary. Hence, $d(v) \preceq_\lambda f_v$. Otherwise, by \cref{prp:vertex-on-boundary}, $d(B(v))$ is the first intra-level face in $\lambda(\mathcal{F})$ that contains a vertex assigned to block $B$ on its boundary. Since $d(v)=d(B(v))$ and since $v$ is assigned to block $B$, it follows that $d(v)\preceq_\lambda f_v$. 
\end{proof}

Next, we introduce the notion of a prime vertex with respect to an intra-level face. We say that a vertex $v$ of $L_0$ belonging to the boundary of an intra-level face $f$ is \emph{prime} with respect to $f$ if no vertex of $L_1$ and no long level edge is encountered in the clockwise traversal of $f$ from $\text{dom}(f)$ to $v$.
By definition, $\text{dom}(f)$ is prime with respect to $f$.
We say that a vertex $v$ is \emph{$f$-prime} if either $v$ is prime with respect to face $f$ or $v$ belongs to $L_1$. By definition, any vertex of $L_1$ is $g$-prime with respect to any intra-level face $g$.
Let $u_j$ be a vertex on $L_0$ that is not $d(u_j)$-prime with $j\in \{1,\ldots,s-1\}$.
Let $f_0^{u_j},\ldots,f_t^{u_j}$ be the faces that have $u_j$ on their boundary in a counterclockwise traversal of $u_j$ starting from $(u_{j-1},u_j)$ and ending at $(u_j,u_{j+1})$ (indices taken modulo $s$). Let $d$ be smallest index such that $f_d^{u_j}=d(u_j)$.
The faces $f_0^{u_j},\ldots,f_{d-1}^{u_j}$ that have $u_j$ as their dominator are called \emph{small}.

\subsubsection{Linear ordering}
\label{sse:linearorder}

The \emph{linear ordering} of the vertices, denoted by $\rho$, is computed as follows. First, the vertices of $L_0$ are embedded in the order $u_0,u_1,\ldots,u_{s-1}$.
The remaining vertices of $G$ (i.e., the vertices of $L_1$) are embedded along the spine based on the blocks that they have been assigned to and according to the following rules: 

\begin{enumerate}
\renewcommand{\labelenumi}{\textbf{\theenumi}}
\renewcommand{\theenumi}{R.\arabic{enumi}}
\item \label{r:1} For $j=0,\ldots,s-1$, let $B^j_0,\ldots,B^j_{t-1}$ be the blocks with $u_j$ as dominator such that the faces that discover them are not small (are small, resp.), and $B^j_i \prec B^j_{i+1}$ for $i=0,1,\ldots,t-2$. The vertices assigned to these blocks are placed right after (before, resp.) $u_j$ in $\rho$. 
\item \label{r:2} The vertices assigned to $B^j_i$ are right before those assigned to $B^j_{i+1}$, for each $i=0,\ldots,t-2$.
\item \label{r:3} The vertices assigned to the same block $B^j_i$ are in the order they appear in a counterclockwise traversal of the boundary of $B^j_i$ starting from the leader of $B^j_i$, for $i=0,\ldots,t-1$.
\end{enumerate}
For a pair of distinct vertices $v$ and $w$, we write $v \prec_\rho w$ if $v$ precedes $w$ in~$\rho$.
By Rule~\ref{r:1}, the vertices of $L_1$ that are discovered by $f$ and the $f$-prime vertices of $L_0$ are right next to each other in $\rho$.  
The next property is consequence of Rules~\ref{r:1}--\ref{r:3}.

\begin{property}\label{prp:consecutively}
The vertices assigned to a block $B$ of $L_1$ appear consecutively in $\rho$.
\end{property}

\noindent The order of the blocks together with Rules~\ref{r:1} and \ref{r:2} yields the following property.

\begin{property}\label{prp:block-vertex-order}
Let $v$ and $w$ be two vertices of $L_1$ assigned to two distinct blocks $B(v)$ and $B(w)$, respectively.
Then, $v\prec_\rho w$ if and only if $B(v)$ precedes $B(w)$.
\end{property}

\noindent The next properties will be useful in Section~\ref{ssec:multiple-levels}.%

\begin{property}\label{prp:ordering-disconnected}
	Let $C_1$ and $C_2$ be two connected components of $C_1(G)$ rooted at their first vertices, and let $B_1$ and $B_2$ be two non-degenerate blocks of $C_1$ and $C_2$, respectively.
	If there exists a vertex $v$ assigned to $B_2$ between $\ell(B_1)$ and the vertices assigned to $B_1$ in $\rho$, then all vertices assigned to $B_2$ appear in $\rho$ between $\ell(B_1)$ and the vertices assigned~to~$B_1$.
\end{property}
\begin{proof}
Let $B'_1$ be the block that $\ell(B_1)$ is assigned to.
Then $B'_1$ is a block of $C_1$ and $B'_1\neq B_1$.
Let $w$ be a vertex assigned to block $B_1$. Then we have $\ell(B_1)\prec_\rho v\prec_\rho w$ with $\ell(B_1)$ assigned to $B'_1$, $v$ assigned to $B_2$, and $w$ assigned to $B_1$.
By \cref{prp:consecutively}, all vertices assigned to the same block are consecutive in $\rho$, and the claim follows.
\end{proof}

\begin{figure}[t]
    \centering
	\includegraphics[width=.5\textwidth,page=1]{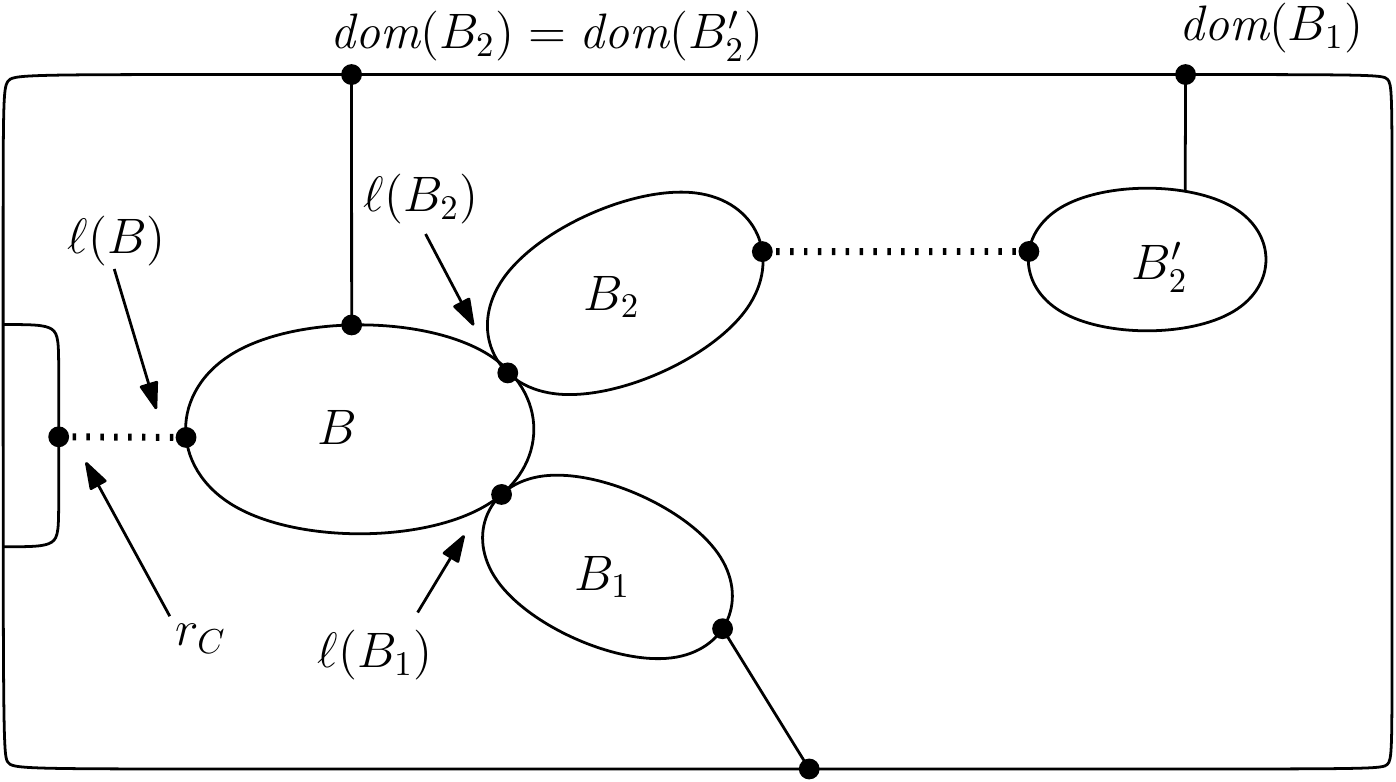}
    \caption{Illustration for the proof of \cref{prp:ordering-descendants}.}
    \label{fig:two-block-children}
\end{figure}

\begin{property}\label{prp:ordering-descendants}	
Let $C$ be a connected component of $C_1(G)$ rooted at its first vertex, and let~$B$ be a non-degenerate block of $C$ with two children $B_1$ and $B_2$.
If $\ell(B_1) \preceq_\rho \ell(B_2)$ and $B_2 \prec B_1$, then 
all vertices assigned to descendant blocks of~$B_2$ (including $B_2$) precede in $\rho$ all vertices assigned to descendant blocks of $B_1$ (including~$B_1$).
\end{property}
\begin{proof}
First, observe that for a block $B$ and any descendant block $B'$ of $B$, we have the order $B\prec B'$. Therefore, any vertex assigned to $B$ precedes any vertex assigned to $B'$ in $\rho$.
Hence, let $B'_2$ be a descendant of $B_2$.
It remains to show that if $B_1$ and $B_2$ are children of the same block, $\ell(B_1) \preceq_\rho \ell(B_2)$, and $B_2 \prec B_1$, then $v\prec_\rho w$ for any vertex $v$ assigned to $B'_2$ and any vertex $w$ assigned to $B_1$.
Since $\sigma(G)$ is planar and biconnected, we get $d(B'_2) \prec_\lambda d(B_1)$; see \cref{fig:two-block-children}.
Hence, $\text{dom}(d(B'_2))\preceq_\rho \text{dom}(d(B_1))$ holds.
Now the claim follows by Rules \ref{r:1} and \ref{r:2}.
\end{proof}

\begin{property}\label{prp:ordering-connected}
Let $C$ be a connected component of $C_1(G)$, and let $B_1$ and $B_2$ be two distinct non-degenerate blocks of $C$.
If there is a vertex $v$ assigned to a block $B_1$ between $\ell(B_2)$ and the remaining vertices of $B_2$ such that $\ell(B_1) \prec_\rho  \ell(B_2)$, then $\ell(B_2)$ is assigned to $B_1$.
\end{property}

\begin{proof}
Assume for a contradiction that $\ell(B_2)$ is assigned to a different block, say $B'_2$.
Let also $B'_1$ be the block that $\ell(B_1)$ is assigned to.
By \cref{prp:block-vertex-order}, we obtain the order of the blocks: $B'_1\preceq B'_2\prec B_1\prec B_2$.
We distinguish two cases based on whether (a) $B'_1\prec B'_2$ or (b) $B'_1=B'_2$ holds.
First, consider Case (a), that is $B'_1\prec B'_2$.
Since $B_1$ is a child of $B'_1$ and $B_1\prec B_2' \prec B_1$, it follows that either $B'_2$ is also a child of $B'_1$ which precedes $B_1$ in the ordering of the blocks, or it is a descendant of another child of $B'_1$ which precedes $B_1$ in the ordering of the blocks.
In both cases, it follows by \cref{prp:ordering-descendants} that $B_2\prec B_1$; a contradiction.
Consider now Case (b). Since $\ell(B_1) \prec_\rho  \ell(B_2)$, and both vertices are assigned to the same block, it follows that $B_2\prec B_1$; a contradiction.
\end{proof}

\begin{property}\label{prp:prime}
Let $v$ be a $d(v)$-prime vertex of $L_0$. Then $v$ is $f$-prime for 
any intra-level face $f$ that has $v$ on its boundary. 
Also, $v=\text{dom}(f)$, except possibly for $f=d(v)$. 
\end{property}
\begin{proof}
Let $f$ be an intra-level face that is different from $d(v)$ such that $f$ has $v$ on its boundary. By planarity, vertex $v$ is the dominator of face $f$. Thus, $v$ is $f$-prime.
\end{proof}

\begin{property}\label{prp:discover-order}
Let $w$ be a $d(w)$-prime vertex.
For any vertex $v$ with $v \prec_\rho w$, $d(v) \preceq_\lambda d(w)$.
\end{property}
\begin{proof}
Since $w$ is $d(w)$-prime, $w$ precedes any vertex discovered by a face $f$ with $d(w)\prec_\lambda f$. 
Assuming to the contrary that $d(w) \prec_\lambda d(v)$, we get $w \prec_\rho v$; a contradiction.
\end{proof}

\noindent By contraposition the following corollary is a direct consequence of \cref{prp:discover-order}.

\begin{corollary}\label{cor:vertex-to-discoverer}
Let $v$ be a $d(v)$-prime vertex. For any vertex $w$, $d(v) \prec_\lambda d(w)$~implies~$v \prec_\rho w$.
\end{corollary}

\subsubsection{Edge-to-Page Assignment}
\label{sse:assignment}

With the linear ordering $\rho$ at hand, we now describe how to perform the edge-to-page assignment which concludes the construction of our book embedding.
We start with some particular types of edges defined as follows. 
An edge $(v,w)$ is a \emph{dominator} edge if $v$ is the dominator of an intra-level face $f_{w}$ containing $w$ on its boundary.
A dominator edge $(v,w)$ is \emph{backward} if $v \prec_\rho w$ or \emph{forward} otherwise.
In the following lemma, we prove that all backward edges of $G$ can be assigned to a single page. 
We note that the proof is reminiscent of a corresponding one by 
Yannakakis~\cite{DBLP:journals/jcss/Yannakakis89} for similarly-defined backward edges.

\begin{figure}[tb!]
	\centering
	\begin{subfigure}{.48\textwidth}
		\centering
		\includegraphics[width=.9\textwidth,page=1]{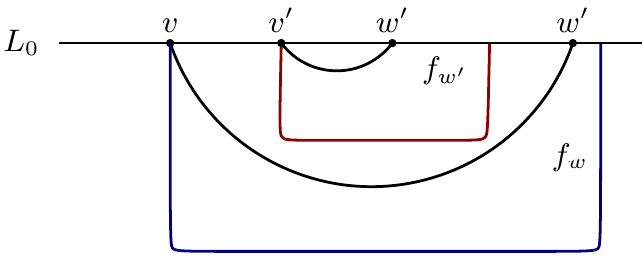}
		\subcaption{}
		\label{fig:backward-1}
	\end{subfigure}
	\begin{subfigure}{.48\textwidth}
		\centering
		\includegraphics[width=.9\textwidth,page=2]{figures/backward}
		\subcaption{}
		\label{fig:backward-2}
	\end{subfigure}
	\caption{%
		Illustration for the proof of \cref{lem:backward}.}
	\label{fig:backward}
\end{figure}

\begin{lemma}\label{lem:backward}
Let $(v,w)$ and $(v',w')$ be two backward edges of $G$, such that $v,w,v'$ and $w'$ are four distinct vertices of $G$ with $v \prec_\rho w$, $v' \prec_\rho w'$ and $v \prec_\rho v'$. Then, $v \prec_\rho w \prec_\rho v' \prec_\rho w'$ or $v \prec_\rho v' \prec_\rho w' \prec_\rho w$ holds.
\end{lemma}
\begin{proof}
By definition, $v$ and $v'$ are the dominators of two intra-level faces $f_w$ and $f_{w'}$ containing $w$ and $w'$ on their boundaries, respectively. Note that if $w\prec_\rho v'$, we have $v \prec_\rho w \prec_\rho v' \prec_\rho w'$. 
Thus, assume $v'\prec_\rho w$. If $w$ belongs to $L_0$, then $w$ is not $f_w$-prime; see \cref{fig:backward-1}. Since $v \prec_\rho v'$, and $v$ and $v'$ are the dominators of $f_w$ and $f_{w'}$, respectively, it follows that $f_w \prec_\lambda f_w'$.
Since vertex $w$ is not $f_w$-prime, we have $w' \prec_\rho w$. Hence, it follows that $v \prec_\rho v' \prec_\rho w' \prec_\rho w$. 
Assume now that $w$ belongs to $L_1$; see \cref{fig:backward-2}. Since $v$ is the dominator of $f_w$, and $v\prec_\rho w$, the vertex $w$ belongs to a block $B(w)$ discovered by $v$. By Rule~\ref{r:1}, there is no vertex of $L_0$ between $v$ and the vertices assigned to $B(w)$ in $\rho$. Hence, $v'$ cannot appear between $v$ and $w$ in $\rho$.  
\end{proof}

\noindent Next, we prove that all forward edges can also be assigned to a single page.

\begin{lemma}\label{lem:forward}
Let $(v,w)$ and $(v',w')$ be two forward edges of $G$, such that $v,w,v'$ and $w'$ are four distinct vertices of $G$ with $w \prec_\rho v$, $w' \prec_\rho v'$ and $v' \prec_\rho v$. Then, $w' \prec_\rho v' \prec_\rho w \prec_\rho v$ or $w \prec_\rho w' \prec_\rho v' \prec_\rho v$ holds.
\end{lemma}

\begin{proof}
By definition, $v$ and $v'$ are the dominators of two intra-level faces $f_w$ and $f_{w'}$ containing $w$ and $w'$ on their boundaries, respectively.
Note that if $v'\prec_\rho w$, then we have $w' \prec_\rho v' \prec_\rho w \prec_\rho v$. 
Thus, assume $w\prec_\rho v'$.
Hence, we have $w\prec_\rho v'\prec_\rho v$ and $w'\prec_\rho v'$, and it remains to show that $w\prec_\rho w'$.
Since $v$ and $v'$ are the dominators of $f_w$ and $f_{w'}$, respectively, and since we know that $w \prec_\rho v$ and $w' \prec_\rho v'$, it follows that $w$ and $w'$ belong to $L_1$ with $d(w)\preceq_\lambda f_w$ and $d(w')\preceq_\lambda f_{w'}$.
Equality holds if $f_w$ or $f_{w'}$ is small. We proceed by distinguishing three cases: (a)~$f_w$ is small, (b)~$f_{w'}$ is small, and (c)~neither $f_w$ nor $f_{w'}$ is not small.

\begin{figure}[tb!]
	\centering
	\begin{subfigure}{.48\textwidth}
		\centering
		\includegraphics[width=.8\textwidth,page=4]{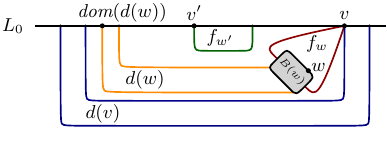}
		\subcaption{}
		\label{fig:forward-4}
	\end{subfigure}		
	\begin{subfigure}{.48\textwidth}
		\centering
		\includegraphics[width=.8\textwidth,page=3]{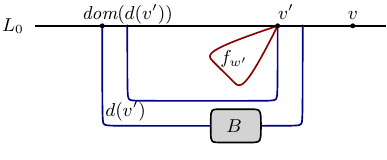}
		\subcaption{}
		\label{fig:forward-5}
	\end{subfigure}
	\begin{subfigure}{.48\textwidth}
		\centering
		\includegraphics[width=.8\textwidth,page=11]{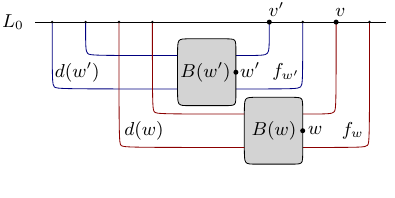}
		\subcaption{}
		\label{fig:forward-1}
	\end{subfigure}	
	\begin{subfigure}{.48\textwidth}
		\centering
		\includegraphics[width=.8\textwidth,page=12]{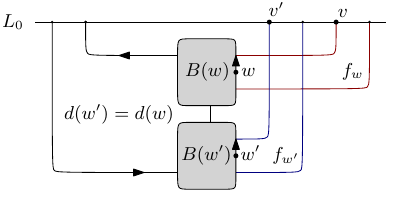}
		\subcaption{}
		\label{fig:forward-2}
	\end{subfigure}		
	\begin{subfigure}{.48\textwidth}
		\centering
		\includegraphics[width=.8\textwidth,page=13]{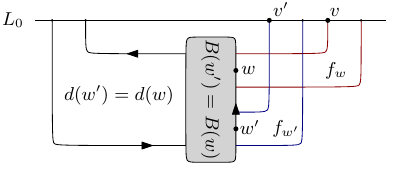}
		\subcaption{}
		\label{fig:forward-3}
	\end{subfigure}	
	\caption{%
		Illustration for the proof of \cref{lem:forward}.}
	\label{fig:forward}
\end{figure}

\begin{itemize}
\item[--] Consider first Case~(a), in which $f_w$ is small.
Since $w\prec_\rho v'$, it follows that $\text{dom}(d(w))\prec_\rho v'$. By the planarity of $\sigma(G)$, we obtain $d(w)\preceq_\lambda d(w'))$; see \cref{fig:forward-4}.
If $d(w) = d(w'))$, then by Rules~\ref{r:2} and \ref{r:3} it follows that $w\prec_\rho w'$ because of the counterclockwise traversal of $d(w)=d(w')$ and the traversal of the blocks.
Otherwise, by \cref{cor:vertex-to-discoverer}, it follows that $w\prec_\rho w'$ .

\item[--] Consider now Case~(b), in which $f_{w'}$ is small.
In this case, the order is $\text{dom}(d(v'))\prec_\rho w'\prec_\rho v'\prec_\rho v$.
As illustrated in \cref{fig:forward-5}, the only $L_1$-vertices that can be on the boundary of $f_w$ and precede $v'$ in $\rho$ are vertices assigned to a block $B$ such that $B$ appears before any $L_0$-vertex different from $\text{dom}(f)$ in a counterclockwise traversal of $f$ starting from $\text{dom}(f)$. But then we obtain $w\prec_\rho w'$.

\item[--] Finally, we consider Case~(c), in which neither $f_w$ nor $f_{w'}$ is small. Hence, $d(w)\prec_\lambda f_w$ and $d(w')\prec_\lambda f_{w'}$.
Observe that if $d(w)\prec_\lambda d(w')$, the claim follows by \cref{cor:vertex-to-discoverer}.
We proceed by considering the two subcases, namely, $d(w') \prec_\lambda d(w)$ and $d(w')=d(w)$.
In the former case, the vertices $v$ and $v'$ are the dominators of the two intra-level faces $f_w$ and $f_{w'}$, and $v' \prec_\rho v$. This yields $f_{w'}\prec_\lambda f_w$.
However, since $w\prec_\rho v'$ and since $v'$ is the dominator of $f_{w'}$, we obtain the order: $d(w')\prec_\lambda d(w)\prec_\lambda f_{w'}\prec_\lambda f_w$.
This contradicts the planarity of $\sigma(G)$, as illustrated in \cref{fig:forward-1}.
Consider now the latter case, in which $d(w')=d(w)$.
Since $w$ belongs to $L_1$, vertex $w$ belongs to the boundary of block $B(w)$ discovered by $d(w')=d(w)$. 
Similarly, vertex $w'$ belongs to the boundary of block $B(w')$ discovered by $d(w')=d(w)$. 
For the two blocks $B(w')$ and $B(w)$, either $B(v) \neq B(w)$ or $B(v)=B(w)$ holds.
Assume first that $B(v) \neq B(w)$. 
Since $B(w')$ and $B(w)$ are discovered by the same face, and since $f_{w'} \prec_\lambda f_w$, it follows that $B(w)$ precedes $B(w')$ in the counterclockwise traversal of $d(w')=d(w)$. Otherwise the faces $f_{w'}$ and $f_w$ would violate the planarity of $\sigma(G)$, as illustrated in \cref{fig:forward-2}.
Thus, by \cref{prp:block-vertex-order}, we obtain $w\prec_\rho w'$.
To complete the proof, it remains to consider the case in which $B(w') = B(w)$. 
Similar to the case above, by Rule~\ref{r:3}, in the counterclockwise traversal of $B(w') = B(w)$ 
starting from its leader, vertex $w$ precedes $w'$ since otherwise the faces $f_{w'}$ and $f_w$ violate the planarity of $\sigma(G)$, as illustrated in \cref{fig:forward-3}.
\end{itemize}
The above case analysis completes the proof.
\end{proof}

\noindent In the following, we describe properties that will be useful in the egde-to-page assignment of the non-dominator edges.

\begin{figure}[tb!]
	\centering
	\begin{subfigure}{.48\textwidth}
	\centering
	\includegraphics[width=.8\textwidth,page=4]{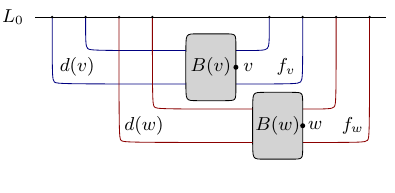}
	\subcaption{}
	\label{fig:pdl-14}
	\end{subfigure}		
	\begin{subfigure}{.48\textwidth}
	\centering
	\includegraphics[width=.8\textwidth,page=9]{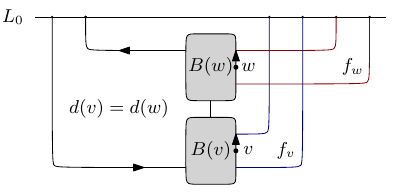}
	\subcaption{}
	\label{fig:pdl-25}
	\end{subfigure}	
	\begin{subfigure}{.48\textwidth}
	\centering
	\includegraphics[width=.8\textwidth,page=10]{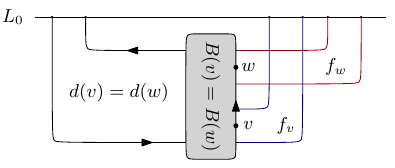}
	\subcaption{}
	\label{fig:pdl-26}
	\end{subfigure}	
	\caption{%
	Illustration for the proof of \cref{lem:pdl}.}
	\label{fig:pdf-256}
\end{figure}

\begin{lemma}\label{lem:pdl}
Let $v$ and $w$ be two vertices of $G$, such that $v \prec_\rho w$. Also, let $f_v$ and $f_w$ be two intra-level faces containing $v$ and $w$ on their boundaries, respectively, such that $f_v \prec_\lambda f_w$. If the following conditions hold, then $f_v \preceq_\lambda d(w)$. 
\begin{enumerate}[(i)]
\item \label{c:pdl0} $v$ is $d(v)$-prime, 
\item \label{c:pdl1} $w$ is $d(w)$-prime, 
\item \label{c:pdl3} $v$ and $w$ are not the dominators of $f_v$ and $f_w$, respectively,
\end{enumerate} 
\end{lemma}
\begin{proof}
First, observe that by \cref{prp:discover-order}, we have $d(v)\preceq_\lambda d(w)$. We proceed by considering four cases based on whether $v$ and $w$ belong to $L_0$ or to $L_1$ as follows: (a)~$v$ and $w$ belong to $L_0$, (b)~$v$ belongs to $L_0$ and $w$ belongs to $L_1$, (c)~$v$ belongs to $L_1$ and $w$ belongs to $L_0$, and (d)~$v$ and $w$ belong to $L_1$.

\begin{itemize}
\item[--] We start with Case~(a), in which $v$ and $w$ belong to $L_0$.
Since $v$ is $d(v)$-prime, it follows by \cref{prp:prime} that $v$ is also $f_v$-prime.
However, since $v$ is not the dominator of $f_v$, it follows that $d(v)=f_v$.
The same holds for vertex $w$ and the faces $d(w)$ and $f_w$.
Now, the claim $f_v\preceq_\lambda d(w)$ is an immediate consequence of the assumption $f_v \prec_\lambda f_w$.

\item[--] Consider now Case~(b), in which $v$ belongs to $L_0$ and $w$ belongs to $L_1$.
By \cref{prp:prime,c:pdl0}, we know that $v$ is $f_v$-prime.
By \cref{prp:first-the-discoverer}, we obtain $d(v)\preceq_\lambda f_v$.
If $d(v)\prec_\lambda f_v$, \cref{prp:prime} implies $v=\text{dom}(f_v)$ which contradicts \cref{c:pdl3}.
However, if $d(v)= f_v$, the claim follows from $d(v)\preceq_\lambda d(w)$.

\item[--] We proceed with Case~(c), in which $v$ belongs to $L_1$ and $w$ belongs to $L_0$.
Consider vertex $w$.
As above, by \cref{prp:prime,c:pdl1}, it follows that $w$ is $f_w$-prime and therefore, by \cref{c:pdl3}, $d(w)= f_w$ holds.
Recalling the assumption $f_v \prec_\lambda f_w$, the claim $f_v \preceq_\lambda d(w)$ is a direct consequence of $f_v \prec_\lambda f_w$.

\item[--] To complete the proof of the lemma, we consider Case~(d), in which $v$ and $w$ belong to $L_1$.
Assume to the contrary that $d(w) \prec_\lambda f_v$.
This implies $d(v) \preceq_\lambda d(w)\prec_\lambda f_v \prec_\lambda f_w$. 
We consider the two subcases, namely, $d(v) \prec_\lambda d(w)$ and $d(v)=d(w)$.
In the former case, since $v$ belongs to $L_1$, vertex $v$ belongs to the boundary of block $B(v)$ discovered by $d(v)$.
Similarly, vertex $w$ belongs to the boundary of block $B(w)$ discovered by $d(w)$.
Hence, we have $B(v) \neq B(w)$, as $d(v) \prec_\lambda d(w)$; see \cref{fig:pdl-14}.
The order $f_v \prec_\lambda f_w$ violates the planarity of $\sigma(G)$; a contradiction.
We now consider the case, in which $d(v)=d(w)$. Since $v$ belongs to $L_1$, vertex $v$ 
belongs to the boundary of block $B(v)$ discovered by $d(v)=d(w)$. 
Similarly, vertex $w$ belongs to the boundary of block $B(w)$ discovered by $d(v)=d(w)$. 
For the two blocks $B(v)$ and $B(w)$ either $B(v)\neq B(w)$ or $B(v)=B(w)$ holds.
First, assume that $B(v) \neq B(w)$; see \cref{fig:pdl-25}. 
$B(v)$ and $B(w)$ are discovered by the same face, and $v \prec_\rho w$. By Rule~\ref{r:2} it follows $B(v)$ precedes $B(w)$ in the counterclockwise traversal of $d(v)=d(w)$. With $f_v \prec_\lambda f_w$, 
the planarity of $\sigma(G)$ is violated; a contradiction. 
Next, assume $B(v) = B(w)$. 
Since $v \prec_\rho w$, by Rule~\ref{r:3}, in the counterclockwise traversal of $B(v) = B(w)$ 
starting from its leader, vertex $v$ precedes $w$; see \cref{fig:pdl-26}. 
The order $f_v \prec_\lambda f_w$ violates the planarity of $\sigma(G)$; a contradiction.
\end{itemize}
The above case analysis completes the proof.
\end{proof}

\noindent The next lemma reveals a relationship between two faces containing two edges that cross in the linear ordering.

\begin{lemma}\label{lem:alternation}
Let $v$, $w$, $x$ and $z$ be four vertices of $G$, such that $(v,w)$ and $(x,z)$ are two non-dominator edges of $G$, and $v \prec_\rho x \prec_\rho w \prec_\rho z$. 
Let $f_{vw}$ be a face with $v$ and $w$ on its boundary, and let $f_{xz}$ be a face with $x$ and $z$ 
on its boundary such that $f_{vw}$ and $f_{xz}$ are two distinct faces. Moreover, $v$ and $w$ are $f_{vw}$-prime, whereas $x$ and $z$ are $f_{xz}$-prime.
Then $d(x) = f_{vw}$ or $d(w)=f_{xz}$ holds.
\end{lemma}
\begin{proof}
We first show that $v$ cannot belong to $L_0$. Assume the contrary. Vertex $v$ is not the dominator of $f_{vw}$, and $v \prec_\rho w $. Thus, it follows that $w$ also belongs to $L_0$. Since $w$ is $f_{vw}$-prime, and $v \prec_\rho w $, the only way for $x$ to lie between $v$ and $w$ in $\rho$ is when $f_{vw}=f_{xz}$ holds; a contradiction.
The same argumentation holds for $x\prec_\rho w\prec_\rho z$.
Hence, we may assume that both $v$ and $x$ belong to $L_1$.
By \cref{prp:discover-order}, we have $d(v) \preceq_\lambda d(x)$.
We assume to the contrary that $d(x) \neq f_{vw}$ and $d(w) \neq f_{xz}$ hold.

We consider the two cases (a)~$f_{vw} \prec_\lambda f_{xz}$ and (b)~$f_{xz} \prec_\lambda f_{vw}$.
First, consider Case~(a). We continue by distinguishing between two subcases based on whether $f_{vw} \prec_\lambda d(x)$ or $d(x) \prec_\lambda f_{vw}$.
We start with $f_{vw} \prec_\lambda d(x)$. 
This implies that every $f_{vw}$-prime vertex precedes any $d(x)$-prime vertex that is discovered by $d(x)$. 
Since $w$ is $f_{vw}$-prime, and $x$ belongs to $L_1$, it follows that $w \prec_\rho x$; a contradiction.
Hence, we may focus on the case $d(x) \prec_\lambda f_{vw}$.
Our plan is to apply \cref{lem:pdl} on vertices $v$ and $x$ for which we know that $v \prec_\rho x$ and $f_{vw}\prec_\lambda f_{xz}$.
Since $v$ and $x$ belong to $L_1$, \cref{c:pdl0,c:pdl1} of \cref{lem:pdl} are satisfied.
Furthermore, with $(v,w)$ and $(x,z)$ being non-dominator edges, \cref{c:pdl3} of \cref{lem:pdl} holds as well.
Hence, by \cref{lem:pdl}, we obtain $f_{vw} \preceq_\lambda d(x)$. This contradicts the original assumption $d(x) \prec_\lambda f_{vw}$.

\begin{figure}[tb!]
	\centering
	\begin{subfigure}{.48\textwidth}
		\centering
		\includegraphics[width=.9\textwidth,page=12]{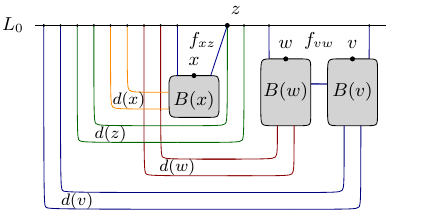}
		\subcaption{}
		\label{fig:conflictlemma-10}
	\end{subfigure}		
	\begin{subfigure}{.48\textwidth}
		\centering
		\includegraphics[width=.9\textwidth,page=16]{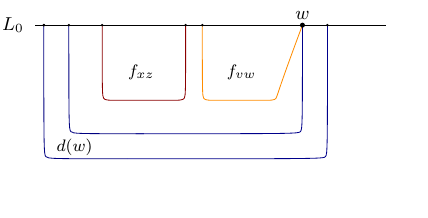}
		\subcaption{}
		\label{fig:conflictlemma-11}
	\end{subfigure}
	\begin{subfigure}{.48\textwidth}
		\centering
		\includegraphics[width=.9\textwidth,page=15]{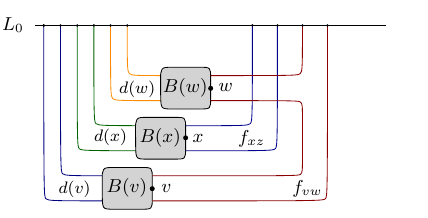}
		\subcaption{}
		\label{fig:conflictlemma-12}
	\end{subfigure}			
	\caption{%
		Illustrations for the proof of \cref{lem:alternation}.}
	\label{fig:conflict-lemma}
\end{figure}

Next, consider Case~(b), in which $f_{xz} \prec_\lambda f_{vw}$. 
By \cref{prp:first-the-discoverer}, we have $d(x) \preceq_\lambda f_{xz}$ which together with $d(v) \preceq_\lambda d(x)$ implies $d(v) \preceq_\lambda d(x) \preceq_\lambda f_{xz} \prec_\lambda f_{vw}$.
By assumption, $d(w) \neq f_{xz}$.
We continue by considering two subcases based on whether $f_{xz} \prec_\lambda d(w)$ or $d(w) \prec_\lambda f_{xz}$.
First, assume $f_{xz} \prec_\lambda d(w)$. By \cref{prp:first-the-discoverer}, it follows that $d(z) \preceq_\lambda f_{xz}$. The latter two inequalities imply $d(z) \prec_\lambda d(w)$. 
We may assume that $z$ is $d(z)$-prime, since otherwise face $d(w)$ violates planarity as shown in \cref{fig:conflictlemma-10}.
Hence, by \cref{cor:vertex-to-discoverer}, it follows that $z \prec_\rho w$ which contradicts our assumption $w \prec_\rho z$.
Hence, in the following we consider the case $d(w) \prec_\lambda f_{xz}$.
We distinguish two subcases based on whether $w$ belongs to $L_0$ or to $L_1$.
First, consider the case, in which $w$ belongs to $L_0$.
If $w$ is $d(w)$-prime, $d(w)=f_{vw}$ follows by \cref{prp:prime} since $w$ is not the dominator of $f_{vw}$.
Therefore, we have $d(w) \prec_\lambda f_{xz}\prec_\lambda f_{vw}=d(w)$; a contradiction.
Thus, we may assume that $w$ is not $d(w)$-prime which yields $d(w) \prec_\lambda f_{xz}\prec_\lambda f_{vw}$.
However, since $z$ is $f_{xz}$-prime, we have $z\prec_\rho w$ as shown in \cref{fig:conflictlemma-11}; a contradiction.
To compete the proof of the lemma, it remains to consider the case, in which $w$ belongs to $L_1$.
Observe that $d(v) \preceq_\lambda d(x) \preceq_\lambda d(w)$ by \cref{prp:discover-order}.
This yields $d(v) \preceq_\lambda d(x) \preceq_\lambda d(w) \prec_\lambda f_{xz}\prec_\lambda f_{vw}$.
As illustrated in~\cref{fig:conflictlemma-12}, $f_{vw}$ violates the planarity of $\sigma(G)$.
\end{proof}

\noindent Observe that in \cref{lem:alternation} the edges $(v,w)$ and $(x,z)$ form two non-dominator edges that cannot be assigned to the same page. \cref{lem:alternation} translates this conflict into a relationship between the two faces $f_{vw}$ and $f_{xz}$ containing these edges. In the following, we model these conflicts as edges of an auxiliary graph which we call the \emph{conflict graph} and denote by $\mathcal{C}(G)$; see also \cref{fig:example-conflict} for an illustration.

\begin{figure}[t!]
	\centering
	\includegraphics[width=.7\textwidth,page=4]{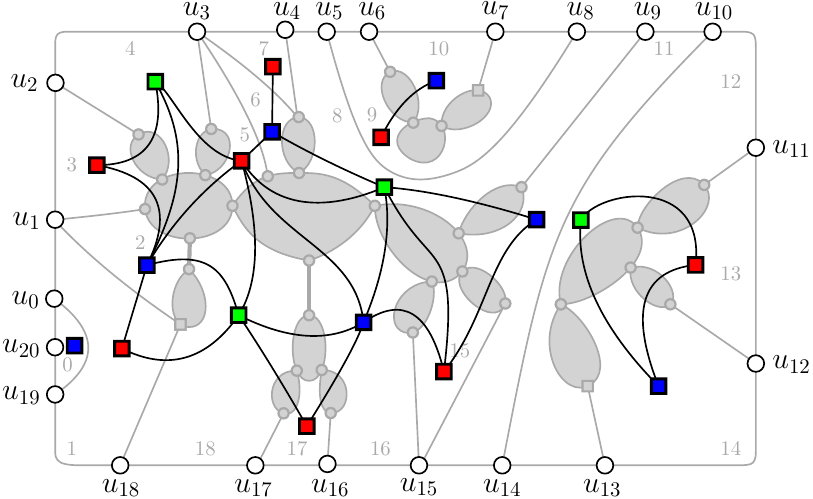}
	\caption{The conflict graph of the example illustrated in \cref{fig:example}. }
	\label{fig:example-conflict}
\end{figure}

\begin{definition}
The \emph{conflict graph} $\mathcal{C}(G)$ of $G$ is an undirected graph whose vertices are the faces of $\mathcal{F}$. There exists an edge $(f,g)$ with $f \neq g$ in $\mathcal{C}(G)$ if and only if there exists a vertex $w$ of level $L_1$ on the boundary of $g$ such that $f=d(w)$.
\end{definition}

\noindent With this definition, we are can restate \cref{lem:alternation} as follows.

\begin{lemma}\label{lem:magic-revised}
Let $(v,w)$ and $(x,z)$ be two non-dominator edges of $G$ belonging to two distinct faces $f_{vw}$ and $f_{xz}$
such that $v$ and $w$ are $f_{vw}$-prime, $x$ and $z$ are $f_{xz}$-prime, $v \prec_\rho w$, and $x \prec_\rho z$.
If $(v,w)$ and $(x,z)$ cross in $\rho$, then there is an edge $(f_{vw},f_{xz})$ in $\mathcal{C}(G)$.
\end{lemma}
\begin{proof} 
Without loss of generality, we may assume $v\prec_\rho x\prec_\rho w\prec_\rho z$.
As in the proof of \cref{lem:alternation}, we first show that $v$ and $x$ belong to $L_1$.
Furthermore, by \cref{lem:alternation}, we have that $f_{vw}=d(x)$ or $f_{xz}=d(w)$ holds.
Since $x$ belongs to $L_1$, it follows that there is an edge $(f_{vw},f_{xz})$ in $\mathcal{C}(G)$ if $f_{vw}=d(x)$ holds.
Thus, consider $f_{xz}=d(w)$.
If $w$ belongs to $L_1$, it follows that there is an edge $(f_{vw},f_{xz})$ in $\mathcal{C}(G)$.
Hence, assume that $w$ is on $L_0$.
Recall that $f_{vw} \neq f_{xz}$ holds, vertex $w$ is $f_{vw}$-prime, and we have $w\neq \text{dom}(f_{vw})$, since $(v,w)$ is not a dominator edge.
We split the proof into the two cases (a)~$f_{vw} \prec_\lambda f_{xz}$ and (b)~$f_{xz} \prec_\lambda f_{vw}$.
In Case (a), we get $d(w)\preceq_\lambda f_{vw} \prec_\lambda f_{xz} =d(w)$ by \cref{prp:first-the-discoverer}; a contradiction.
In Case (b), we observe that if $w$ is $d(w)$-prime, we have $d(w)=f_{xz}\prec_\lambda f_{vw}$ and thus, $w=\text{dom}(f_{vw})$ by \cref{prp:prime}; a contradiction.
Hence, we may assume that $w$ is not $d(w)$-prime.
However, since $w$ is $f_{vw}$-prime and $w\neq \text{dom}(f_{vw})$, there is at least one vertex on $L_0$ right before $w$ in a clockwise traversal of $L_0$ that is also on the boundary of $f_{vw}$. This is illustrated in \cref{fig:conflict-L1-1}.
Now, recall that by \cref{prp:discover-order} and since $v$ and $x$ belong to $L_1$, we have $d(v)\preceq_\lambda d(x)$.
Together with \cref{prp:first-the-discoverer}, we conclude that $d(v)\preceq_\lambda d(x) \preceq_\lambda f_{xz}$.
In fact, $d(v)= d(x) = f_{xz}$ has to hold; otherwise not both $d(v)$ and $f_{vw}$ could bound the block $B(v)$ without violating the planarity of $\sigma(G)$. Since $d(v)= f_{xz}$ and since $v$ belongs to $L_1$, the edge $(f_{vw},f_{xz})$ exists in $\mathcal{C}(G)$.
\end{proof}

\begin{figure}[t]
	\centering
	\includegraphics[width=.4\textwidth,page=1]{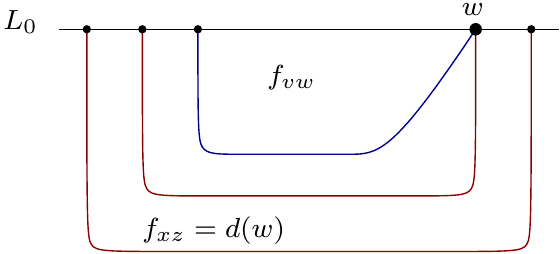}
	\caption{Illustration for the proof of \cref{lem:magic-revised}.}
	\label{fig:conflict-L1-1}
\end{figure}

\noindent In the following lemma, we prove an important property of the conflict graph.

\begin{lemma}\label{lem:one-page}
Graph $\mathcal{C}(G)$ is 1-page book embeddable.
\end{lemma}
\begin{proof}
We order the vertices of $\mathcal{C}(G)$ according to $\lambda(\mathcal{F})$. 
Suppose for contradiction that two edges $(f,g)$ and $(f',g')$ of $\mathcal{C}(G)$ cross in $\lambda(\mathcal{F})$
such that, without loss of generality, $f \prec_\lambda f' \prec_\lambda g \prec_\lambda g'$.
By definition~of~$\mathcal{C}(G)$, there is either a vertex $v$ of level $L_1$ on the boundary of $f$ such that $g=d(v)$, or there is a vertex $w$ of level $L_1$ on the boundary of $g$ such that $f=d(w)$.
In the first case, by \cref{prp:first-the-discoverer}, we have $d(v)\preceq_\lambda f$, which contradicts $g=d(v)\preceq_\lambda f\prec_\lambda g$.
Now consider the second case.
We argue analogously for the edge $(f',g')$.
Hence, there exist two vertices $w$ and $w'$ of level $L_1$ on the boundaries of $g$ and $g'$, respectively, such that $f=d(w)$ and $f'=d(w')$ hold. 
This yields $d(w) \prec_\lambda d(w') \prec_\lambda g \prec_\lambda g'$.
Since $w$ and $w'$ belong to $L_1$, they are $d(w)$- and $d(w')$-prime, respectively.
By \cref{cor:vertex-to-discoverer} and since $w\neq w'$, we have $w\prec_\rho w'$.
Now we apply \cref{lem:pdl} on $w$ and $w'$ with $f_v=g$ and $f_w=g'$, and obtain $g \preceq_\lambda d(w)$, a contradiction to the fact that $d(w) \prec_\lambda g$.
\end{proof}

\noindent Since $\mathcal{C}(G)$ is 1-page book embeddable, it is outerplanar~\cite{DBLP:journals/jct/BernhartK79}. Hence, the following corollary becomes a direct implication of \cref{lem:one-page}.

\begin{corollary}\label{lem:3-coloring}
Graph $\mathcal{C}(G)$ admits a vertex coloring with three colors.
\end{corollary}

We are now ready to describe how to assign the edges of $G$ to the pages of the book embedding. First, we embed all backward edges in a single page $p_0$ and all forward edges in a single page $p_1$. 
By \cref{lem:backward,lem:forward}, this assignment is valid. Next, we assign the remaining edges of $G$ to a total of $3 \cdot\left\lceil \frac{k}{2} \right\rceil$ pages. To ease the description, we partition these pages into three sets $R^1$, $B^1$, and $G^1$, each containing $\left\lceil \frac{k}{2} \right\rceil$ pages as follows: $R^1 = \{r^1_1,\ldots,r^1_{\lceil k/2 \rceil}\}$, $B^1 = \{b^1_1,\ldots,b^1_{\lceil k/2 \rceil}\}$, and $G^1 = \{g^1_1,\ldots,g^1_{\lceil k/2 \rceil}\}$.
The actual assignment is done by processing the intra-level faces of $\mathcal{F}$ according to the ordering $\lambda(\mathcal{F})$. Assume that we have processed a certain number of faces in this order and that we have assigned all the non-dominator edges of $G$ that are induced by the vertices of these faces in the pages mentioned above. Let $f$ be the next face to process.
By \cref{lem:3-coloring}, face $f$ has a color out of three available ones, say red, blue, and green. 
Now, observe that the vertices of $f$ induce at most a $k$-clique $Q_f$ in $G$.
Also, observe that some of the edges on the boundary of $f$ may have been already assigned to a page.
We assign the remaining non-dominator edges of $Q_f$ to the pages of one of the sets $R^1$, $B^1$, and $G^1$ according to the color of $f$.
Since $Q_f$ is at most a $k$-clique, $\left\lceil \frac{k}{2} \right\rceil$ pages are sufficient regardless of the underlying linear order~\cite{DBLP:journals/jct/BernhartK79}. 

The remainder of this section is devoted in proving that the (non-dominator) edges assigned to the pages in $R^1$, $B^1$, and $G^1$ do not cross, and thus that the computed book embedding is valid. Consider two non-dominator edges $(v,w)$ and $(x,z)$, and let $f_{vw}$ and $f_{xz}$ be the faces of $\mathcal{F}$ responsible for assigning $(v,w)$ and $(x,z)$ to one of the pages of $R^1 \cup B^1 \cup G^1$.
If $v$ and $w$ are $f_{vw}$-prime, and if $x$ and $z$ are $f_{xz}$-prime, then by \cref{lem:magic-revised}, we know that $(v,w)$ and $(x,z)$ do not cross.
Hence, we may assume that the edges $(v,w)$ and $(x,z)$ are incident to vertices that are not prime with respect to the face that belongs to that edge. In this direction, we need a few auxiliary lemmata.

\begin{property}\label{prp:prime-order}
Let $v$ and $w$ be two vertices with $v \prec_\rho w$ on the boundary of a face $f_{vw}$. If $w$ is $f_{vw}$-prime, 
then $v$ is also $f_{vw}$-prime. If $v$ is not $f_{vw}$-prime, then $w$ is not $f_{vw}$-prime.
\end{property}
\begin{proof}
Both claims follow from the fact that all vertices that are $f_{vw}$-prime precede 
those that are not $f_{vw}$-prime. Since $v \prec_\rho w$, the property follows. 
\end{proof}

\begin{property}\label{prp:dom-L0-L0}
Let $v$ and $w$ be two vertices of $G$. If the following conditions hold, then  $d(v)= f_{vw}$.
\begin{enumerate}[(i)]
\item \label{prp:dom-L0-L0-1}$v$ and $w$ belong to $L_0$,
\item \label{prp:dom-L0-L0-2}$v$ and $w$ are on the boundary of a face $f_{vw}$, and
\item \label{prp:dom-L0-L0-3}$\text{dom}(f_{vw}) \prec_\rho v\prec_\rho w$.
\end{enumerate}
\end{property}
\begin{proof}
Condition~\ref{prp:dom-L0-L0-1} and \cref{prp:first-the-discoverer} imply that $d(v) \preceq_\lambda f_{vw}$. To prove the property, assume to the contrary $d(v) \prec_\lambda f_{vw}$. Since, by Condition~\ref{prp:dom-L0-L0-3}, $\text{dom}(f_{vw})$ precedes $v$, vertex $v$ cannot be prime with respect to face $d(v)$ that discovers $v$. However, it follows that $v$ is the last vertex on $L_0$ in the ordering $\rho$ that is on the boundary of $f_{vw}$; see \cref{fig:dom-L0-L0}. This contradicts the existence of vertex $w$, which is also on $L_0$ (by Condition~\ref{prp:dom-L0-L0-1}), on the boundary of $f_{vw}$ and follows $v$ in the ordering $\rho$ (by Condition~\ref{prp:dom-L0-L0-3}).
\end{proof}

\begin{figure}[t!]
    \centering
    \begin{subfigure}{.48\textwidth}
	\includegraphics[width=.8\textwidth,page=4]{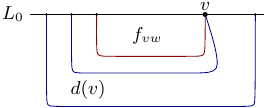}
	\caption{}
	\label{fig:dom-L0-L0}
	\end{subfigure}
	\begin{subfigure}{.48\textwidth}
    \centering
	\includegraphics[width=.8\textwidth,page=5]{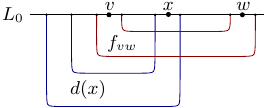}
	\caption{}
	\label{fig:L0-v-L0}
	\end{subfigure}
	\caption{Illustrations for the proofs of 
	(a)~\cref{prp:dom-L0-L0}, and 
	(b)~\cref{prp:L0-v-L0}}
\end{figure}

\begin{property}\label{prp:L0-v-L0}
Let $v$, $w$ and $x$ be three vertices of $G$. If the following conditions hold, then  $f_{vw} \preceq_\lambda d(x)$.
\begin{enumerate}[(i)]
\item \label{prp:L0-v-L0-1} $v$ and $w$ belong to $L_0$,
\item \label{prp:L0-v-L0-2} $v$ and $w$ are on the boundary of a face $f_{vw}$, and
\item \label{prp:L0-v-L0-3} $\text{dom}(f_{vw}) \prec_\rho v\prec_\rho x \prec_\rho w$.
\end{enumerate}
\end{property}
\begin{proof}
Assume to the contrary that $d(x) \prec_\lambda f_{vw}$, which implies that $\text{dom}(d(x)) \preceq_\rho \text{dom}(f_{vw})$. Hence, by Condition~\ref{prp:L0-v-L0-3}, we obtain $\text{dom}(d(x)) \preceq_\rho \text{dom}(f_{vw}) \prec_\rho v\prec_\rho x \prec_\rho w$.
Recall that $x$ is placed between $v$ and $w$ (by Condition~\ref{prp:L0-v-L0-3}), both $v$ and $w$ belong to $L_0$ (by Condition~\ref{prp:L0-v-L0-1}) and on the boundary of $f_{vw}$ (Condition~\ref{prp:L0-v-L0-2}), and neither $v$ nor $w$ is the dominator of $f_{vw}$ (by Condition~\ref{prp:L0-v-L0-3}).
It follows that either $x$ also belongs to $L_0$, or $x$ is discovered by a face $d(x)$ with $v\preceq_\rho \text{dom}(d(x))$.
The latter case contradicts the fact that $\text{dom}(d(x)) \prec_\rho v$.
In the former case, it follows that the faces $d(x)$ and $f_{vw}$ violate planarity of $\sigma(G)$; refer to \cref{fig:L0-v-L0} for an illustration. Since both cases have been led to a contradiction, the proof follows.
\end{proof}

\clearpage

\begin{figure}[t!]
    \centering
	\includegraphics[width=.35\textwidth,page=6]{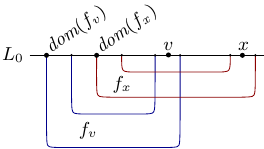}
	\caption{Illustration for the proof of \cref{prp:dom-dom-L0}.}
	\label{fig:dom-dom-L0}
\end{figure}

\begin{property}\label{prp:dom-dom-L0}
Let $v$ and $x$ be two vertices of $G$. If the following conditions hold, then $\text{dom}(f_{v}) \preceq_\rho \text{dom}(f_{x})\preceq_\rho x \prec_\rho v$.
\begin{enumerate}[(i)]
\item\label{prp:dom-dom-L0-1} $v$ and $x$ belong to $L_0$,
\item\label{prp:dom-dom-L0-2} $v$ is on the boundary of a face $f_{v}$,
\item\label{prp:dom-dom-L0-3} $x$ is on the boundary of a face $f_{x}$,
\item\label{prp:dom-dom-L0-4} $f_v \prec_\lambda f_x$, and
\item\label{prp:dom-dom-L0-5}$\text{dom}(f_{x})\prec_\rho v$.
\end{enumerate}
\end{property}
\begin{proof}
By \cref{prp:dom-dom-L0-4}, we obtain $\text{dom}(f_{v}) \preceq_\rho \text{dom}(f_{x})$.
Since $x$ is on the boundary of $f_{x}$ (by Condition~\ref{prp:dom-dom-L0-3}) and on $L_0$ (by Condition~\ref{prp:dom-dom-L0-1}), it follows that $\text{dom}(f_{x})\preceq_\rho x$.
This together with Condition~\ref{prp:dom-dom-L0-5} imply that, in order to prove the property, it suffices to show that $x \prec_\rho v$; recall that we have already shown that $\text{dom}(f_{v}) \preceq_\rho \text{dom}(f_{x})$. Assume to the contrary that $v \prec_\rho x$.
By \cref{prp:dom-dom-L0-4,prp:dom-dom-L0-5}, it follows that $v$ is not $f_v$-prime.
Since $v \prec_\rho x$, this leads to the order $\text{dom}(f_{v}) \preceq_\rho \text{dom}(f_{x})\prec_\rho v\prec_\rho x$ and all of these vertices belong to $L_0$ (by Condition~\ref{prp:dom-dom-L0-1}). Together with \cref{prp:dom-dom-L0-4}, this violates the planarity of $\sigma(G)$, as illustrated in \cref{fig:dom-dom-L0}.
\end{proof}

\begin{lemma}\label{lem:excluding}
Let $x$ and $z$ be two vertices of $G$ belonging to the boundary of a face $f_{xz}$ such that $\text{dom}(f_{xz})\prec_\rho x \prec_\rho z$, and let $f$ be a face preceding $f_{xz}$ in $\lambda(\mathcal{F})$, that is, $f \prec_\lambda f_{xz}$.  
Then, for any vertex $y$ of $G$ with $x \prec_\rho y \prec_\rho z$, we have that $y$ is not on the boundary of $f$.
\end{lemma}
\begin{proof}
First, we claim that $x$ is discovered by $f_{xz}$, that is $f_{xz} = d(x)$.
If $x$ belongs to $L_1$, the claim follows from $\text{dom}(f_{xz})\prec_\rho x$.
Now consider the case in which $x$ belongs to $L_0$.
Since $x$ is preceded by $\text{dom}(f_{xz})$ and followed by vertex $z$, and both vertices belong to the boundary of $f_{xz}$, vertex $z$ must belong to $L_0$ as well. \cref{prp:dom-L0-L0} concludes the claim.
Assume for a contradiction that there exists a vertex $y$ with $x \prec_\rho y \prec_\rho z$ that is on the boundary of $f$. Note that by assumption $x \neq y \neq z$ holds. We distinguish two cases.

\begin{itemize}
\item[--] \emph{Vertex $y$ belongs to $L_1$}:
In this case, $y$ is $f$-prime and assigned to the block $B(y)$. Since $y$ is on the boundary of $f$, we obtain $d(B(y)) \preceq_\lambda f \prec_\rho f_{xz} = d(x)$. Hence, it follows by \cref{cor:vertex-to-discoverer} that $y \prec_\rho x$; a contradiction.

\item[--] \emph{Vertex $y$ belongs to $L_0$}: We first observe that $\text{dom}(f_{xz}) \prec_\rho y$ holds, as otherwise we have that $y \preceq_\rho \text{dom}(f_{xz}) \prec_\rho x \prec_\rho z$, which is a clear contradiction. 
Vertex $z$ either belongs to $L_0$ or to $L_1$.
First, assume that $z$ belongs to $L_0$.
By \cref{prp:dom-dom-L0}, we obtain $\text{dom}(f) \preceq_\rho \text{dom}(f_{xz})\prec_\rho z \prec_\rho y$; a contradiction.
In the latter case, $z$ is assigned to the block $B(z)$ and with $\text{dom}(f_{xz}) \prec_\rho z$, we get $d(B(z))=d(z)=f_{xz}$. 
By Rule \ref{r:1}, $z$ is placed right after $\text{dom}(f_{xz})$ and to the left of the next vertex on $L_0$ after $\text{dom}(f_{xz})$. With $y$ belonging to $L_0$, we obtain $z\prec_\rho y$; a contradiction.
\end{itemize}
Since each of the cases above have been led to a contradiction, the proof of the lemma follows.
\end{proof}

\noindent As a next step, we will consider all cases of crossing non-dominator edges that might arise depending on whether the endpoints are prime or not.
In order to reduce the number of cases we show the two following lemmata.

\begin{lemma}\label{lem:prime-xz}
Let $(v,w)$ and $(x,z)$ be two non-dominator edges of $G$ belonging to two distinct faces $f_{vw}$ and $f_{xz}$, respectively, such that $v\prec_\rho w$, $x\prec_\rho z$ and $f_{vw} \prec_\lambda f_{xz}$. If $(v,w)$ and $(x,z)$ cross, then either the edge $(f_{vw},f_{xz})$ exists in $\mathcal{C}(G)$, or
there exists a non-dominator edge $(x', z')$ in $f_{xz}$ with $x'$ and $z'$ being $f_{xz}$-prime such that $(v,w)$ and $(x',z')$ cross.
\end{lemma}
\begin{proof}
Since $v\prec_\rho w$, $x\prec_\rho z$ and since $(v,w)$ and $(x,z)$ cross, either (a)~$v \prec_\rho x \prec_\rho w \prec_\rho z$ or (b)~$x \prec_\rho v \prec_\rho z \prec_\rho w$ holds. We proceed by distinguishing different cases depending on whether $x$ and $z$ are $f_{xz}$-prime or not. 

We first claim that  at least one of the vertices $x$ and $z$ is $f_{xz}$-prime. For a contradiction, assume that neither $x$ nor $z$ is $f_{xz}$-prime. In this case, $\text{dom}(f_{xz}) \prec_\rho x \prec_\rho z$ holds.  For the partial order of $v$, $w$, $x$ and $z$ of Case~(a), we apply \cref{prp:L0-v-L0} on vertices $x$, $z$, and $w$, and we obtain $f_{xz} \preceq_\lambda d(w)$. By \cref{prp:first-the-discoverer}, we further obtain that $d(w) \preceq_\lambda f_{vw}$. Hence, $f_{xz}\preceq_\lambda d(w)\preceq_\lambda f_{vw}$ must hold, which is a contradiction to the fact that $f_{vw}\prec_\lambda f_{xz}$. For the partial order of Case~(b), we obtain a contradiction by applying an argument analogous to the one above in which we interchange the roles of $w$ and~$v$.

\begin{figure}[tb!]
	\centering
	\begin{subfigure}{.48\textwidth}
		\centering
		\includegraphics[width=.8\textwidth,page=1]{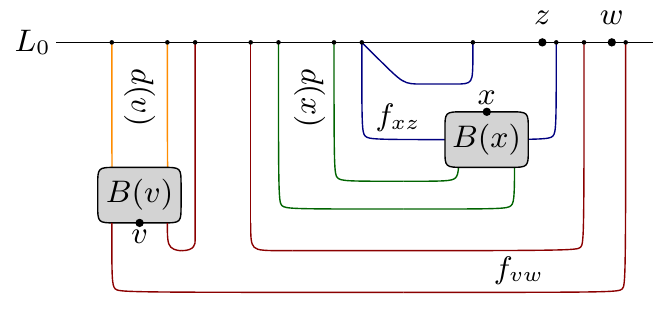}
		\subcaption{}
		\label{fig:prime-xz-1}
	\end{subfigure}
	\begin{subfigure}{.48\textwidth}
		\centering
		\includegraphics[width=.8\textwidth,page=7]{figures/prime-xz}
		\subcaption{}
		\label{fig:prime-xz-3}
	\end{subfigure}
	\begin{subfigure}{.48\textwidth}
		\centering
		\includegraphics[width=.8\textwidth,page=8]{figures/prime-xz}
		\subcaption{}
		\label{fig:prime-xz-4}
	\end{subfigure}
	\begin{subfigure}{.48\textwidth}
		\centering
		\includegraphics[width=.8\textwidth,page=2]{figures/prime-xz}
		\subcaption{}
		\label{fig:prime-xz-2}
	\end{subfigure}
	\begin{subfigure}{.48\textwidth}
		\centering
		\includegraphics[width=.8\textwidth,page=9]{figures/prime-xz}
		\subcaption{}
		\label{fig:prime-xz-5}
	\end{subfigure}
	\caption{%
	Illustrations for the proof of \cref{lem:prime-xz}.}
	\label{fig:prime-xz}
\end{figure}

By the above claim, we may assume that at least one of the vertices $x$ and $z$ is $f_{xz}$-prime. Note that if $x$ is not $f_{xz}$-prime, then, by \cref{prp:prime-order}, $z$ is not $f_{xz}$-prime either. Hence, we can conclude that $x$ is $f_{xz}$-prime, while $z$ is not $f_{xz}$-prime. We proceed by setting $x'$ to be $x$ (i.e., $x':=x$). Since $z$ is not $f_{xz}$-prime, $z$ belongs to $L_0$. It follows that $\text{dom}(f_{xz})\prec_\rho z$.

We first rule out the case, in which there exists an $f_{xz}$-prime vertex $\overline{z}$, such that $\text{dom}(f_{xz})  \prec_\rho \overline{z} \prec_\rho z$. By \cref{lem:excluding}, there is no vertex between $\overline{z}$ and $z$ in $\rho$ that belongs to the boundary of $f_{vw}$.
Hence, edges $(v,w)$ and $(x,\overline{z})$ cross, since $(v,w)$ and $(x,z)$ cross. The proof of the lemma follows by setting $z'$ to be $\overline{z}$.

To complete the proof of the lemma, we have to focus on the case, in which there exists no $f_{xz}$-prime vertex as defined above. In this case, the dominator of $f_{xz}$ is the only $f_{xz}$-prime vertex on $L_0$. Since $x$ is $f_{xz}$-prime and since $x$ is not the dominator of face $f_{xz}$ (recall that the edge $(x,z)$ is a non-dominator edge and $x \prec_\rho y$), it follows that $x \prec_\rho \text{dom}(f_{xz})$, which in particular implies that $x$ belongs to $L_1$. Since $x \prec_\rho \text{dom}(f_{xz})$, either the face $d(x)$ that discovers $x$ strictly precedes $f_{xz}$ in $\lambda(\mathcal{F})$ or $d(x)$ is identified with $f_{xz}$ and $f_{xz}$ is small.

We first prove that the latter case does not apply. To see this, assume for a contradiction that $f_{xz}$ is small. Then, $\text{dom}(f_{xz})$ is the only $L_0$-vertex on the boundary of $f_{xz}$. Since $(x,z)$ is a non-dominator edge, it follows that $\text{dom}(f_{xz})\neq z$, which is a contradiction since $z$ belongs to $L_0$. From the discussion above it follows that $d(x) \prec_\lambda f_{xz}$. We next argue that $d(x)=f_{vw}$ holds, which implies that the edge $(f_{vw},f_{xz})$ exist in $\mathcal{C}(G)$, since we have already proved that $x$ belongs to $L_1$. Hence, the proof of this property also concludes the proof of this lemma.

We assume for a contradiction that $d(x) \neq f_{vw}$ holds. We distinguish two cases based on whether $f_{vw} \prec_\lambda d(x)$ or $d(x) \prec_\lambda f_{vw}$. First, suppose that $f_{vw} \prec_\lambda d(x)$ and consider the partial order of Case~(a). Since $x \prec_\rho w$ and $f_{vw} \prec_\lambda d(x)$, it follows that vertex $w$ is $f_{vw}$-prime. Thus, $w$ belongs to $L_0$ and $\text{dom}(d(x)) \prec_\rho w$.  
By the planarity of $\sigma(G)$, it follows that $z \prec_\rho w$ (see \cref{fig:prime-xz-1}); a contradiction.
Consider now the partial order of Case~(b). By \cref{prp:first-the-discoverer}, we obtain we get $d(v)\preceq_\lambda f_{vw} \prec_\lambda d(x)$. Since $x\prec_\rho v$, it follows that $v$ belongs to $L_0$ and is not $f_{vw}$-prime. By \cref{prp:prime-order}, $w$ is also not $f_{vw}$-prime; see \cref{fig:prime-xz-3} for an illustration.
For $v\prec_\rho z\prec_\rho w$ to hold, we must have $v\preceq_\rho \text{dom}(d(z))\prec_\rho w$ and $d(z)$ cannot be small.
Thus, $x$ and $z$ cannot both be on the boundary of $f_{xz}$ without violating the planarity of $\sigma(G)$; a contradiction.

Suppose now that $d(x) \prec_\lambda f_{vw}$ and consider first the partial order of Case~(a). Since $x$ belongs to $L_1$, it is $d(x)$-prime.
We apply \cref{prp:discover-order} with $v \prec_\rho x$, and we get $d(v) \preceq_\lambda d(x)$. Hence, the order is $d(v) \preceq_\lambda d(x) \prec_\lambda f_{vw} \prec_\lambda f_{xz}$. If $v$ belongs to $L_0$, vertex $v$ cannot be $d(v)$-prime, since otherwise $v=\text{dom}(f_{vw})$ follows by \cref{prp:prime}; a contradiction.
Since $v$ is not the dominator of $f_{vw}$, it follows that $v$ is the last $L_0$ vertex on the boundary of $f_{vw}$, and hence, we get that $w\prec_\rho v$, as illustrated in \cref{fig:prime-xz-4}.
Thus, we may assume that both $v$ and $x$ belong to $L_1$.
Furthermore, $d(x)$ and $f_{xz}$ are both incident to $B(x)$ while $d(v)$ and $f_{vw}$ are both incident to $B(v)$.
By the planarity of $\sigma(G)$, we have $B(x)=B(v)$ which we abbreviate with $B$.
Thus, $d(x)=d(v)$, which we abbreviate with $d$.
The three faces $d$, $f_{vw}$, and $f_{xz}$ are incident to block $B$ and by the fact that $d \prec_\lambda f_{vw} \prec_\lambda f_{xz}$, they appear in this counterclockwise order around $B$.
This violates the planarity of $\sigma(G)$, as illustrated in \cref{fig:prime-xz-2}.

Next, consider the partial order of Case~(b).
We have $d(x)\prec_\lambda f_{vw}\prec_\lambda f_{xz}$.
Recall that $z$ is not $f_{xz}$-prime and therefore an $L_0$-vertex different from $\text{dom}(f_{xz})$; see \cref{fig:prime-xz-5}.
Now vertex $w$ either belongs to $L_0$ or to $L_1$.
In the first case we have the order $\text{dom}(d(x))\preceq_\rho \text{dom}(f_{vw})\preceq_\rho \text{dom}(f_{xz})\prec_\rho z\prec_\rho w$ of vertices on $L_0$.
Together with $d(x)\prec_\lambda f_{vw}\prec_\lambda f_{xz}$, and in order for $f_{vw}$ to bound vertex $w$, the planarity of $\sigma(G)$ is violated.
Assume the second case, that is $w$ belongs to $L_1$.
By $z\prec_\rho w$, we have $z\preceq_\rho \text{dom}(d(w))$ on $L_0$.
With \cref{prp:first-the-discoverer}, we obtain $\text{dom}(d(w))\preceq_\rho \text{dom}(f_{vw})$.
However, then we have $\text{dom}(f_{vw})\preceq_\rho \text{dom}(f_{xz})\prec_\rho z\preceq_\rho \text{dom}(d(w))\preceq_\rho \text{dom}(f_{vw})$; a contradiction.
\end{proof}

\begin{lemma}\label{lem:nonprime-vxv}
Let $(v,w)$ and $(x,z)$ be two non-dominator edges of $G$ belonging to two distinct faces $f_{vw}$ and $f_{xz}$, respectively, such that $v$ and $w$ are $f_{vw}$-prime, $x$ is $f_{xz}$-prime, and $z$ is not $f_{xz}$-prime. 
If $(v,w)$ and $(x,z)$ cross such that $v \prec_\rho x \prec_\rho w \prec_\rho z$, then the edge $(f_{vw},f_{xz})$ exists in $\mathcal{C}(G)$.
\end{lemma}
\begin{proof}
First, we rule out the case, in which $f_{vw} \prec_\lambda f_{xz}$.
Similar to the proof of \cref{lem:alternation}, we argue that $v$ cannot belong to $L_0$. To see this, assume the contrary. Since $v$ is not the dominator of $f_{vw}$ and since $v \prec_\rho w $, it follows that $w$ also belongs to $L_0$. Since $w$ is also $f_{vw}$-prime and since $v \prec_\rho w $, the only way for $x$ to appear between $v$ and $w$ in $\rho$, is if $f_{vw}=f_{xz}$, which is a contradiction to the fact that $f_{vw}$ and $f_{xz}$ are distinct. Next, we claim that $x$ belongs to $L_1$ as well. Assume the contrary. Since $z$ is not $f_{xz}$-prime, $z$ also belongs to $L_0$. Since $(x,z)$ is a non-dominator edge, it follows that $\text{dom}(f_{xz}) \prec_\rho x\prec_\rho z$. We apply \cref{lem:excluding} to $\text{dom}(f_{xz}) \prec_\rho x\prec_\rho w\prec_\rho z$ with $f_{vw} \prec_\lambda f_{xz}$ and obtain that $w$ be on the boundary of $f_{vw}$, which is a contradiction. Thus, we may assume that both $v$ and $x$ belong to $L_1$. By \cref{prp:discover-order}, it follows that $d(v) \preceq_\lambda d(x)$. Observe that if $d(x)=f_{vw}$, then the edge $(f_{vw},f_{xz})$ exists in $\mathcal{C}(G)$, as desired, since we have already shown that $x$ belongs to $L_1$. 

\begin{figure}[tb!]
	\centering
	\begin{subfigure}{.48\textwidth}
		\centering
		\includegraphics[width=.8\textwidth,page=10]{figures/100s}
		\subcaption{}
		\label{fig:nonprime-vxv-3}
	\end{subfigure}
	\begin{subfigure}{.48\textwidth}
		\centering
		\includegraphics[width=.8\textwidth,page=5]{figures/100s}
		\subcaption{}
		\label{fig:nonprime-vxv-4}
	\end{subfigure}
	\begin{subfigure}{.48\textwidth}
		\centering
		\includegraphics[width=.8\textwidth,page=1]{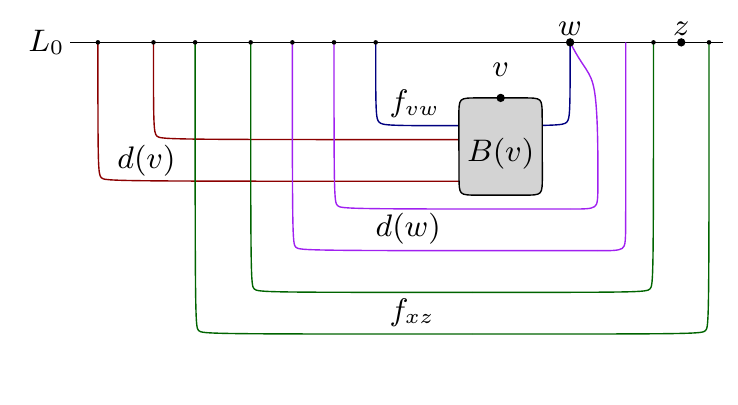}
		\subcaption{}
		\label{fig:nonprime-vxv-1}
	\end{subfigure}
	\begin{subfigure}{.48\textwidth}
		\centering
		\includegraphics[width=.8\textwidth,page=2]{figures/100s}
		\subcaption{}
		\label{fig:nonprime-vxv-2}
	\end{subfigure}
	\begin{subfigure}{.48\textwidth}
		\centering
		\includegraphics[width=.8\textwidth,page=13]{figures/100s}
		\subcaption{}
		\label{fig:nonprime-vxv-5}
	\end{subfigure}
	\caption{Illustrations for the proof of \cref{lem:nonprime-vxv}.}
\end{figure}

In order to prove the lemma for the case, in which $f_{vw} \prec_\lambda f_{xz}$, it suffices to show that the case, in which $d(x) \neq f_{vw}$, does not apply. Our proof is by contradiction. First, assume $f_{vw} \prec_\lambda d(x)$. This implies that every vertex that is $f_{vw}$-prime precedes any vertex that is $d(x)$-prime and that is discovered by $d(x)$. Since $w$ is $f_{vw}$-prime and since $x$ belongs to $L_1$, it follows that $w \prec x$, which is a contradiction. Hence, we may focus on the case, in which $d(x) \prec_\lambda f_{vw}$. Since $d(v) \preceq_\lambda d(x)$ and since $f_{vw} \prec_\lambda f_{xz}$, it follows that $d(v) \preceq_\lambda d(x) \prec_\lambda f_{vw} \prec_\lambda f_{xz}$. By \cref{prp:first-the-discoverer}, we obtain $d(w) \preceq_\lambda f_{vw}$. Our plan is to apply \cref{lem:pdl} on vertices $v$ and $x$ for which we know that $v \prec_\rho x$ and $f_{vw} \prec_\lambda f_{xz}$. Since $v$ and $x$ belong to $L_1$, \cref{c:pdl0,c:pdl1} of \cref{lem:pdl} are satisfied. Also, since $(v,w)$ and $(x,z)$ are non-dominator edges, \cref{c:pdl3} of \cref{lem:pdl} is satisfied. Hence, by \cref{lem:pdl}, we have $f_{vw} \preceq_\lambda d(x)$. This contradicts the previous assumption that $d(x) \prec_\lambda f_{vw}$.

To complete the proof of the lemma, we now consider the case, in which $f_{xz} \prec_\lambda f_{vw}$. Our aim is to apply \cref{lem:pdl} on $x \prec_\rho w$. 
For \cref{c:pdl0} of \cref{lem:pdl} to hold, we prove an even stronger argument, namely that $x$ belongs to $L_1$. Assume to the contrary that $x$ is on $L_0$. Since $z$ is not $f_{xz}$-prime, $z$ also belongs to $L_0$. Since $v \prec_\rho x \prec_\rho w \prec_\rho z$, and since $(v,w)$ and $(x,z)$ are non-dominator edges, we obtain the following order of vertices on $L_0$: $\text{dom}(d(v))\prec_\rho x \prec_\rho w \prec_\rho z$ or $\text{dom}(d(v))\prec_\rho x \preceq_\rho \text{dom}(d(w)) \prec_\rho z$, depending on whether $w$ belongs to $L_0$ or to $L_1$. However, in both cases face $f_{vw}$ violates the planarity of $\sigma(G)$ as shown in Figures \ref{fig:nonprime-vxv-3} and \ref{fig:nonprime-vxv-4}. Hence, $x$ belongs to $L_1$ and \cref{c:pdl0} of \cref{lem:pdl} is satisfied. 
We now claim that $v$ belongs to $L_1$ as well. To prove the claim, assume the contrary. Since $v$ and $w$ are on the boundary of the same face and since $v\prec_\rho w$, it follows that $w$ belongs to $L_0$, too. Since $(v,w)$ is a non-dominator edge, we get $\text{dom}(f_{vw})\prec_\rho v\prec_\rho w$. By applying \cref{lem:excluding} on $\text{dom}(f_{vw}) \prec_\rho v\prec_\rho x \prec_\rho w$, we conclude that that $x$ cannot be on the boundary of $f_{xz}$, which is a contradiction.
Hence, $v$ belong to $L_1$, as desired.
Next, we prove \cref{c:pdl1} of \cref{lem:pdl}, that is, $w$ is $d(w)$-prime. For a contradiction, assume that $w$ is not $d(w)$-prime, which yields that $w$ belongs to $L_0$. Since, by assumption, $w$ is $f_{vw}$-prime, we get $d(w)\neq f_{vw}$. In particular, by \cref{prp:first-the-discoverer}, we have that $d(w)\prec_\lambda f_{vw}$. Since $v$ and $x$ belong to $L_1$, by \cref{prp:discover-order,prp:first-the-discoverer}, it follows that $d(v)\preceq_\lambda d(x)\preceq_\lambda f_{xz}$. If $d(v)=f_{xz}$ holds, then the edge $(f_{vw},f_{xz})$ exists in $\mathcal{C}(G)$, since we have already shown that $v$ belongs to $L_1$. Thus, assume $d(v)\neq f_{xz}$ which yields $ d(v)\prec_\lambda f_{xz}$. We illustrate these relationships in \cref{fig:nonprime-vxv-1} and observe that in order for $d(v)$ and $f_{vw}$ to be incident to block $B(v)$, the planarity of $\sigma(G)$ is violated. Hence, we may assume that $w$ is $d(w)$-prime and therefore \cref{c:pdl1} of \cref{lem:pdl} is satisfied.
Finally, \cref{c:pdl3} of \cref{lem:pdl} holds trivially by the assumption that we only consider non-dominator edges which ensures that neither $x$ nor $w$ is the dominator of $f_{xz}$ or $f_{vw}$, respectively. Hence, we can apply \cref{lem:pdl} on $x \prec_\rho w$ yielding $f_{xz} \preceq_\lambda d(w)$.

Recall that if $f_{xz}=d(w)$ holds and $w$ belongs to $L_1$, then the edge $(f_{vw},f_{xz})$ exists in $\mathcal{C}(G)$. For a contradiction, we may assume that $f_{xz}$ and $f_{vw}$ do not induce an edge in $\mathcal{C}(G)$. Thus, either $f_{xz} \neq d(w)$ holds or $w$ belongs to $L_0$. However, if $f_{xz} \neq d(w)$, we obtain $d(v) \preceq_\lambda d(x) \preceq_\lambda f_{xz} \prec_\lambda d(w) \preceq_\lambda f_{vw}$, and $d(v) \neq f_{xz}$ implies $d(v) \prec_\lambda f_{xz} \prec_\lambda d(w) \preceq_\lambda f_{vw}$. Since $z$ is not $f_{xz}$-prime, it belongs to $L_0$. Since $w \prec_\rho z$, we obtain the order $\text{dom}(d(w)) \prec_\rho z$ or $w \prec_\rho z$ on $L_0$ depending on whether $w$ belongs to $L_0$ or to $L_1$. However, \cref{fig:nonprime-vxv-2,fig:nonprime-vxv-5} show that in both cases the planarity of $\sigma(G)$ is violated. Finally, assume $f_{xz} = d(w)$, but $w$ belongs to $L_0$. Since $w$ is $f_{vw}$-prime, we have that $v, x$ and $w$ are $d(v)$-, $d(x)$- and $d(w)$-prime, respectively. This yields $d(v)\preceq_\lambda d(x)\preceq_\lambda d(w)$ by \cref{prp:discover-order}. From $f_{xz}\prec_\lambda f_{vw}$ we obtain $d(v)\preceq_\lambda d(x)\preceq_\lambda d(w)=f_{xz}\prec_\lambda f_{vw}$. However, by \cref{prp:prime}, $w$ is the dominator of $f_{vw}$; a contradiction to the fact that $(v,w)$ is a non-dominator edge.
\end{proof}

\smallskip\noindent\textbf{The edge-to-page assignment.} We embed all backward edges in page $p_0$, and all forward edges in page $p_1$. 
We next assign the remaining edges of $G$ to three sets $R^1$, $B^1$ and $G^1$, each containing $\lceil \frac{k}{2} \rceil$ pages. We process the intra-level faces of $\mathcal{F}$ according to $\lambda(\mathcal{F})$. Let $f$ be the next face to process. By \cref{lem:3-coloring}, face $f$ has a color in $\{r,b,g\}$.  The vertices of $f$ induce at most a $k$-clique $C_f$ in $G$. We assign the non-dominator edges of $C_f$ to the pages of one of the sets $R^1$, $B^1$ and $G^1$ depending on whether the color of $f$ is $r$, $b$, or $g$, respectively. This is possible since $C_f$ is at most a $k$-clique~\cite{DBLP:journals/jct/BernhartK79}. 
In the following, we prove that this assignment is valid, which is the main result of this section.

\begin{theorem}\label{th:2-level}
The book thickness of a two-level \framed{k} graph $G$ is at most $3\cdot\left\lceil \frac{k}{2} \right\rceil + 2$.
\end{theorem}
\begin{proof}
Consider two non-dominator edges $(v,w)$ and $(x,z)$, and assume without loss of generality that $v \prec_\rho w$ and $x \prec_\rho z$ in $\rho$. For a contradiction, assume $(v,w)$ and $(x,z)$ have been assigned to the same page $p$ and that either $v \prec_\rho x \prec_\rho w \prec_\rho z$ or $x \prec_\rho v \prec_\rho z \prec_\rho w$, i.e., $(v,w)$ and $(x,z)$ cross in the same page. By \cref{lem:backward,lem:forward}, $p \notin \{p_0,p_1\}$. Hence, $p\in R^1 \cup B^1 \cup G^1$. Let $f_{vw}$ and $f_{xz}$ be the two faces of $\mathcal{F}$ responsible for assigning $(v,w)$ and $(x,z)$ to one of the pages of $R^1 \cup B^1 \cup G^1$. Assume without loss of generality that $f_{vw} \prec_\lambda f_{xz}$. If $v$ and $w$ are $f_{vw}$-prime, and $x$ and $z$ are $f_{xz}$-prime, then by \cref{lem:magic-revised}, $(v,w)$ and $(x,z)$ cannot cross. Also, by \cref{lem:prime-xz}, we may assume that $x$ and $z$ are $f_{xz}$-prime. On the other hand, each of $v$ and $w$ can be $f_{vw}$-prime or not. In the following, we distinguish cases based on the relative order of $x$, $z$, $u$ and $w$ and on the types of the vertices $v$ and $w$.  

\begin{figure}[tb!]
	\centering
	\begin{subfigure}{.48\textwidth}
		\centering
		\includegraphics[width=.9\textwidth,page=14]{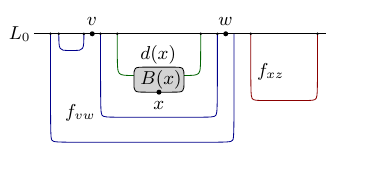}
		\subcaption{}
		\label{fig:2-level-1}
	\end{subfigure}
	\begin{subfigure}{.48\textwidth}
		\centering
		\includegraphics[width=.9\textwidth,page=15]{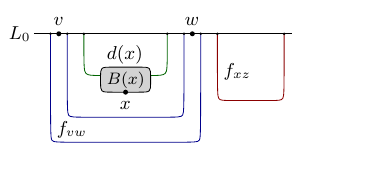}
		\subcaption{}
		\label{fig:2-level-2}
	\end{subfigure}
	\begin{subfigure}{.48\textwidth}
		\centering
		\includegraphics[width=.9\textwidth,page=16]{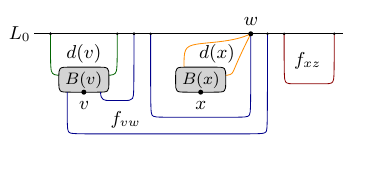}
		\subcaption{}
		\label{fig:2-level-3}
	\end{subfigure}
	\begin{subfigure}{.48\textwidth}
		\centering
		\includegraphics[width=.9\textwidth,page=17]{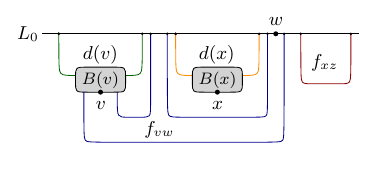}
		\subcaption{}
		\label{fig:2-level-4}
	\end{subfigure}
	\begin{subfigure}{.48\textwidth}
		\centering
		\includegraphics[width=.9\textwidth,page=18]{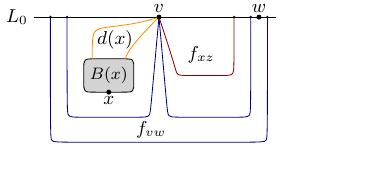}
		\subcaption{}
		\label{fig:2-level-5}
	\end{subfigure}
	\begin{subfigure}{.48\textwidth}
		\centering
		\includegraphics[width=.9\textwidth,page=19]{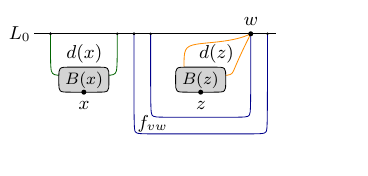}
		\subcaption{}
		\label{fig:2-level-6}
	\end{subfigure}
	\begin{subfigure}{.48\textwidth}
		\centering
		\includegraphics[width=.9\textwidth,page=20]{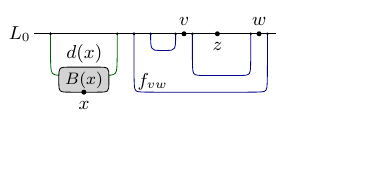}
		\subcaption{}
		\label{fig:2-level-7}
	\end{subfigure}
	\begin{subfigure}{.48\textwidth}
		\centering
		\includegraphics[width=.9\textwidth,page=21]{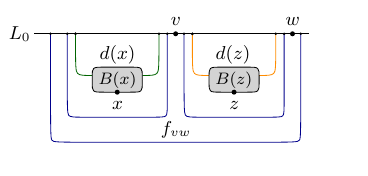}
		\subcaption{}
		\label{fig:2-level-8}
	\end{subfigure}
	\label{fig:2-level}
	\caption{%
		Illustrations for the proof of \cref{th:2-level}.}
\end{figure}

\medskip\noindent Assume first that the relative order of the vertices $x$, $z$, $u$ and $w$ is $v \prec_\rho x \prec_\rho w \prec_\rho z$. Since $x\prec_\rho z$, since both vertices are on the boundary of $f_{xz}$, and since $(x,z)$ is non-dominator, it follows that if $x$ belongs to $L_0$, then $z$ also belongs to $L_0$, in which case the order on $L_0$ is $\text{dom}(f_{xz})\prec_\rho x\prec_\rho w\prec_\rho z$. However, by \cref{lem:excluding}, this contradicts the fact that $f_{vw} \prec_\lambda f_{xz}$. Thus, $x$ necessarily belongs to $L_1$. Next, we distinguish cases based on the types of vertices~$v$~and~$w$. 

\begin{itemize}
\item[--] \emph{Vertex $v$ is not $f_{vw}$-prime}, which, by \cref{prp:prime-order}, implies that $w$ is also not $f_{vw}$-prime. Hence, both $v$ and $w$ belong to $L_0$, and as result $\text{dom}(f_{vw})\prec_\rho v\prec_\rho w$. Since $w \prec_\rho z$ and since $z$ is $f_{xz}$-prime, it follows that $w \preceq_\rho \text{dom}(f_{xz})$. By \cref{prp:discover-order} and since $v$ and $w$ belong to $L_0$, we get $v\preceq_\rho \text{dom}(d(x))\preceq_\rho w \preceq_\rho \text{dom}(f_{xz})$. If $\text{dom}(d(x))=w$, then $d(x)$ has to be small, since otherwise $w\prec_\rho x$. Therefore, regardless of whether $\text{dom}(d(x))=w$ or $\text{dom}(d(x))\prec_\rho w$ holds, are arise at a situation as the one illustrated in \cref{fig:2-level-1}; recall that $f_{vw}\prec_\lambda f_{xz}$. If $w=\text{dom}(f_{xz})$ holds, then $f_{xz}$ has to be small, as otherwise the planarity of $\sigma(G)$ is violated. However, the fact that $z\prec_\rho w$ contradicts the fact that $f_{xz}$ is small. Hence, $w \prec_\rho \text{dom}(f_{xz})$ must hold. In this case, $v$ cannot be on the boundary of the intra-level face $f_{xz}$ without violating the planarity of $\sigma(G)$, which is again a contradiction.
\item[--] \emph{Vertex $v$ is $f_{vw}$-prime and $w$ is not $f_{vw}$-prime}. First, we show that vertex $v$ belongs to $L_1$. To this end, we assume to the contrary that $v$ belongs to $L_0$. Since $v$ belongs to $L_0$ and $v$ is $f_{vw}$-prime, it follows that $v\neq \text{dom}(f_{vw})$. Since $w$ also belongs to $L_0$, we have $\text{dom}(f_{vw})\prec_\rho v\prec_\rho w$. We apply \cref{prp:prime-order,prp:dom-L0-L0} which yields $d(v)=f_{vw}\preceq_\lambda d(x)$. Observe that since $z$ is $f_{xz}$-prime, we have $w\preceq_\rho \text{dom}(f_{xz})$, as otherwise $z \prec_\rho w$, which is a contradiction. Similarly, if $f_{xz}$ is small, it follows again that $z\prec_\rho w$, which is the same contradiction. Hence, $f_{xz}$ cannot small. Hence, $f_{xz}$ follows $d(w)$ in a counterclockwise traversal of $w$ starting from $(u_{j-1},u_j)$ and ending at $(u_{j},u_{j+1})$ with $u_j=w$. Thus, we arise at a situation as the one illustrated in \cref{fig:2-level-2}, which shows that $x$ cannot be on the boundary of $f_{xz}$ without violating the planarity of $\sigma(G)$; a contradiction. Thus, $v$ belongs to $L_1$, as desired. Now all conditions of \cref{lem:pdl} for vertices $v$ and $x$ are satisfied, which implies that $f_{vw} \preceq_\lambda d(x)$. We are now ready to show that the $(f_{vw},f_{xz})$ exist in graph $\mathcal{C}(G)$, which completes the proof this case, since it also implies that $(u,v)$ and $(x,z)$ have been assigned to different pages. Assume for a contradiction that there exists no edge $(f_{vw},f_{xz})$ in o$\mathcal{C}(G)$. Since $x$ belongs to $L_1$, it follows that $f_{vw}\neq d(x)$. In total, we have $d(v) \preceq_\lambda f_{vw} \prec_\lambda d(x) \preceq_\lambda f_{xz}$. Since $x\prec_\rho w$, we have either that $\text{dom}(d(x))\prec_\rho w$ or that $\text{dom}(d(x))=w$ and $d(x)$ is small. If $d(x)$ is small, then we arise at a situation as the one illustrated in \cref{fig:2-level-3}. In order for $w\prec_\rho z$ to hold, either $w\prec_\rho \text{dom}(f_{xz})$ or $w=\text{dom}(f_{xz})$ and $f_{xz}$ is not small. However, in both cases face $f_{xz}$ violates the planarity of $\sigma(G)$; a contradiction. Thus, we may assume $\text{dom}(d(x))\prec_\rho w$, as illustrated in \cref{fig:2-level-4}. Since $w\prec_\rho z$ and since $z$ is $f_{xz}$-prime, we have $w\preceq_\rho \text{dom}(f_{xz})$. If equality holds, $f_{xz}$ cannot be small, since otherwise it follows that $z\prec_\rho w$. Hence, according to the definition of small faces, $f_{xz}$ follows $d(w)$ in a counterclockwise traversal of $w$ starting from $(u_{j-1},u_j)$ and ending at $(u_{j},u_{j+1})$ with $u_j=w$. Now, $f_{xz}$ cannot have vertex $x$ on its boundary without violating the planarity of $\sigma(G)$; a contradiction.
\item[--] \emph{Vertices $v$ and $w$ are $f_{vw}$-prime}. By \cref{lem:magic-revised}, the edge $(f_{vw},f_{xz})$ exists in $\mathcal{C}(G)$, which implies that $(u,v)$ and $(x,z)$ have been assigned to different pages.
\end{itemize}

\noindent Consider now the case, in which the relative order of $x$, $z$, $u$ and $w$ is $x \prec_\rho v \prec_\rho z \prec_\rho w$. We proceed as above by considering subcases based on the types of vertices $v$ and $w$. 

\begin{itemize}
\item[--] \emph{Vertex $v$ is not $f_{vw}$-prime}, which, by \cref{prp:prime-order}, implies that $w$ is also not $f_{vw}$-prime.
Hence, both $v$ and $w$ belong to $L_0$ and since $(v,w)$ is a non-dominator edge, we obtain $\text{dom}(f_{vw})\prec_\rho v\prec_\rho w$. Observe that by \cref{prp:dom-L0-L0}, vertex $v$ is discovered by $f_{vw}$. On the other hand, we have $f_{vw} \preceq_\lambda d(x)$ by \cref{prp:L0-v-L0}. We claim that $x$ belongs to $L_1$. Assume the contrary. Since $x$ precedes $z$ and both vertices are on the boundary of $f_{xz}$, it follows that $z$ also belongs to $L_0$. Therefore, all four vertices belong to $L_0$ and their order is $x \prec_\rho v \prec_\rho z \prec_\rho w$. Since $v$ and $w$ are on the boundary of $f_{vw}$ and $x$ and $z$ on the boundary of $f_{xz}$, the two faces $f_{vw}$ and $f_{xz}$ clearly violate the planarity of $\sigma(G)$. Thus, we may assume that $x$ belongs to $L_1$, as we initially claimed. We are now ready to show that the $(f_{vw},f_{xz})$ exist in graph $\mathcal{C}(G)$, which completes the proof this case. Assume for a contradiction that there exists no edge $(f_{vw},f_{xz})$ in o$\mathcal{C}(G)$. Since $x$ belongs to $L_1$, we have $f_{vw}\neq d(x)$. Therefore, we get $d(v)=f_{vw} \prec_\lambda d(x)$. In order for $x \prec_\rho v$ to hold, the dominator of $d(x)$ either precedes $v$ on $L_0$ or the dominator of $d(x)$ is $v$ and $d(x)$ is small. By applying the same arguments on vertices $x$ and $z$, we can similarly conclude that the dominator of $d(z)$ either precedes $w$ on $L_0$ or the dominator of $d(z)$ is $w$ and $d(z)$ is small. This gives rise to three subcases to consider.

\begin{itemize}
\item \emph{$d(x)$ is small and $\text{dom}(d(x))=v$}. Since $v\prec_\rho z$ and $z$ is $f_{xz}$-prime, we have $v\preceq \text{dom}(f_{xz})$. If $v=\text{dom}(f_{xz})$ holds, then $f_{xz}$ is not small since otherwise it follows that $z\prec_\rho v$; a contradiction.
Thus, $f_{xz}$ is not small. However, \cref{fig:2-level-5} shows that in this case the face $f_{xz}$ cannot have $x$ on its boundary without violating the planarity of $\sigma(G)$.

\item \emph{$d(z)$ is small and $\text{dom}(d(z))=w$}. Having ruled out the case above, we may further assume that $d(x)$ is not small. Since $d(x)$ is not small and since $x\prec_\rho v$, we get $\text{dom}(d(x))\prec_\rho v$. As illustrated in \cref{fig:2-level-6}, face $f_{xz}$ cannot have $x$ and $z$ on its boundary without violating the planarity of $\sigma(G)$; a contradiction.

\item \emph{Neither $d(x)$ nor $d(z)$ is small}. This yields $\text{dom}(d(x)) \prec_\rho v$ and $\text{dom}(d(z)) \prec_\rho w$ on $L_0$. We claim that $v\preceq_\rho \text{dom}(d(z))$. Assume the contrary, that is, $\text{dom}(d(z))\prec_\rho v$. Since $v\prec_\rho z$, vertex $z$ cannot be $d(z)$-prime and therefore $z$ belongs to $L_0$. We obtain the order $\text{dom}(d(x))\prec_\rho v\prec_\rho z\prec_\rho w$ on $L_0$. As shown in \cref{fig:2-level-7}, face $f_{vw}$ violates the planarity of $\sigma(G)$. Thus, we conclude that $\text{dom}(d(x)) \prec_\rho v \preceq_\rho \text{dom}(d(z)) \prec_\rho w$. With $f_{vw}\prec_\lambda d(x)$, we get the situation illustrated in \cref{fig:2-level-8}, in which face $f_{xz}$ violates the planarity of $\sigma(G)$.
\end{itemize}
\item[--] \emph{Vertex $v$ is $f_{vw}$-prime but $w$ is not $f_{vw}$-prime}. By \cref{lem:nonprime-vxv}, the edge $(f_{vw},f_{xz})$ exists in $\mathcal{C}(G)$, which implies that $(u,v)$ and $(x,z)$ have been assigned to different pages.
\item[--] \emph{Vertices $v$ and $w$ are $f_{vw}$-prime}. Again, by \cref{lem:magic-revised}, the edge $(f_{vw},f_{xz})$ exists in $\mathcal{C}(G)$, which implies that $(u,v)$ and $(x,z)$ have been assigned to different pages.
\end{itemize}
From the above case analysis, we can conclude that edges $(v,w)$ and $(x,z)$ cannot be assigned to the same, which concludes the proof.
\end{proof}

\subsection{Inductive step: Multi-level instances}
\label{ssec:multiple-levels}

\begin{figure}[tb!]
	\centering	\includegraphics[width=.9\textwidth,page=1]{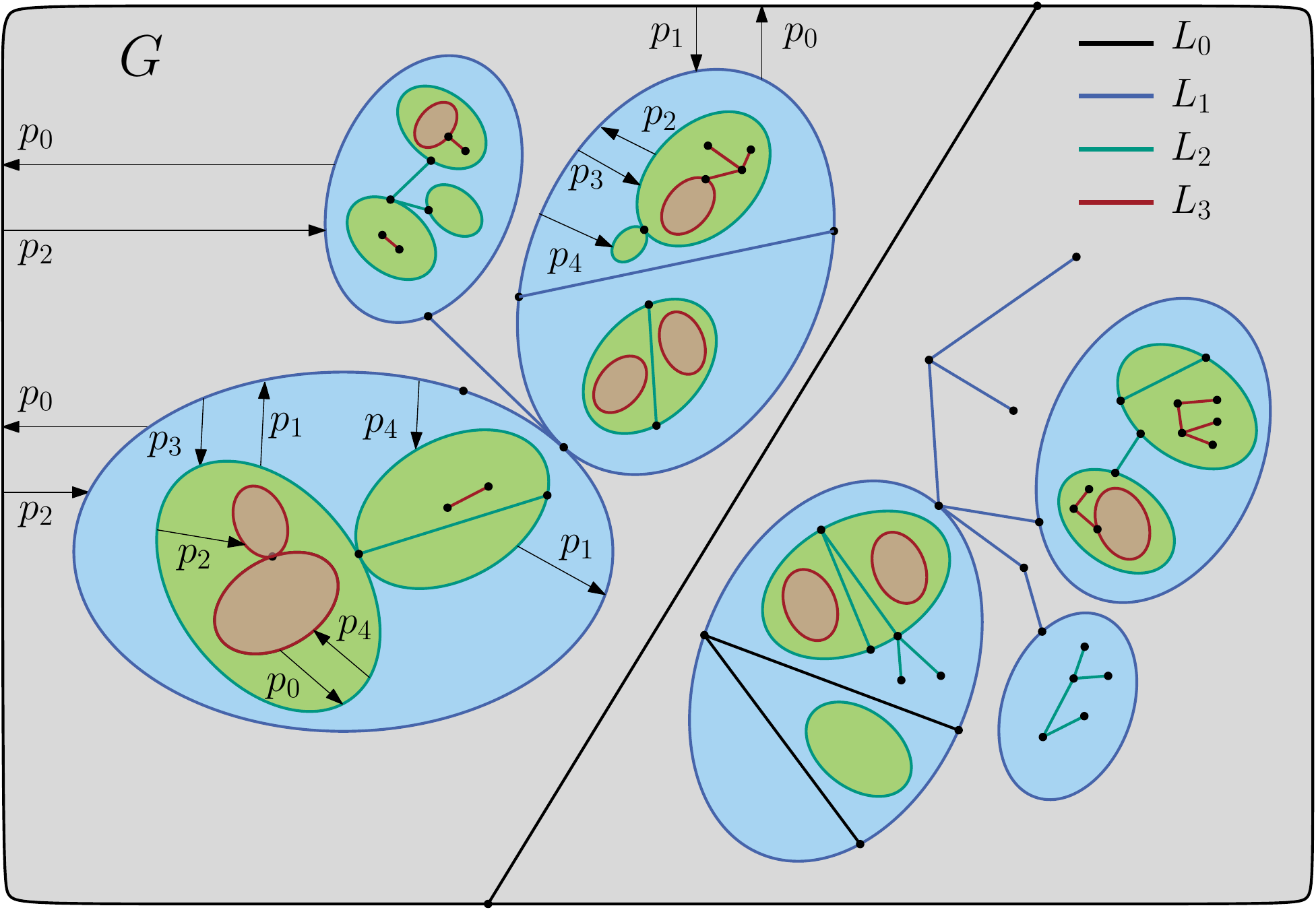}
	\caption{A multi-level instance $G$ with four levels of vertices, such that the bicomponents of $\hat{G}_2$ (which are shaded blue) form two connected components.
    Incoming edge and the two outgoing edges incident to the components are used to indicate page to which the backward edges and the the two sets of forward edges of each bicomponent are assigned, respectively.}
	\label{fig:multi-level-instance}
\end{figure}

In this section, we consider the general instances, which we call \emph{multi-level instances}, in which the input \framed{k} graph~$G$ consists of $q \ge 3$ levels $L_0,L_1,\ldots,L_{q-1}$. We refer to \cref{fig:multi-level-instance} for a schematic representation of a multi-level instance. Initially, we assume that the unbounded face of $\sigma(G)$ contains no crossing edges in its interior; we will eventually drop this assumption.
Recall that $G_i$ denotes the subgraph of $G$ induced by the vertices of $L_0 \cup \ldots \cup L_i$ containing neither chords of $\sigma_i(G)$ nor the crossing edges that are in the interior of the unbounded face of $\sigma(G)$.
We will further denote by $\hat{G}_i$ the subgraph of $G_i$ that is induced by the vertices of $L_{i-1} \cup L_i$ without the chords of $\sigma_{i+1}(G)$.
Observe that $\hat{G}_i$ is not necessarily connected; however, its maximal biconnected components, refered to as \emph{bicomponents} in the following, form two-level instances. To ease the description, we refer to the blocks of all bicomponents of $\hat{G}_i$ simply as the blocks of $\hat{G}_i$. In a book embedding of $G_i$, we say that two vertices of the level $L_j$ (with $j \leq i$) are \emph{sequential} if there is no other vertex of level $L_j$ between them along the spine. 
We say that a set $U$ of vertices of level $L_{j'}$ is \emph{$j$-delimited}, with $j' \neq j$, if either: (a) there exist two sequential vertices of level $L_{j}$ such that~all vertices of $U$ appear between them  along the spine, or (b) all vertices of $U$ are preceded or followed along the spine by all vertices of $L_j$.  

A book embedding~$\mathcal{E}_{i}$ of $G_i$ is \emph{good} if it satisfies the following properties\footnote{We stress at this point that even though Properties~P.\ref{prp:back-forward}, P.\ref{prp:forward} and P.\ref{prp:alter} might be a bit difficult to be parsed, they formalize the main idea of Yannakakis' algorithm for reusing the same set of pages in a book embedding. Notably, this formalization in the original seminal paper~\cite{DBLP:journals/jcss/Yannakakis89} is not present.}:

\begin{enumerate}[\bf {P.}1]
\item \label{prp:cyclicorder} The left-to-right order of the vertices on the boundary of each non-degenerate block $B$ of $\hat{G}_i$ in $\mathcal{E}_{i}$ complies with the order of these vertices in a counterclockwise (clockwise) traversal of the boundary of $B$, if $i$ is odd (even). 

\item \label{prp:vert-cons-1} All vertices of each block $B$ of $\hat{G}_i$, except possibly for its leftmost vertex, are consecutive and $(i-1)$-delimited.

\item \label{prp:block-adjacency-1} If between the leftmost vertex $\ell(B)$ of a block $B$ of $\hat{G}_i$ and the remaining vertices of $B$ there is a vertex $v$ of $L_{i}$ that belongs to a block $B'$ of $\hat{G}_i$ in the same connected component as $B$, such that the leftmost vertex $\ell(B')$ of $B'$ is to the left of $\ell(B)$, then $B$ and $B'$ share $\ell(B)$.

\item \label{prp:block-adjacency-2} Let $B$ and $B'$ be two blocks of $\hat{G}_i$ for which P.\ref{prp:block-adjacency-1} does not apply, and let $\ell(B)$ and $\ell(B')$ be their leftmost vertices. If $\ell(B)$ precedes $\ell(B')$, then either $\ell(B')$ precedes all remaining vertices of $B$ or all remaining vertices of $B'$ precede all remaining vertices~of~$B$.

\item \label{prp:vert-cons-2} For any $j \leq i-2$, all the vertices of each block of $\hat{G}_i$ are $j$-delimited.

\item \label{prp:pages} The edges of $G_i$ are assigned to $6\lceil k/2 \rceil + 5$ pages partitioned as  
(i) $P=\{p_0,\ldots,p_4\}$, and
(ii) $R^j = \{r^j_1,\ldots,r^j_{\lceil k/2 \rceil}\}$,
$B^j = \{b^j_1,\ldots,b^j_{\lceil k/2 \rceil}\}$,
$G^j = \{g^j_1,\ldots,g^j_{\lceil k/2 \rceil}\}$, $j\in\{0,1\}$. 

\item \label{prp:classification} The edges of ${G}_{i}$ are classified as backward, forward, or non-dominator in such a way that the following hold:

\begin{enumerate}[\bf\small a]
\item \label{prp:intra} For $\zeta \leq i$, the non-dominator edges of $\hat{G}_{\zeta}$ are assigned to $R^j \cup B^j \cup G^j$ with $j = \zeta \mod 2$.


\item \label{prp:leftmost} The edges that are incident to the leftmost vertex of a bicomponent of $\hat{G}_{i}$ and that are in its interior are backward.

\item \label{prp:back-forward} Let $\mathcal{B}_{i}$ be a bicomponent of $\hat{G}_{i}$. The backward edges of $\hat{G}_{i}$ in the interior of  $\mathcal{B}_{i}$ are assigned to a single page $b(\mathcal{B}_{i})$, while the forward edges are assigned to two pages $f_1(\mathcal{B}_{i})$ and $f_2(\mathcal{B}_{i})$ of $P$ different from $b(\mathcal{B}_{i})$; refer to \cref{fig:multi-level-instance}.

\item \label{prp:forward} Let $\mathcal{B}_{i-1}$ be a bicomponent of $\hat{G}_{i-1}$. 
The blocks $B_{i-1}^1,\ldots,B_{i-1}^\mu$ of $\mathcal{B}_{i-1}$ are the boundaries of several bicomponents of $\hat{G}_i$. Then, the forward edges of $\hat{G}_{i-1}$ incident to $B_{i-1}^j$, with $j=1,\ldots,\mu$, are either all assigned to $f_1(\mathcal{B}_{i-1})$ or to $f_2(\mathcal{B}_{i-1})$.
%

\item \label{prp:alter} Let $\langle p'_0, \ldots , p'_4 \rangle $ be a permutation of $P$. 
Assume that the backward edges of $\hat{G}_{i-2}$ that are in the interior of a bicomponent $\mathcal{B}_{i-2}$ of $\hat{G}_{i-2}$ have been assigned to $p'_0$ (in accordance with P.\ref{prp:back-forward}), while the forward edges of $\hat{G}_{i-2}$ that are in the interior of $\mathcal{B}_{i-2}$ have been assigned to $p'_1$ and $p'_2$ (in accordance to P.\ref{prp:back-forward} and P.\ref{prp:forward}). 
The blocks of $\mathcal{B}_{i-2}$ are the boundaries of several bicomponents $\mathcal{B}_{i-1}^1,\ldots,\mathcal{B}_{i-1}^\mu$ of $\hat{G}_{i-1}$. 
Consider now a bicomponent $\mathcal{B}_{i-1}^j$ with $1 \leq j \leq \mu$ of $\hat{G}_{i-1}$.
Assume w.l.o.g.~that the forward edges of $\mathcal{B}_{i-2}$ incident to $\mathcal{B}_{i-1}^j$ are assigned to $p'_1$.  
%
%
Then, the backward edges of $\mathcal{B}_{i-1}^j$ (which are incident to its blocks, and thus to the bicomponents of $\hat{G}_i$) are assigned to $p'_2$, while its forward edges to $ p_3'$ and $p_4'$. 
\end{enumerate}
\end{enumerate}

\noindent We next argue that the book embeddings computed by the algorithm of \cref{ssec:two-levels}  can be easily adjusted to become good.

\begin{lemma}\label{lem:good-two-levels}
Any two-level instance admits a good book embedding. 
\end{lemma}
\begin{proof}
To prove the lemma, we show that the book embedding $\mathcal{E}$ of a two-level instance $G$ computed by the algorithm of \cref{ssec:two-levels} can be slightly modified to satisfy the properties of a good book embedding. 
We first observe that Properties~P.\ref{prp:vert-cons-2} and P.\ref{prp:alter} are clearly satisfied, since $G$ consists of only two levels. Regarding the remaining properties, we argue as follows. 
Property~P.\ref{prp:cyclicorder} holds by construction. 
Property~P.\ref{prp:vert-cons-1} directly follows from \cref{prp:consecutively} of \cref{sse:linearorder} and Rule~\ref{r:1} of the constructed linear order. 
Property~P.\ref{prp:block-adjacency-1} follows from \cref{prp:ordering-connected} of \cref{sse:linearorder}. 
Property~P.\ref{prp:block-adjacency-2} follows from \cref{prp:ordering-disconnected} of \cref{sse:linearorder}. 
Property~P.\ref{prp:pages} follows from the page assignment described in \cref{sse:assignment}; 
in particular, since $G$ consists of only two levels, its backward edges can be assigned to page $p_0$ by \cref{lem:backward}, while its non-dominator edges can be assigned to pages in $R^1 \cup B^1 \cup G^1$. 
Hence, Property~P.\ref{prp:intra} holds.
Property~P.\ref{prp:leftmost} holds by the definition of backward edges.  
Finally, as already discussed, the backward edges of $G$ are assigned to a single page $p_0$ of $P$ in $\mathcal{E}$. Further, by \cref{lem:forward} all forwards edges of $G$ can be embedded in a single page of $P$ in $\mathcal{E}$. However, in order to satisfy Property~P.\ref{prp:forward}, we reassign the forward edges to two pages of $P$ in $\mathcal{E}$ as follows. Assume that each connected component of the blocks of $G$ is rooted at the degenerate block corresponding to its first vertex. We assign the forward edges towards the blocks that are at odd (even) distance from such a root block to $p_1=f_1(\mathcal{B}_{1})$ ($p_2=f_2(\mathcal{B}_{1})$, resp.) of $P$, where $G=\mathcal{B}_1$.
\end{proof}

\noindent Finally, the next lemma deals with good book embeddings of multi-level instances. 

\begin{lemma}\label{lem:good-multi-levels}
Any multi-level instance admits a good book embedding.
\end{lemma}
\begin{proof}
Assume that we have recursively computed a good book embedding $\mathcal{E}_i$ of $G_i$. We next show how to extend $\mathcal{E}_i$ to a good book embedding $\mathcal{E}_{i+1}$ of $G_{i+1}$. Note that $G_{i+1}$ is the union of $G_i$ and $\hat{G}_{i+1}$, which share the vertices of $L_i$ and the edges of $C_i(G)$. 

Consider the set $\mathcal{H}$ of bicomponents $\mathcal{B}_1,\ldots,\mathcal{B}_\chi$ of $\hat{G}_{i+1}$. As already mentioned, each of the bicomponents in $\mathcal{H}$ forms a two-level instance. 
Consequently, the vertices delimiting the unbounded faces of $\mathcal{B}_1,\ldots,\mathcal{B}_\chi$ form blocks $B_1,\ldots,B_\chi$ of $\hat{G}_i$, which in turn form a set of cacti in $\sigma_i(G)$. We assume that each connected component in this set is rooted at one of its blocks. This allows as to associate each bicomponent $\mathcal{B}_i$ out of the initial ones with a root bicomponent denoted by $r(\mathcal{B}_i)$, $i=1,\ldots,\chi$. This further allows us to also associate each bicomponent $\mathcal{B}_i$ with a parity bit $\epsilon(\mathcal{B}_i)$ that expresses whether the distance between $\mathcal{B}_i$ and $r(\mathcal{B}_i)$ is odd or even.

We process the bicomponents of $\mathcal H$ one by one as follows.
Assume now that we have processed the first $x-1 < \chi$ bicomponents $\mathcal{B}_1,\ldots,\mathcal{B}_{x-1}$ of $\mathcal{H}$ and that we have extended $\mathcal{E}_{i}$ to a good book embedding $\mathcal{E}_i^{x-1}$ of $G_i$ together with $\mathcal{B}_1,\ldots,\mathcal{B}_{x-1}$. Consider the next bicomponent $\mathcal{B}_x$ of $\hat{G}_{i+1}$ in $\mathcal{H}$. Observe that the boundary of $\mathcal{B}_x$ is a simple cycle consisting of vertices of level $L_i$. As a result, the vertices and the edges of this cycle are present in $G_i$ and therefore they have been embedded in $\mathcal{E}_{i}$ and thus in $\mathcal{E}_{i}^{x-1}$. 

In the following, we show how to extend $\mathcal{E}_{i}^{x-1}$ to a good book embedding $\mathcal{E}_{i}^{x}$ of $G_i$ together with $\mathcal{B}_1,\ldots,\mathcal{B}_x$. Once all blocks in $\mathcal{H}$ have been processed, the obtained book embedding $\mathcal{E}_{i}^{\chi}$ is the desired good book embedding $\mathcal{E}_{i+1}$ of $G_{i+1}$. 
The vertices that delimit the unbounded face of $\mathcal{B}_x$ form a block $B_x$ of $\hat{G}_i$. By Property~P.\ref{prp:cyclicorder}, their left to right order in $\mathcal{E}_{i}^{x-1}$ (say $u_0,\ldots,u_{s-1}$) complies with the order in which these vertices appear in either a counterclockwise or in a clockwise traversal of the boundary of $B_x$, depending on whether if $i$ is odd or even, respectively. We proceed by computing a good book embedding $\mathcal{E}_x$ of $\mathcal{B}_x$ which exists by \cref{lem:good-two-levels}, such that the left-to-right order of the vertices of $\mathcal{B}_x$ is $u_0,\ldots,u_{s-1}$ in $\mathcal{E}_x$. Note that this can be achieved by flipping $\mathcal{B}_x$, if $i$ is even. Further, note that $\mathcal{E}_{x}$ is good by \cref{lem:good-two-levels}. We extend $\mathcal{E}_{i}^{x-1}$ to a good book embedding $\mathcal{E}_{i}^{x}$ in two steps as follows. 

In the first step, for $j=0,1,\ldots, s-2$, the vertices of $\mathcal{B}_x$ that appear between $u_j$ and $u_{j+1}$ in $\mathcal{E}_{x}$, if any, are embedded right before $u_{j+1}$ in $\mathcal{E}_{i}^{x-1}$ in the same left-to-right order as in $\mathcal{E}_{x}$; also, the vertices of $\mathcal{B}_x$ that appear after $u_{s-1}$ in $\mathcal{E}_{x}$, if any, are embedded right after $u_{s-1}$ in $\mathcal{E}_{i}^{x-1}$ in the same left-to-right order as in $\mathcal{E}_{x}$.  Let $\mathcal{E}_{i}^{x}$ be the resulting embedding (which still does not contain all the edges of $\mathcal{B}_x$). Since $\mathcal{E}_{x}$ is a good book embedding and since we do not change relative order of the vertices of $\mathcal{B}_x$ in $\mathcal{E}_{x}$ and in $\mathcal{E}_{i}^{x}$, Properties~P.\ref{prp:cyclicorder} and P.\ref{prp:vert-cons-1} hold for $\mathcal{E}_{i}^{x}$. 
Since Property~P.\ref{prp:vert-cons-1} holds for block $B_x$ in $\mathcal{E}_i^{x-1}$, it follows that there is no vertex of level $L_j$, with $j\leq i-1$, in $\mathcal{E}_{i}^{x-1}$ between any two vertices of $\{u_1,\ldots,u_{s-1}\}$. This and the fact that we have placed the remaining vertices of $\mathcal{B}_x$ either right before or right after any of $u_1,\ldots,u_{s-1}$ implies that there exists no vertex of level $L_{j}$, with $j \leq i-1$, between the vertices of $\mathcal{B}_x$ along the spine, which proves  Property~P.\ref{prp:vert-cons-2} for $\mathcal{E}_{i}^{x}$.

In the second step, we assign the internal edges of $\mathcal{B}_x$ to the already existing pages of $\mathcal{E}_{i}^{x}$ to complete the embedding, which also implies that Property~P.\ref{prp:pages} will not be deviated. 
This step will complete the extension of $\mathcal{E}_{i}^{x-1}$ to $\mathcal{E}^x_i$. The assignment is done in a straight-forward manner. 
The backward, forward, and non-dominator edges of $\mathcal{E}_x$  that are internal in $\mathcal{B}_x$ will be classified as backward, forward, and non-dominator, respectively, also in $\mathcal{E}_{i}^{x}$, which guarantees Property~P.\ref{prp:classification}. 
To guarantee that Property~P.\ref{prp:intra} holds for $\mathcal{E}_x$, we proceed as follows. The non-dominators edges of $\mathcal{E}_x$ that are internal in $\mathcal{B}_x$ and are assigned to $r^1_1,\ldots,r^1_{\lceil k/2 \rceil}$, $b^1_1,\ldots,b^1_{\lceil k/2 \rceil}$, $g^1_1,\ldots,g^1_{\lceil k/2 \rceil}$ in $\mathcal{E}_x$ are assigned to $r^j_1,\ldots,r^j_{\lceil k/2 \rceil}$, $b^j_1,\ldots,b^j_{\lceil k/2 \rceil}$, $g^j_1,\ldots,g^j_{\lceil k/2 \rceil}$ in $\mathcal{E}_{i}^{x}$, respectively, where $j=i+1 \mod 2$. Hence, Property~P.\ref{prp:intra} holds for $\mathcal{E}_x$, as desired. 

We now show that no two edges assigned to any of these pages cross. Assume for a contradiction that there is a crossing in page $p \in R^j \cup B^j \cup G^j$ with $j=i+1 \mod 2$. Since $\mathcal{E}_{i}^{x-1}$ is a good book embedding, this crossing must necessarily involve an edge $e$ of $\mathcal{B}_x$. Let $e'$ be the second edge involved in the crossing. We distinguish two cases: $(i)$~$e'$ belongs to one of $\mathcal{B}_1,\ldots,\mathcal{B}_x$, and $(ii)$ $e'$ belongs to some previously embedded graph $\hat{G}_\zeta$ with $\zeta < i+1$. In Case~$(i)$, we first observe that $e'$ cannot belong to $\mathcal{B}_x$, as otherwise $e$ and $e'$ would also cross in $\mathcal{E}_x$, contradicting the fact that $\mathcal{E}_x$ is a good book embedding of $\mathcal{B}_x$. Hence, we may assume that $e'$ belongs to $\mathcal{B}_j$ with $j<x$. Since $e \in \mathcal{B}_x$ and $e'\in \mathcal{B}_j$, by Property~P.\ref{prp:vert-cons-1}, at least one of $e$ and $e'$ must be incident to the leftmost vertex of the blocks $B_x$ and $B_j$ that delimit the unbounded faces of $\mathcal{B}_x$ and $\mathcal{B}_j$, respectively, which, by Property~P.\ref{prp:leftmost}, implies that at least one of them is backward; a contradiction. Consider now Case~$(ii)$ and recall that in this case $e'$ belongs to some graph $\hat{G}_\zeta$ with $\zeta < i+1$. Since $e$ and $e'$ cross in $p$, it follows that $\zeta \equiv i+1 \mod 2$. The latter property further implies that $\zeta \leq i-1$. In this case, however, Property~P.\ref{prp:vert-cons-2} implies the endpoint of edge $e$ are $(i-1)$-delimited, which in turn implies that $e$ and $e'$ nest, which contradicts our initial assumption.

By \cref{lem:good-two-levels}, all backward edges of $\mathcal{E}_x$ have been assigned to page $p_0$ in $\mathcal{E}_x$, while its forward edges have been assigned to $p_1$ and $p_2$; also, recall that no edge of $\mathcal{E}_x$ has been assigned to pages $p_3$ and $p_4$. To guarantee Property~P.\ref{prp:back-forward} in $\mathcal{E}_{i}^{x}$, the backward edges of $\mathcal{E}_x$ that are interior to $B_x$ will be assigned to $\mathcal{E}_{i}^{x}$ to a common page $p$ of $P$ (i.e., not necessarily to $p_0$), while the corresponding forward edges assigned to $p_1$ and $p_2$ in $\mathcal{E}_x$ will be reassigned to two pages $f_1$ and $f_2$, respectively.

To determine pages $p$, $f_1$ and $f_2$, we have to take into account Properties P.\ref{prp:forward} and P.\ref{prp:alter} that hold for $\mathcal{E}_{i}^{x-1}$. Assume first that $i \geq 3$; the case $i=2$ is immediate. Then, there is a bicomponent $\mathcal{B}_{i-2}$ of $\hat{G}_{i-2}$, whose boundary vertices form a cycle that, in $G_{i+1}$,  contains the bicomponent $\mathcal{B}_x$ in its interior. Assume w.l.o.g.\ that the backward edges of $\mathcal{B}_{i-2}$ are assigned to page $p'_0 \in P$, in accordance to P.\ref{prp:back-forward}. It follows by P.\ref{prp:alter} that we may further assume w.l.o.g.\ that all the backward edges of the bicomponents of $\hat{G}_{i-1}$, whose boundaries are blocks of $\mathcal{B}_{i-2}$, have been assigned to pages $p'_1$ and $p'_2$ different from $p'_0$. 
Assume also, w.l.o.g., that the forwards edges of $\mathcal{B}_{i-2}$ incident to $\mathcal{B}_x$ have been assigned to $p'_1$.
By Property P.\ref{prp:alter}, this implies that the backward (forward) edges of bicomponent $\mathcal{B}_x$ must be assigned to page $p'_2$ (to $p'_3$ and $p'_4$, respectively). Note that also of all the previously processed bicomponents of $\hat{G}_{i+1}$ in $\mathcal{H}$ make use of these three pages plus the page $p'_1$. 
Hence, both Properties~P.\ref{prp:back-forward} and P.\ref{prp:alter} are satisfied. 
%
The choice between the two pages $p_3'$ and $p_4'$ is done based on the parity bit $\epsilon(\mathcal{B}_x)$, so that, all forward edges of all bicomponents in $\mathcal{H}$ having the same parity bit will be assigned to the same page in $\{p'_3,p'_4\}$, thus guaranteeing that Property~P.\ref{prp:back-forward} holds for $\mathcal{E}_i^x$.

We conclude the proof by showing that no two edges assigned to pages in $\{p'_2,p'_3,p'_4\}$ cross in $\mathcal{E}_i^x$. We first focus on page $p'_2$. Clearly, no two edge in $p'_2$ belonging to $\mathcal{B}_x$ can cross, since $\mathcal{E}_x$ is a good book embedding. Hence, if there is a crossing in $p'_2$ it must involve an edge $e$ in $\mathcal{B}_x$ and an edge $e'$ of either $G_i$ or of one of the previously embedded bicomponents of $\hat{G}_{i+1}$ in $\mathcal{H}$. We first consider the case, in which $e'$ belongs to $G_i$. In particular, by Property~\ref{prp:vert-cons-2} since all the vertices of $B_x$ are $(i-2)$-delimited, it follows that $e'$ is an edge of $\hat{G}_i$. By Property~P.\ref{prp:alter}, $e'$ must be incident to the leftmost vertex of $B_x$. Now, observe that the edges of $\mathcal{B}_x$ that are incident to the leftmost vertex of $\mathcal{B}_x$ in $\mathcal{E}_x$ are by definition backward; thus, they are not assigned to $p'_2$. Since by Property~P.\ref{prp:vert-cons-1} the remaining vertices of $\mathcal{B}_x$ are $(i-1)$-delimited, it follows that if there exists a crossing in page $p_2$, this should involve a previously embedded bicomponent of $\hat{G}_{i+1}$ in $\mathcal{H}$. As a result, we can assume that $e'$ belongs to $\mathcal{B}_j$, with $j<x$. 
Let $B_x$ and $B_j$ be the blocks that delimit the unbounded faces of $\mathcal{B}_x$ and $\mathcal{B}_j$, respectively. Since $e \in \mathcal{B}_x$ and $e'\in \mathcal{B}_j$, by P.\ref{prp:block-adjacency-2}, it follows that $B_x$ and $B_j$ belong to the same connected component $C$ formed by the blocks of $\hat{G}_i$. By Property~P.\ref{prp:vert-cons-1}, at least one of $e$ and $e'$ must be incident to the leftmost vertex of $B_x$ or $B_j$ in $\mathcal{E}_{i}^{x}$, respectively. Since $B_x$ and $B_j$ belongs to $C$, by Property~P.\ref{prp:block-adjacency-1}, $B_x$ and $B_j$ must share a common vertex, which implies that $\mathcal{B}_x$ and $\mathcal{B}_j$ have different parity bits, i.e. $\epsilon(\mathcal{B}_x) \neq \epsilon(\mathcal{B}_j)$. However, since $e$ is assigned to $p'_2$, edge $e'$ is assigned to $p_1'$, contradicting our assumption that $e$ and $e'$ cross. Hence, we can conclude that there is no two edges assigned to $p'_2$ that cross in $\mathcal{E}_i^x$.

We now focus on the edges of $\{p'_3,p'_4\}$. Assume w.l.o.g.\ that $e$ is assigned to $p'_3$. As above, we argue that $e'$ either belongs to $G_i$ (in particular, to $\hat{G}_i$) or to one of the previously embedded bicomponents of $\hat{G}_{i+1}$ in $\mathcal{H}$. The former case is actually not possible, since by Property~\ref{prp:alter} there is no edge of $\hat{G}_i$ assigned to $p_3'$ that is incident to $\mathcal{B}_x$. So, we may focus on the latter case, in which $e'$ belongs to $\mathcal{B}_j$, with $j<x$. As above, we can conclude that $\mathcal{B}_x$ and $\mathcal{B}_j$ should belong to the same connected component $C$ formed by the blocks of $\hat{G}_i$, and in particular, the corresponding blocks $B_x$ and $B_j$ that delimit their unbounded faces share a common vertex, which implies that $\mathcal{B}_x$ and $\mathcal{B}_j$ have different parity bits. In this case, however, the involved edges $e$ and $e'$ are assigned to $p'_1$ and $p_2'$, and thus they cannot cross in $p_3'$.

From the discussion above, we can conclude that $\mathcal{E}_{i}^{x}$ is in fact a good book embedding.  However, recall that we initially assumed that the unbounded face of $\sigma(G)$ contains no crossing edges in its interior, to support the recursive strategy. We complete the proof by dropping this assumption as follows. We assign these edges to the pages of $R^0 \cup B^0 \cup G^0$, which results in a good book embedding of $G$, since the endvertices of the edges already assigned to these pages are $0$-delimited.  
\end{proof}

\noindent Altogether, \cref{lem:good-multi-levels} in conjunction with \cref{lem:good-two-levels} completes the proof of \cref{th:main}.

\section{Application of \cref{th:main} to map graphs}\label{sec:map}

We begin by formally defining map graphs (refer also to~\cite{DBLP:journals/jacm/ChenGP02}). A \emph{map graph} $G$ is one that admits a \emph{map} $\mathcal M$, i.e., a bijection that puts in correspondence each vertex $v$ of $G$ with a region $\mathcal M(v)$ of the sphere homeomorphic to a closed disk, called \emph{nation}, in such a way that the following properties hold: (i) the interiors of any two distinct nations are disjoint, and (ii)  two vertices $u$ and $v$ are adjacent in $G$  if and only if the boundaries of $\mathcal M(u)$ and $\mathcal M(v)$ intersect. The points of the sphere that are not covered by any nation fall into open connected regions; the closure of each such region is a \emph{hole} of $\mathcal M$. A \emph{\map{k}} graph (with $k>1$) is a graph that admits a map $\mathcal M$ such that at most $k$ nations intersect in a single point. Also, $G$ is \emph{well-formed} if for every edge $(u,v)$ of $G$ the intersection of $\mathcal M(u)$ and $\mathcal M(v)$ is either a single point or a single curve segment. Moreover, if $\mathcal M$ does not contain holes, $G$ is a \emph{\hfmap{k}}.

In order to prove \cref{co:map}, we first deal with a simpler case. Namely, we prove that well-formed \hfmap{k} graphs are \framed{k}, which, by \cref{th:main}, implies they have book thickness at most $6\lceil \frac{k}{2} \rceil + 5$.

\begin{figure}[tb!]
		\centering
		\begin{subfigure}{.32\textwidth}
			\centering
			\includegraphics[width=.9\textwidth,page=1]{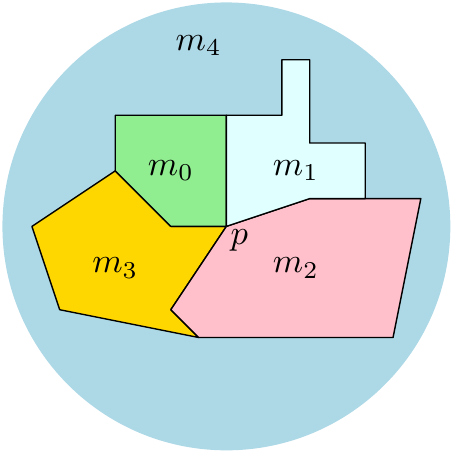}
			\subcaption{$\mathcal M$}
		\end{subfigure}
		\begin{subfigure}{.32\textwidth}
			\centering
			\includegraphics[width=.9\textwidth,page=2]{figures/map}
			\subcaption{$\mathcal M'$}
		\end{subfigure}
		\begin{subfigure}{.32\textwidth}
			\centering
			\includegraphics[width=.9\textwidth,page=3]{figures/map}
			\subcaption{$\Gamma'$}
		\end{subfigure}
		\caption{%
			Illustration for the proof of \cref{lem:wellformed}. (a) A well-formed \hfmap{4} graph $G$ with one $h$-point with $h>3$, denoted by $p$. (b) The well-formed \map{3} graph $G'$ obtained by deleting $p$. (c) A drawing of the \skeleton of $G'$. }
		\label{fig:map}
	\end{figure}

\begin{lemma}\label{lem:wellformed}
Every well-formed \hfmap{k} graph is \framed{k}.
\end{lemma}
\begin{proof}
Let $\mathcal M$ be a well-formed \hfmap{k} of a graph $G$ and refer to \cref{fig:map} for an illustration. A point $p$ of $\mathcal M$ is an \emph{$h$-point}, if $h>1$ nations intersect in $p$. Let $p$ be an $h$-point of $\mathcal M$ (if any) with $4 \le h \le k$. The operation of \emph{deleting} the $h$-point $p$ works as follows. Denote by $V_p = \{m_0, m_1, \dots, m_{h-1}\}$, the set of $h$ nations that intersect in $p$. Consider now a small open disk $D$ in $\mathcal M$ centered at $p$ such that any point in $D$ is either $p$, an interior point of a nation in $V_p$, or a point where exactly two nations of $V_p$ intersect. We shall indeed assume that, excluding point $p$ and up to a relabeling of the nations, the only adjacencies realized in $D$ are those between $m_i$ and $m_{i+1}$, for $i=0,1,\dots,h-1$ (indices taken modulo $h$). Clearly, for a sufficiently small radius, such disk always exists. Removing the parts of the nations in $D$ from $\mathcal M$ introduces a hole in the map, removes $p$, and does not introduce any new $h'$-point with $h'>3$. Let $\mathcal M'$ be the well-formed \map{3} obtained by deleting all $h$-points of $\mathcal M$ with $h>3$, and let $G'$ be the corresponding map graph. We aim at proving that $G$ admits a \framed{k} drawing $\Gamma$ having $G'$ as \skeleton. 

First of all note that $G'$ is simple, because $\mathcal M'$ is well-formed, and spanning, because we do not destroy any nation. Recall that, by definition of well-formed, any two adjacent vertices $u$ and $v$ are such that $\mathcal M'(u)$ and $\mathcal M'(v)$ intersect in either a single point $p_{uv}$ or in a curve segment $s_{uv}$. In the latter case, we denote by $p_{uv}$ an arbitrary $2$-point of $\mathcal M'$ along $s_{uv}$.

Let $\Gamma'$ be a drawing of $G'$ obtained by representing each vertex $u$ as an interior point $p_u$ of $\mathcal M'(u)$, and each edge $(u,v)$ as a Jordan arc that starts a $p_u$, traverses $\mathcal M'(u)$ until $p_{uv}$, and finally traverses  $\mathcal M'(v)$ ending in $p_v$.  The circular order of the edges around a vertex $u$ is kept the same as the circular order of the corresponding points $p_{uv}$ around $\mathcal M'(u)$; this ensures that no two arcs intersect in an interior point of a nation. On the other hand, the only point where two Jordan arcs may intersect is a $3$-point (if it exists). In such a case it suffices to slightly perturb the curves around such $h$-point so to avoid any crossing. Thus $\Gamma'$ does not contain any crossing. Note that $\Gamma'$ is a spherical drawing, in what follows we consider its stereographic projection onto the plane, i.e., we view $\Gamma'$ as a planar drawing. Concerning the size of the largest face of $G'$, observe that the maximum degree of a face of $G'$ (including the unbounded face) cannot be larger than the greatest number of nations that intersect the same hole of $\mathcal M'$, which is at most $k$ by construction. 

It remains to prove that: (a) $G'$ is biconnected, and (b) all edges of $\overline{E} = G \setminus G'$ can be drawn entirely inside faces of $G'$ and are all crossed.

Concerning (a), if there existed a vertex $v$ whose removal disconnects $G'$, this would imply that the original map $\mathcal M'$ contains a hole that intersects $\mathcal M'(v)$ at least twice and (at least) two nations whose corresponding vertices in $G'$ are connected only by paths containing $v$. (Recall that any nation is homeomorphic to a closed disk, i.e., it intersects neither holes nor other nations in its interior.) However, such hole does not exist, because by construction any hole of $\mathcal M'$ intersects a set of nations $m_0, m_1, \dots, m_{h-1}$ whose induced graph contains a cycle.


Concerning (b), recall that any edge $(u,v)$ in $\overline{E}$ connects two nations that intersect in $\mathcal M$ and do not intersect anymore in $\mathcal M'$. In particular, there exists at least one hole in $\mathcal M'$ intersecting $\mathcal M'(u)$ and $\mathcal M'(v)$. When constructing $\Gamma'$ from $\mathcal M'$, such hole yields a face in $\Gamma'$ having both $u$ and $v$ on its boundary. Thus, we can draw a copy of $(u,v)$ inside each such face. This concludes the proof.
\end{proof}

\noindent The next theorem extends the proof of \cref{lem:wellformed} and, together with \cref{th:main}, implies \cref{th:map}.

\begin{theorem}\label{th:map}
Every \map{k} graph is \pframed{2k}.
\end{theorem}
\begin{proof}
\begin{figure}[tb!]
		\centering
		\begin{subfigure}{.32\textwidth}
			\centering
			\includegraphics[width=.9\textwidth,page=4]{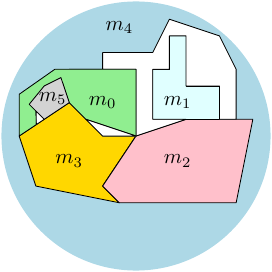}
			\subcaption{}
		\end{subfigure}
		\begin{subfigure}{.32\textwidth}
			\centering
			\includegraphics[width=.9\textwidth,page=5]{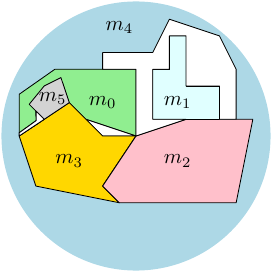}
			\subcaption{}
		\end{subfigure}
		\begin{subfigure}{.32\textwidth}
			\centering
			\includegraphics[width=.9\textwidth,page=6]{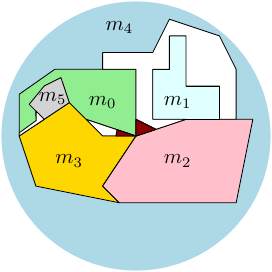}
			\subcaption{}
		\end{subfigure}
		\begin{subfigure}{.32\textwidth}
			\centering
			\includegraphics[width=.9\textwidth,page=7]{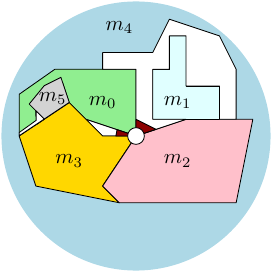}
			\subcaption{}
		\end{subfigure}
		\begin{subfigure}{.32\textwidth}
			\centering
			\includegraphics[width=.9\textwidth,page=8]{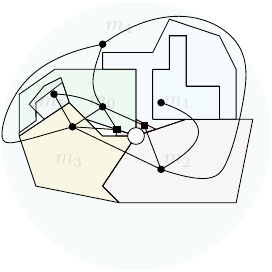}
			\subcaption{}
		\end{subfigure}
		\begin{subfigure}{.32\textwidth}
			\centering
			\includegraphics[width=.9\textwidth,page=9]{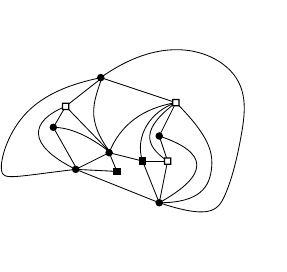}
			\subcaption{}
		\end{subfigure}
		\caption{%
			Illustration for the proof of \cref{th:map}. (a) A \map{3} graph that is not well-formed, in particular $m_0$ and $m_3$ have two distinct adjacencies. (b) Removing multiple adjacencies by introducing holes. The resulting graph is a well-formed \map{3} graph $G$ with only one $h$-point with $h=3$ that intersects two holes. (c) The well-formed \map{5} graph $G'$ obtained by adding a dummy nation such that no $h$-point with $h>2$ intersects a hole. (d) Deleting $h$-points with $h>3$. (e) A drawing of the \skeleton $\Gamma'$; vertices representing dummy nations are black squares. (f) Augmenting the \skeleton to make it biconnected; the inserted dummy vertices are white squares. }
		\label{fig:map-2}
	\end{figure}
Refer to \cref{fig:map-2} for an illustration.
Let $\mathcal {M}_0$ be a \map{k} of a graph $G$. 
We first aim at turning $\mathcal {M}_0$ into a \emph{nearly well-formed} \map{k} $\mathcal M$ of $G$, i.e., a \map{k} in which multiple adjacencies occur only in presence of $h$-points with $h>2$
Recall that each intersection between two nations is either a single $h$-point ($h \ge 2)$ or a curve segment. 
If any two nations intersect at most once, then $\mathcal M = \mathcal {M}_0$. 
Else, let $m$ and $m'$ be two nations that intersect $r>1$ times, and consider any intersection that is a $2$-point or a curve segment, excluding from this segment possible $h$-points with $h>2$. We can remove each such intersection between $m$ and $m'$ by locally retracting $m$ (or $m'$). Such operation introduces a hole (which can possibly merge with some other holes) in place of the intersection between $m$ and $m'$ and does not destroy any other intersection because we avoided $h$-points with $h>2$. The resulting  \map{k} $\mathcal M$ may not be well-formed yet, however it is nearly well-formed.

Under this assumption, the proof of \cref{lem:wellformed} can be adjusted as follows.  
Observe that an $h$-point $p$ in $\mathcal M$ touches at most $h$ holes. 
For each $h$-point $p$ in $\mathcal M$ that touches  $\chi \le h$ holes and with $h>2$, we introduce a sufficiently small \emph{dummy} nation such that $p$ becomes an $(h+\chi)$-point that does not touch holes anymore (this operation does not introduce new $h$-points with $h>2$). After this preliminary operation, we let $\mathcal M^*$ be the resulting map and $G^* \supseteq G$ be the corresponding map graph. We remark that $G^*$ does not contain any new edge connecting two vertices of $G$. Moreover, $\mathcal M^*$ is a \map{2k}, which is still nearly well-formed. We then apply the procedure in the proof of \cref{lem:wellformed}. Namely, we first delete all $h$-points with $h>3$, which implies that the resulting map is now well-formed. We then compute a drawing $\Gamma'$ of the \skeleton of the resulting graph $G'$.  The proof of \cref{lem:wellformed} ended by showing how to reinsert the edges in $\overline{E} = G^* \setminus G'$ inside their corresponding faces of $\Gamma'$ so to create a \framed{2k} drawing $\Gamma$ of $G$. Before applying this last step, we observe that the absence of holes was used in the proof to guarantee that the size of each face of $\Gamma'$ is at most $2k$ and that $G'$ is biconnected. We show that, after a suitable augmentation of $G'$, both properties still hold. 

A face $f$ of $\Gamma'$ is \emph{large} if the size of $f$ is greater than $3$ and $f$ does not contain crossing edges in the final drawing $\Gamma$ of $G^*$ (i.e., $f$ is generated by a disk inserted to eliminate an $h$-point with $h>2$). Let $f$ be a large face; the \emph{stellation operation} of $f$ works as follows. We insert a vertex $v_f$ inside $f$ and connect it to all vertices on the boundary of $f$ by drawing the new edges inside $f$ without creating edge crossings; if a vertex of $f$ is a cut-vertex, we connect it to $v_f$ only once. The stellation operation removes $f$ and creates new faces of size strictly smaller than the size of $f$. By repeatedly applying this operation until there are no large faces we obtain a \skeleton $G''$ such that the boundary of each face is a simple cycle, which implies that $G''$ is biconnected. Also, the size of a face of $G''$ is at most $3$ if it does not contain crossing edges in $\Gamma$, and at most $2k$ otherwise. By finally reintroducing the crossing edges inside the faces of $G''$ of size (at most) $2k$, we obtain a \framed{2k} drawing of a \framed{2k} graph, which is a super graph of the input graph $G$. This proves that every \map{k} graph is \pframed{2k}. 
\end{proof}

We conclude this section by giving the following simple result, which implies that the book thickness of \framed{k} graphs (and hence of \pframed{k} graphs) is bounded by the book thickness of \map{k} graphs. 

\begin{theorem}
Every \framed{k} graph is a \map{k} graph, under the assumption that each face of the \skeleton induces a clique of size $k$.
\end{theorem}
\begin{proof}
Consider a \framed{k} drawing $\Gamma$ of a \framed{k} graph $G$. Let $\Gamma'$ be the \skeleton of $\Gamma$. As already said, we shall assume that each face of $\Gamma'$ induces a clique of size $k$ in $\Gamma$. We replace each vertex $v$ of $\Gamma$ with a sufficiently thin star-shaped nation that includes each curve representing an each edge $(u,v)$ up to the midpoint of such curve. Since $\Gamma'$ is planar, this operation transforms $\Gamma'$ into a \map{2} $\mathcal M'$. Observe that a face of size $h \le k$ in $\Gamma'$ corresponds to a hole in $\mathcal M'$. Thus, the crossing edges of $\Gamma$ can be easily reintroduced by creating an $h$-point inside each such hole.
\end{proof}

\section{Conclusions and open problems}\label{sec:open}
Our research generalizes a fundamental result by Yannakakis in the area of book embeddings. To achieve $O(k)$ pages for \pframed{k} graphs, we exploit the special structure of these graphs which allows us to model the conflicts of the crossing edges by means of a graph with bounded chromatic number (thus keeping the unavoidable relationship with $k$ low).

Even though our result only applies to a subclass of $h$-planar graphs, it provides useful insights towards a positive answer to the intriguing question of determining whether the book thickness of (general) $h$-planar graphs is bounded by a function of $h$ only. 

Another natural question that stems from our research is whether \map{k} graphs are \pframed{k}, and in particular, whether  \cref{th:map} can be improved. 

A third direction for extending our result is to drop the biconnectivity requirement of \pframed{k} graphs.

We conclude by mentioning that the time complexity of our algorithm is $O(k^2 n)$, assuming that a \framed{k} drawing of the considered graph is also provided.  It is of interest to investigate whether (partial) \framed{k} graphs can be recognized in polynomial time. 
The question remains valid even for the class of optimal $2$-planar graphs, which exhibit a quite regular structure. In relation to this question, Brandenburg~\cite{DBLP:journals/algorithmica/Brandenburg19} provided a corresponding linear-time recognition algorithm for the class of optimal $1$-planar graphs, while Da Lozzo et al.~\cite{DBLP:conf/isaac/LozzoJKR14} showed that the related question of determining whether a graph admits a planar embedding whose faces have all degree at most $k$ is polynomial-time solvable for $k \le 4$ and NP-complete for $k \ge 5$.

\paragraph{Acknowledgements.} This work began at the Dagstuhl Seminar 19092 ``Beyond-Planar Graphs: Combinatorics, Models and Algorithms'' (February 24 - March 1, 2019). We thank the organizers and the participants for useful discussions and feedback. We also thank the useful comments of the anonymous referees of this paper; in particular, we acknowledge one referee for suggesting us to explore the relationship between \framed{k} and $k$-map graphs.

\medskip\noindent Research by M. A. Bekos is partially supported by DFG grant KA812/18-1. Research by G. Da Lozzo is partially supported by MSCA-RISE project ``CONNECT'', N$^\circ$~734922, and by MIUR, grant 20174LF3T8 ``AHeAD: efficient Algorithms for HArnessing networked Data''. Research by F. Montecchiani is partially supported by MIUR, grant 20174LF3T8 ``AHeAD: efficient Algorithms for HArnessing networked Data'', and Dip. Ingegneria Univ. Perugia, grant RICBA19FM: ``Modelli, algoritmi e sistemi per la visualizzazione di grafi e reti''.

\bibliographystyle{plainurl} \bibliography{stacks}

\begin{thebibliography}{10}

\bibitem{DBLP:conf/gd/AkitayaDHL17}
Hugo~A. Akitaya, Erik~D. Demaine, Adam Hesterberg, and Quanquan~C. Liu.
\newblock Upward partitioned book embeddings.
\newblock In {\em Graph Drawing}, volume 10692 of {\em Lecture Notes in
  Computer Science}, pages 210--223. Springer, 2017.

\bibitem{DBLP:conf/gd/AlamBK13}
Md.~Jawaherul Alam, Franz~J. Brandenburg, and Stephen~G. Kobourov.
\newblock Straight-line grid drawings of 3-connected 1-planar graphs.
\newblock In Stephen~K. Wismath and Alexander Wolff, editors, {\em Graph
  Drawing}, volume 8242 of {\em LNCS}, pages 83--94. Springer, 2013.
\newblock \href {https://doi.org/10.1007/978-3-319-03841-4\_8}
  {\path{doi:10.1007/978-3-319-03841-4\_8}}.

\bibitem{DBLP:journals/corr/AlamBK15}
Md.~Jawaherul Alam, Franz~J. Brandenburg, and Stephen~G. Kobourov.
\newblock On the book thickness of 1-planar graphs.
\newblock {\em CoRR}, abs/1510.05891, 2015.
\newblock \href {http://arxiv.org/abs/http://arxiv.org/abs/1510.05891}
  {\path{arXiv:http://arxiv.org/abs/1510.05891}}.

\bibitem{DBLP:journals/algorithmica/BekosBKR17}
Michael~A. Bekos, Till Bruckdorfer, Michael Kaufmann, and Chrysanthi~N.
  Raftopoulou.
\newblock The book thickness of 1-planar graphs is constant.
\newblock {\em Algorithmica}, 79(2):444--465, 2017.
\newblock \href {https://doi.org/10.1007/s00453-016-0203-2}
  {\path{doi:10.1007/s00453-016-0203-2}}.

\bibitem{DBLP:journals/algorithmica/BekosGR16}
Michael~A. Bekos, Martin Gronemann, and Chrysanthi~N. Raftopoulou.
\newblock Two-page book embeddings of 4-planar graphs.
\newblock {\em Algorithmica}, 75(1):158--185, 2016.
\newblock \href {https://doi.org/10.1007/s00453-015-0016-8}
  {\path{doi:10.1007/s00453-015-0016-8}}.

\bibitem{DBLP:conf/compgeom/Bekos0R17}
Michael~A. Bekos, Michael Kaufmann, and Chrysanthi~N. Raftopoulou.
\newblock On optimal 2- and 3-planar graphs.
\newblock In Boris Aronov and Matthew~J. Katz, editors, {\em {SoCG}}, volume~77
  of {\em LIPIcs}, pages 16:1--16:16. Schloss Dagstuhl - Leibniz-Zentrum fuer
  Informatik, 2017.
\newblock \href {https://doi.org/10.4230/LIPIcs.SoCG.2017.16}
  {\path{doi:10.4230/LIPIcs.SoCG.2017.16}}.

\bibitem{DBLP:journals/jct/BernhartK79}
Frank Bernhart and Paul~C. Kainen.
\newblock The book thickness of a graph.
\newblock {\em J. Comb. Theory, Ser. {B}}, 27(3):320--331, 1979.
\newblock \href {https://doi.org/10.1016/0095-8956(79)90021-2}
  {\path{doi:10.1016/0095-8956(79)90021-2}}.

\bibitem{DBLP:journals/jgaa/BiedlSWW99}
Therese~C. Biedl, Thomas~C. Shermer, Sue Whitesides, and Stephen~K. Wismath.
\newblock Bounds for orthogonal {3D} graph drawing.
\newblock {\em J. Graph Algorithms Appl.}, 3(4):63--79, 1999.
\newblock \href {https://doi.org/10.7155/jgaa.00018}
  {\path{doi:10.7155/jgaa.00018}}.

\bibitem{DBLP:conf/compgeom/BinucciLGDMP19}
Carla Binucci, Giordano {Da Lozzo}, Emilio~Di Giacomo, Walter Didimo, Tamara
  Mchedlidze, and Maurizio Patrignani.
\newblock Upward book embeddings of st-graphs.
\newblock In {\em {SoCG}}, volume 129 of {\em LIPIcs}, pages 13:1--13:22.
  Schloss Dagstuhl - Leibniz-Zentrum fuer Informatik, 2019.
\newblock \href {https://doi.org/10.4230/LIPIcs.SoCG.2019.13}
  {\path{doi:10.4230/LIPIcs.SoCG.2019.13}}.

\bibitem{Bla03}
Robin~L. Blankenship.
\newblock {\em Book Embeddings of Graphs}.
\newblock PhD thesis, Louisiana State University, 2003.

\bibitem{DBLP:journals/dam/Brandenburg19}
Franz~J. Brandenburg.
\newblock Characterizing 5-map graphs by 2-fan-crossing graphs.
\newblock {\em Discret. Appl. Math.}, 268:10--20, 2019.
\newblock \href {https://doi.org/10.1016/j.dam.2019.04.012}
  {\path{doi:10.1016/j.dam.2019.04.012}}.

\bibitem{DBLP:journals/algorithmica/Brandenburg19}
Franz~J. Brandenburg.
\newblock Characterizing and recognizing 4-map graphs.
\newblock {\em Algorithmica}, 81(5):1818--1843, 2019.
\newblock \href {https://doi.org/10.1007/s00453-018-0510-x}
  {\path{doi:10.1007/s00453-018-0510-x}}.

\bibitem{DBLP:conf/stoc/BussS84}
Jonathan~F. Buss and Peter~W. Shor.
\newblock On the pagenumber of planar graphs.
\newblock In Richard~A. DeMillo, editor, {\em {ACM} Symposium on Theory of
  Computing}, pages 98--100. {ACM}, 1984.
\newblock URL: \url{http://doi.acm.org/10.1145/800057.808670}, \href
  {https://doi.org/10.1145/800057.808670} {\path{doi:10.1145/800057.808670}}.

\bibitem{DBLP:journals/jacm/ChenGP02}
Zhi{-}Zhong Chen, Michelangelo Grigni, and Christos~H. Papadimitriou.
\newblock Map graphs.
\newblock {\em J. {ACM}}, 49(2):127--138, 2002.
\newblock \href {https://doi.org/10.1145/506147.506148}
  {\path{doi:10.1145/506147.506148}}.

\bibitem{DBLP:journals/algorithmica/ChenGP06}
Zhi{-}Zhong Chen, Michelangelo Grigni, and Christos~H. Papadimitriou.
\newblock Recognizing hole-free 4-map graphs in cubic time.
\newblock {\em Algorithmica}, 45(2):227--262, 2006.
\newblock \href {https://doi.org/10.1007/s00453-005-1184-8}
  {\path{doi:10.1007/s00453-005-1184-8}}.

\bibitem{DBLP:conf/soda/ChenHK99}
Zhi{-}Zhong Chen, Xin He, and Ming{-}Yang Kao.
\newblock Nonplanar topological inference and political-map graphs.
\newblock In Robert~Endre Tarjan and Tandy~J. Warnow, editors, {\em {SODA}},
  pages 195--204. {ACM/SIAM}, 1999.
\newblock URL: \url{http://dl.acm.org/citation.cfm?id=314500.314558}.

\bibitem{CLR87}
Fan R.~K. Chung, Frank~T. Leighton, and Arnold~L. Rosenberg.
\newblock Embedding graphs in books: A layout problem with applications to
  {VLSI} design.
\newblock {\em SIAM Journal on Algebraic and Discrete Methods}, 8(1):33--58,
  1987.

\bibitem{DBLP:journals/mp/CornuejolsNP83}
G{\'{e}}rard Cornu{\'{e}}jols, Denis Naddef, and William~R. Pulleyblank.
\newblock Halin graphs and the travelling salesman problem.
\newblock {\em Math. Program.}, 26(3):287--294, 1983.
\newblock \href {https://doi.org/10.1007/BF02591867}
  {\path{doi:10.1007/BF02591867}}.

\bibitem{DBLP:conf/isaac/LozzoJKR14}
Giordano {Da Lozzo}, V{\'{\i}}t Jel{\'{\i}}nek, Jan Kratochv{\'{\i}}l, and
  Ignaz Rutter.
\newblock Planar embeddings with small and uniform faces.
\newblock In Hee{-}Kap Ahn and Chan{-}Su Shin, editors, {\em {ISAAC}}, volume
  8889 of {\em Lecture Notes in Computer Science}, pages 633--645. Springer,
  2014.
\newblock \href {https://doi.org/10.1007/978-3-319-13075-0\_50}
  {\path{doi:10.1007/978-3-319-13075-0\_50}}.

\bibitem{DBLP:journals/dcg/FraysseixMP95}
Hubert de~Fraysseix, Patrice~Ossona de~Mendez, and J{\'{a}}nos Pach.
\newblock A left-first search algorithm for planar graphs.
\newblock {\em Discrete {\&} Computational Geometry}, 13:459--468, 1995.
\newblock \href {https://doi.org/10.1007/BF02574056}
  {\path{doi:10.1007/BF02574056}}.

\bibitem{DBLP:journals/csur/DidimoLM19}
Walter Didimo, Giuseppe Liotta, and Fabrizio Montecchiani.
\newblock A survey on graph drawing beyond planarity.
\newblock {\em {ACM} Comput. Surv.}, 52(1):4:1--4:37, 2019.
\newblock \href {https://doi.org/10.1145/3301281} {\path{doi:10.1145/3301281}}.

\bibitem{DBLP:journals/jgaa/DujmovicF18}
Vida Dujmovi{\'{c}} and Fabrizio Frati.
\newblock Stack and queue layouts via layered separators.
\newblock {\em J. Graph Algorithms Appl.}, 22(1):89--99, 2018.
\newblock \href {https://doi.org/10.7155/jgaa.00454}
  {\path{doi:10.7155/jgaa.00454}}.

\bibitem{DBLP:journals/jct/DujmovicMW17}
Vida Dujmovic, Pat Morin, and David~R. Wood.
\newblock Layered separators in minor-closed graph classes with applications.
\newblock {\em J. Comb. Theory, Ser. {B}}, 127:111--147, 2017.
\newblock \href {https://doi.org/10.1016/j.jctb.2017.05.006}
  {\path{doi:10.1016/j.jctb.2017.05.006}}.

\bibitem{DBLP:journals/dcg/DujmovicW07}
Vida Dujmovi{\'{c}} and David~R. Wood.
\newblock Graph treewidth and geometric thickness parameters.
\newblock {\em Discrete {\&} Computational Geometry}, 37(4):641--670, 2007.
\newblock \href {https://doi.org/10.1007/s00454-007-1318-7}
  {\path{doi:10.1007/s00454-007-1318-7}}.

\bibitem{Ewald1973}
G{\"u}nter Ewald.
\newblock Hamiltonian circuits in simplicial complexes.
\newblock {\em Geometriae Dedicata}, 2(1):115--125, 1973.
\newblock \href {https://doi.org/10.1007/BF00149287}
  {\path{doi:10.1007/BF00149287}}.

\bibitem{DBLP:journals/dam/GanleyH01}
Joseph~L. Ganley and Lenwood~S. Heath.
\newblock The pagenumber of $k$-trees is {$O(k)$}.
\newblock {\em Discrete Applied Mathematics}, 109(3):215--221, 2001.
\newblock \href {https://doi.org/10.1016/S0166-218X(00)00178-5}
  {\path{doi:10.1016/S0166-218X(00)00178-5}}.

\bibitem{DBLP:journals/corr/GuanY2018}
Xiaxia Guan and Weihua Yang.
\newblock Embedding 5-planar graphs in three pages.
\newblock {\em CoRR}, 1801.07097, 2018.

\bibitem{DBLP:conf/focs/Heath84}
Lenwood~S. Heath.
\newblock Embedding planar graphs in seven pages.
\newblock In {\em {FOCS}}, pages 74--83. {IEEE} Computer Society, 1984.
\newblock \href {https://doi.org/10.1109/SFCS.1984.715903}
  {\path{doi:10.1109/SFCS.1984.715903}}.

\bibitem{DBLP:conf/esa/0001K19}
Michael Hoffmann and Boris Klemz.
\newblock Triconnected planar graphs of maximum degree five are subhamiltonian.
\newblock In Michael~A. Bender, Ola Svensson, and Grzegorz Herman, editors,
  {\em {ESA}}, volume 144 of {\em LIPIcs}, pages 58:1--58:14. Schloss Dagstuhl
  - Leibniz-Zentrum f{\"{u}}r Informatik, 2019.
\newblock \href {https://doi.org/10.4230/LIPIcs.ESA.2019.58}
  {\path{doi:10.4230/LIPIcs.ESA.2019.58}}.

\bibitem{Istrail1988a}
Sorin Istrail.
\newblock An algorithm for embedding planar graphs in six pages.
\newblock {\em Iasi University Annals, Mathematics-Computer Science},
  34(4):329--341, 1988.

\bibitem{DBLP:conf/focs/Jacobson89}
Guy Jacobson.
\newblock Space-efficient static trees and graphs.
\newblock In {\em Symposium on Foundations of Computer Science}, pages
  549--554. {IEEE} Computer Society, 1989.
\newblock \href {https://doi.org/10.1109/SFCS.1989.63533}
  {\path{doi:10.1109/SFCS.1989.63533}}.

\bibitem{DBLP:journals/appml/KainenO07}
Paul~C. Kainen and Shannon Overbay.
\newblock Extension of a theorem of whitney.
\newblock {\em Appl. Math. Lett.}, 20(7):835--837, 2007.
\newblock \href {https://doi.org/10.1016/j.aml.2006.08.019}
  {\path{doi:10.1016/j.aml.2006.08.019}}.

\bibitem{DBLP:journals/csr/KobourovLM17}
Stephen~G. Kobourov, Giuseppe Liotta, and Fabrizio Montecchiani.
\newblock An annotated bibliography on 1-planarity.
\newblock {\em Computer Science Review}, 25:49--67, 2017.
\newblock \href {https://doi.org/10.1016/j.cosrev.2017.06.002}
  {\path{doi:10.1016/j.cosrev.2017.06.002}}.

\bibitem{DBLP:journals/jal/Malitz94a}
Seth~M. Malitz.
\newblock Genus $g$ graphs have pagenumber {$O(\sqrt{q})$}.
\newblock {\em J. Algorithms}, 17(1):85--109, 1994.
\newblock \href {https://doi.org/10.1006/jagm.1994.1028}
  {\path{doi:10.1006/jagm.1994.1028}}.

\bibitem{DBLP:journals/jal/Malitz94}
Seth~M. Malitz.
\newblock Graphs with {E} edges have pagenumber {$O(\sqrt{E})$}.
\newblock {\em J. Algorithms}, 17(1):71--84, 1994.
\newblock \href {https://doi.org/10.1006/jagm.1994.1027}
  {\path{doi:10.1006/jagm.1994.1027}}.

\bibitem{DBLP:journals/disopt/MnichRS18}
Matthias Mnich, Ignaz Rutter, and Jens~M. Schmidt.
\newblock Linear-time recognition of map graphs with outerplanar witness.
\newblock {\em Discret. Optim.}, 28:63--77, 2018.
\newblock \href {https://doi.org/10.1016/j.disopt.2017.12.002}
  {\path{doi:10.1016/j.disopt.2017.12.002}}.

\bibitem{DBLP:journals/siamcomp/MunroR01}
J.~Ian Munro and Venkatesh Raman.
\newblock Succinct representation of balanced parentheses and static trees.
\newblock {\em {SIAM} J. Comput.}, 31(3):762--776, 2001.
\newblock \href {https://doi.org/10.1137/S0097539799364092}
  {\path{doi:10.1137/S0097539799364092}}.

\bibitem{DBLP:books/daglib/0030491}
Jaroslav Nesetril and Patrice~Ossona de~Mendez.
\newblock {\em Sparsity - Graphs, Structures, and Algorithms}, volume~28 of
  {\em Algorithms and combinatorics}.
\newblock Springer, 2012.
\newblock \href {https://doi.org/10.1007/978-3-642-27875-4}
  {\path{doi:10.1007/978-3-642-27875-4}}.

\bibitem{NC08}
Takao Nishizeki and Norishige Chiba.
\newblock {\em Planar Graphs: Theory and Algorithms}, chapter 10. {H}amiltonian
  {C}ycles, pages 171--184.
\newblock Dover Books on Mathematics. Courier Dover Publications, 2008.

\bibitem{Oll73}
T.~Ollmann.
\newblock On the book thicknesses of various graphs.
\newblock In F.~Hoffman, R.B. Levow, and R.S.D. Thomas, editors, {\em
  Southeastern Conference on Combinatorics, Graph Theory and Computing}, volume
  VIII of {\em Congressus Numerantium}, page 459, 1973.

\bibitem{DBLP:conf/stoc/Pratt73}
Vaughan~R. Pratt.
\newblock Computing permutations with double-ended queues, parallel stacks and
  parallel queues.
\newblock In Alfred~V. Aho, Allan Borodin, Robert~L. Constable, Robert~W.
  Floyd, Michael~A. Harrison, Richard~M. Karp, and H.~Raymond Strong, editors,
  {\em {ACM} Symposium on Theory of Computing}, pages 268--277. {ACM}, 1973.
\newblock \href {https://doi.org/10.1145/800125.804058}
  {\path{doi:10.1145/800125.804058}}.

\bibitem{DBLP:conf/cocoon/RengarajanM95}
S.~Rengarajan and C.~E.~Veni Madhavan.
\newblock Stack and queue number of 2-trees.
\newblock In Ding{-}Zhu Du and Ming Li, editors, {\em {COCOON}}, volume 959 of
  {\em LNCS}, pages 203--212. Springer, 1995.
\newblock \href {https://doi.org/10.1007/BFb0030834}
  {\path{doi:10.1007/BFb0030834}}.

\bibitem{R65}
Gerhard Ringel.
\newblock Ein {S}echsfarbenproblem auf der kugel.
\newblock {\em Abhandlungen aus dem Mathematischen Seminar der Universitaet
  Hamburg}, 29(1--2):107--117, 1965.

\bibitem{DBLP:journals/tc/Rosenberg83}
Arnold~L. Rosenberg.
\newblock The diogenes approach to testable fault-tolerant arrays of
  processors.
\newblock {\em {IEEE} Trans. Computers}, 32(10):902--910, 1983.
\newblock \href {https://doi.org/10.1109/TC.1983.1676134}
  {\path{doi:10.1109/TC.1983.1676134}}.

\bibitem{DBLP:journals/jacm/Tarjan72}
Robert~E. Tarjan.
\newblock Sorting using networks of queues and stacks.
\newblock {\em J. {ACM}}, 19(2):341--346, 1972.
\newblock \href {https://doi.org/10.1145/321694.321704}
  {\path{doi:10.1145/321694.321704}}.

\bibitem{DBLP:conf/focs/Thorup98}
Mikkel Thorup.
\newblock Map graphs in polynomial time.
\newblock In {\em {FOCS}}, pages 396--405. {IEEE} Computer Society, 1998.
\newblock \href {https://doi.org/10.1109/SFCS.1998.743490}
  {\path{doi:10.1109/SFCS.1998.743490}}.

\bibitem{Wig82}
Avi Wigderson.
\newblock The complexity of the {H}amiltonian circuit problem for maximal
  planar graphs.
\newblock Technical Report TR-298, EECS Department, Princeton University, 1982.
\newblock \href {http://arxiv.org/abs/https://www.math.ias.edu/avi/node/820}
  {\path{arXiv:https://www.math.ias.edu/avi/node/820}}.

\bibitem{DBLP:journals/jal/Wood02}
David~R. Wood.
\newblock Degree constrained book embeddings.
\newblock {\em J. Algorithms}, 45(2):144--154, 2002.
\newblock \href {https://doi.org/10.1016/S0196-6774(02)00249-3}
  {\path{doi:10.1016/S0196-6774(02)00249-3}}.

\bibitem{DBLP:conf/stoc/Yannakakis86}
Mihalis Yannakakis.
\newblock Four pages are necessary and sufficient for planar graphs (extended
  abstract).
\newblock In Juris Hartmanis, editor, {\em {ACM} Symposium on Theory of
  Computing}, pages 104--108. {ACM}, 1986.
\newblock \href {https://doi.org/10.1145/12130.12141}
  {\path{doi:10.1145/12130.12141}}.

\bibitem{DBLP:journals/jcss/Yannakakis89}
Mihalis Yannakakis.
\newblock Embedding planar graphs in four pages.
\newblock {\em J. Comput. Syst. Sci.}, 38(1):36--67, 1989.
\newblock \href {https://doi.org/10.1016/0022-0000(89)90032-9}
  {\path{doi:10.1016/0022-0000(89)90032-9}}.

\end{thebibliography}

\end{document}